\documentclass[pdflatex,sn-basic]{sn-jnl}

\usepackage{graphicx}%
\usepackage{multirow}%
\usepackage{amsmath,amssymb,amsfonts}%
\usepackage{amsthm}%
\usepackage{mathrsfs}%
\usepackage[title]{appendix}%
\usepackage{xcolor}%
\usepackage{textcomp}%
\usepackage{manyfoot}%
\usepackage{booktabs}%
\usepackage{algorithm}%
\usepackage{algorithmicx}%
\usepackage{algpseudocode}%
\usepackage{listings}%
\usepackage[mathcal]{euscript}

\usepackage{enumitem}

\def\set#1{\{#1\}}              
\def\tuple#1{\langle#1\rangle}  

\algnewcommand{\IfThen}[2]{
  \State \algorithmicif\ #1\ \algorithmicthen\ #2}

\algtext*{EndFor}
\algtext*{EndWhile}
\algtext*{EndIf}
\algtext*{EndIIf}

\theoremstyle{thmstyleone}%
\newtheorem{theorem}{Theorem}
\newtheorem{proposition}[theorem]{Proposition}%

\theoremstyle{thmstyletwo}%

\theoremstyle{thmstylethree}%
\newtheorem{definition}{Definition}%

\newtheorem{lemma}[theorem]{Lemma}

\begin{document}

\title[From Necklaces to Coalitions: Fair and Self-Interested Distribution of Coalition Value Calculations]{From Necklaces to Coalitions: Fair and Self-Interested Distribution of Coalition Value Calculations}


\author*[1]{\fnm{Terry R.} \sur{Payne}}\email{T.R.Payne@liverpool.ac.uk}
\author[2]{\fnm{Luke} \sur{Riley}}\email{luke.riley@quant.network}
\equalcont{These authors contributed equally to this work.}

\affil*[1]{\orgdiv{School of Computer Science and Informatics}, \orgname{University of Liverpool}, \orgaddress{\city{Liverpool}, \country{UK}}}

\affil[2]{\orgname{Quant}, \orgaddress{\city{London}, \country{UK}}}

\abstract{
A key challenge in distributed coalition formation within characteristic function games is determining how to allocate the
calculation of coalition values across a set of agents.
The number of possible coalitions grows exponentially with the number of agents, and existing distributed approaches may produce uneven or redundant allocations, or assign coalitions to agents that are not themselves members.
In this article, we present the \emph{Necklace-based Distributed Coalition Algorithm} (N-DCA), a communication-free algorithm in which each agent independently determines its own coalition value calculation allocation using only its identifier and the total number of agents.
The approach builds on the notion of \emph{Increment Arrays} (IAs), for which we develop a complete mathematical framework: equivalence classes under circular shifts, periodic IAs, and a rotated designation scheme with formal load-balance guarantees (tight bounds).
We establish a bijection between canonical representative IAs and two-colour combinatorial necklaces, enabling the use of efficient necklace generation algorithms to enumerate allocations in constant amortised time.
N-DCA is, to the best of our knowledge, the only distributed coalition value calculation algorithm for unrestricted characteristic function games to provably satisfy five desirable properties: no inter-agent communication, equitable allocation, no redundancy, balanced load, and self-interest.  
An empirical evaluation against DCVC \citep{Rahwan2007} demonstrates that, although DCVC is faster by a constant factor, this difference becomes negligible under realistic characteristic-function evaluation costs, while N-DCA offers advantages in working memory, scalability, and the self-interest guarantee.
}


\maketitle

\section{Introduction} \label{sec:intro}

One of the many advantages of exploiting the Multi-Agent Systems paradigm for managing distributed real-world tasks \citep{WooldridgeKER95,Belecheanu2006} is the ability for agents to form \emph{coalitions} within which they collaborate, thus exploiting the synergy of their actions when tackling joint tasks. 
The formation of coalitions occurs through the evaluation of the different utility values (using some \emph{characteristic function}) for candidate coalitions which are \emph{satisfiable}, and then determining a mutually acceptable division of utility values forming a stable distribution \citep{sandholm99}. 
Such coalitions are typically transient, goal-directed, and are characterised by a flat (as opposed to hierarchical) organisation, where coordination typically occurs between members of a single coalition, as opposed to between members of other coalitions \citep{RahwanAIJ2015}. 
Coalition formation is a well-studied research area and has a wide range of potential applications including: electronic auctions/market places (for example, to take advantage of bulk buying); communication networks; the smart grid; grid computing; distributed vehicle routing; distributed sensor networks; multi-agent planning; and computational trust \citep{DBLP:conf/aips/ContrerasKY98,Dang:2006:OCF:1597538.1597640,Li:2010:CCF:1749629.1750129,woolnewbook}.

\citet{sandholm99} framed the problem of identifying 
potential coalition structures as a three-stage process: 
(i) calculating the utility value of each possible coalition; 
(ii) identifying a feasible set of coalitions; and 
(iii) dividing the utility among agents in a stable manner 
(i.e.\ no agent can object to its assigned payoff).
The generation of a coalition structure itself can be \emph{exogenous} (i.e.\ determined externally, for example, by an oracle or system designer) or \emph{endogenous} (i.e.\ determined by the agents themselves).
By restricting the types of coalition structures to those used in \emph{Characteristic Function Games (CFGs)} \citep{RahwanAIJ2015}, the utility value ascribed to each coalition (i.e.\ the characteristic function $\nu$) depends only on the agents within that coalition, and thus only needs to be calculated for each of the $2^{n}-1$ possible coalitions~\citep{sandholm99,Rahwan2007}.

Several approaches have adopted an endogenous approach by distributing the calculations of each of the coalition values across each of the agents, with the goal of reducing the per-agent computation cost and possibly reducing the overall computation time~\citep{shehory95,shehory96,shehory98,dang04,Rahwan2007,Michalak2010,Vinyals2012,Voice12,RahwanAIJ2015,RileyAAAI15}.
Although this eliminates the need for a trusted central authority responsible for determining the value calculations, there is still the significant risk of \emph{redundancy} in the calculation of each coalition~\citep{shehory95,shehory96,shehory98}.
This redundancy is unnecessary, as agents share the same value for each coalition within these games; furthermore, the complexity of calculating an individual coalition's value in a characteristic function game can vary, and is potentially exponential \citep{sandholm97}. Even if agents only calculate the values of those coalitions in which they participate, a significant overlap of calculations can occur (i.e.\ $2^{n-1}$), with this redundancy converging to 100\% (as the fraction of distinct coalitions per agent, $\frac{2^n-1}{n \cdot 2^{(n-1)}}$, tends to 0 as $n \rightarrow \infty$).

Furthermore, some approaches distribute the coalition value calculations unevenly across the agents \citep{Vinyals2012}, whereas others allocate coalitions to agents that are not themselves members; an arrangement that is problematic not only in adversarial environments where misreporting may occur, but in any setting where agents are self-interested, since a rational agent has no incentive to compute accurately the value of a coalition it will never join \citep{Rahwan2007,Voice12}.
Therefore, the following desirable properties have been identified which are not always present in the majority of these systems \citep{Rahwan2007,RileyAAAI15}:

\begin{enumerate}
    \item eliminating the need for \emph{communication} and explicit coordination between agents when determining coalition shares;
    \item ensuring that each agent has an approximately equal number of coalitions to evaluate, thus ensuring that the division is \emph{equitable};
    \item avoiding \emph{redundancy} in coalition value calculations, i.e.\ ensuring that no calculation is allocated to more than one agent;
    \item maintaining a \emph{balanced load} across agents - although two agents may have similar numbers of coalitions to examine, the size of each coalition itself can affect the computational load, and thus it becomes undesirable for one agent to evaluate mainly small coalition sizes, with another evaluating larger ones;
    \item maintaining \emph{self-interest}, i.e.\ ensuring that an agent evaluating a coalition is itself a member of that coalition. An agent assigned to calculate $\nu(C)$ for a coalition $C$ in which it does not participate has no direct stake in whether $C$ forms, and may be positively incentivised to misreport $\nu(C)$ strategically in order to steer the subsequent structure-search or payoff-division stages toward a more personally favourable outcome.
\end{enumerate}

Properties 2 and 4 together capture the notion of \emph{fairness}: each agent receives not only an approximately equal number of coalitions, but also an approximately equal computational burden across coalition sizes, and we use the term \emph{fair} to refer to algorithms that satisfy both of these properties.
Properties 1 and 3 concern \emph{efficiency}: eliminating communication overhead and redundant computation.
Property 5 reflects a distinct concern: \emph{incentive alignment}.
In a decentralised system without a trusted enforcement mechanism, agents cannot be compelled to compute assigned values correctly; compliance must be individually rational.
The self-interest property ensures that every agent's assigned set of coalitions is one it has a genuine stake in evaluating correctly, since its own potential payoff depends on those very values.
The broader implications, including the relationship to mechanism design and the residual risk of misreporting within assigned coalitions, are examined in Section~\ref{sec:future-directions}.
Throughout, we assume an unrestricted characteristic function game in which all agents know the total number of agents $n$; games with externalities and open or dynamic environments are outside the current scope (see Section~\ref{sec:future-directions} for discussion of limitations).

The notion of \emph{Increment Arrays} (\emph{IAs}) was introduced in earlier studies by~\citet{RileyAAMAS14,RileyAAAI15}, as a means of addressing all of these issues when generating \emph{coalition value calculation allocations} (\emph{CV}s).
When combined with an agent identifier, the \emph{IA}s identify a set of unique coalitions that can be allocated to those agents that appear in them, as well as balancing the computational load approximately evenly (with respect to the number of coalition value allocation calculations and number of operations) across
the agents.
The resulting CVs can also be combined with other approaches to complete the coalition formation process, such as distributed structure-search algorithms (e.g.\ D-IP \citep{Michalak2010}). Alternatively, if the agents need to find a core/$\epsilon$-core stable solution \citep{GameTheory}, the resulting coalition value allocation could be used as input to algorithms such as those proposed by \citet{WuSIAM1977} or \citet{Cesco98}.

The way in which \emph{integer partitions} are used to construct the Increment Arrays is similar to the different arrangements of a string of ``beads'' that exist within a \emph{two-colour combinatorial necklace}.
Combinatorial necklaces \citep{RUSKEY1992414} are strings of characters such that the string is the smallest lexicographic representative of all the cyclic shifts of that string. 
They are often described as a set of beads of different colours with different arrangements (see Figure~\ref{fig:necklaceExample}), or more formally, given a set of $n$ beads that can each be one of $k$ colours, the number of different arrangements of a string of beads that can be constructed within a circular loop can be determined \citep{CATTELL2000267,FREDRICKSEN1986181,FREDRICKSEN1978207,RUSKEY1992414,Sawada:1999:EAG:314500.314910}. 
More precisely, it is assumed that two or more arrangements can be considered to be identical if they differ only by a rotation inside the loop; i.e.\ we can say that a \emph{necklace of n beads in $k$ colours} is an equivalence class of $k$-ary $n$-tuples under rotation. 

\begin{figure}[t]
    \centering
    \includegraphics[width=\linewidth]{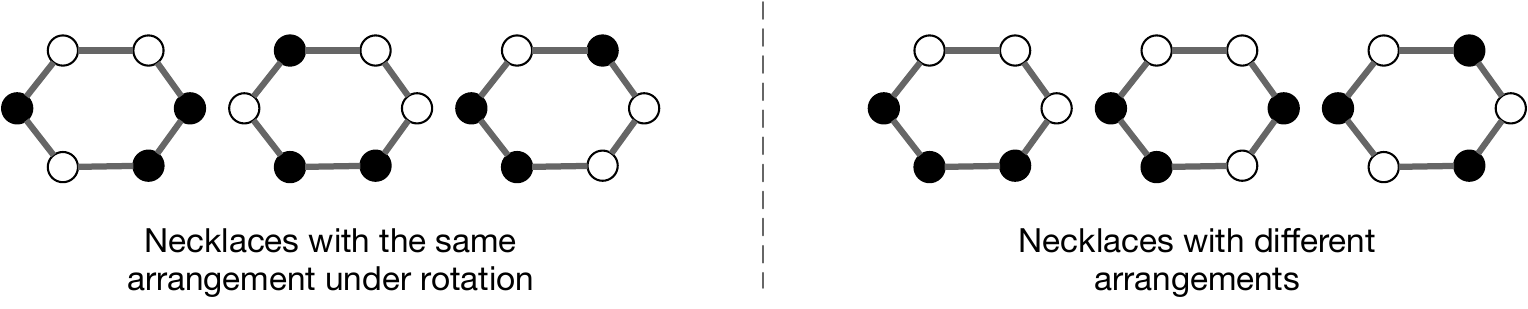}
    \caption{Examples of two-colour necklaces with six beads (i.e.\ $n=6, k=2$), illustrated with the same (left), and different (right) arrangements. In each arrangement, there are three white and three black beads.  The necklaces on the left all fall within the same equivalence class as they are the same under rotation. In those on the right, the ordering of the beads in each of the three necklaces is different, and thus they belong to different equivalence classes.}
	\label{fig:necklaceExample}
\end{figure}

Such necklaces have been used in a variety of settings, including astronomy \citep{ARNAS2017291}, music theory \citep{Rappaport2007,7839578,TOUSSAINT20102}, the discovery of wireless transceivers in MANETS (Mobile Ad Hoc Networks) \citep{8292573} and within data compression techniques \citep{6033796}. Although the investigation of efficient algorithms for generating necklaces is still an active area of research (for example, the work by \citet{8005509}), efficient algorithms exist that can generate two-colour necklaces in \emph{constant amortised time} (CAT) \citep{RUSKEY1992414}, i.e.\ where the total time is $O(N_k(n))$.

In this article, we re-state the notion of IAs (originally introduced in~\citet{RileyAAMAS14,RileyAAAI15}), together with the characterisation of equivalence classes that eliminate redundancy through the identification of a minimal canonical subset of IAs.
An approach for generating CVs is formally defined, and 
two \emph{designation schemes} (i.e.\ mechanisms for assigning
responsibility for periodic IAs to specific agents) are compared,
each yielding coalition allocations for each agent with formal load-balance
guarantees: (i) a \emph{per-size offset} variant that guarantees a maximum per-size imbalance of at most one coalition between any two agents (Theorem~9); and (ii) a \emph{global offset} variant that guarantees an aggregate imbalance of at most one (Theorem~10).
The resulting \emph{Necklace-based Distributed Coalition Algorithm (N-DCA)} is presented, and we show that N-DCA is, to the best of our knowledge, the only distributed coalition value calculation algorithm for unrestricted characteristic function games that provably satisfies all five of the above properties \emph{simultaneously} (Table~\ref{tab:comp}).  In particular, the self-interest property is not guaranteed by the current state-of-the-art DCVC approach~\citep{Rahwan2007}.

N-DCA is evaluated empirically against DCVC for $2 \le n \le 25$.  Although DCVC is faster by a constant factor (from approximately $2{\times}$ for small $n$ to approximately $5.8{\times}$ for larger values, stabilising beyond $n \approx 20$), this difference is due to N-DCA's per-necklace $\mathcal{O}(n)$ increment-array construction, and becomes negligible once realistic characteristic-function evaluation costs are included.
N-DCA's working memory is approximately one third of that of DCVC due to the fact that its implementation is purely iterative (and thus requires no recursion), and its necklace-based encoding avoids the binomial-coefficient index arithmetic that limits DCVC to $n < 68$ under standard 64-bit integer processor architectures.  The load-balance guarantees of both designation schemes are then validated empirically, and the theoretical bounds are shown to be tight.

This article significantly extends the original work by \citet{RileyAAMAS14,RileyAAAI15}, which introduced the notion of Increment Arrays and a preliminary Distributed Coalition Generation (DCG) algorithm, and the preliminary necklace-based approach in \citet{PayneIEEEWIC24}.
The principal contributions beyond those earlier studies include:
\begin{enumerate}
\item A complete mathematical framework for Increment Arrays, including formal definitions of equivalence classes under circular shifts, periodicity and stride, and the coalition generation function;
\item The necklace--IA bijection, which replaces DCG's ad hoc partition-enumeration-and-filtering construction with a principled mapping to a well-studied combinatorial concept, yielding constant amortised time enumeration;
\item The formulation of the \emph{designation problem} for periodic IAs, together with two offset variants and formal load-balance guarantees: a per-size bound and an aggregate bound;
\item Full formal proofs for all of the main results, with the detailed proofs of Lemma~\ref{lemma:same-sets} and Theorem~\ref{thm:main} provided in the Appendix;
\item A comprehensive empirical evaluation against DCVC \citep{Rahwan2007} over $2 \le n \le 25$, including timing, memory, component profiling, load-balance validation and amortised-cost analysis.
\end{enumerate}

The article is organised as follows:
Section~\ref{sec:related} reviews the main decentralised approaches to distributing coalition value calculations across agents.
Section~\ref{sec:preliminaries} provides a brief background on Characteristic Function Games (CFGs) and Combinatorial Necklaces.  
In Section~\ref{sec:ippe}, we develop the mathematical framework: \emph{Increment Arrays} (IAs) and their equivalence classes, periodic IAs, and the designation problem for allocating coalitions to agents.  
The algorithms for implementing this framework (N-DCA) are presented in Section~\ref{sec:algorithm}, where the constituent procedures for necklace generation, run-length encoding, and coalition construction are described.
Section~\ref{sec:evaluation} then provides an empirical evaluation of the approach, including a comparative analysis against DCVC.
Section~\ref{sec:discussion} discusses the results: comparing N-DCA against existing algorithms with respect to the five desirable properties, examining the source and feasibility of reducing the observed constant-factor overhead, and identifying broader implications and directions for future work, before concluding in Section~\ref{sec:conclusions}.

\section{Related Work} \label{sec:related}

Coalition formation can be studied in formal \emph{coalitional game} models. 
One such class of coalitional games is the \emph{Characteristic Function Game (CFG)} \citep{vonNeumann1944} model, where the value of the coalition is dependent solely on the agents within the coalition itself.  
These games are typically defined by a tuple $\langle Ag, \nu \rangle$, where $Ag$ is the set of $n$ agents, and $\nu$ is the  \emph{characteristic function} that maps the coalition members $C \subseteq Ag$ to some value, such that $\nu:2^{n} \longrightarrow \mathbb{R}$.
Thus, for any given task, there may be many coalitions within a CFG
that an agent can join, where each coalition
comprises agents with potentially different competences and with different
payoff expectations, and where the coalitions themselves can have different sizes and different characteristic function values.
A second class of coalitional game is the
\emph{Partition Function Game (PFG)} \citep{RahwanMWJ12},
where the value of a coalition not only depends on its
constituent agents, but also on the composition of other coalitions that
are formed, as the formation of one coalition may not be independent of
others, and can positively or negatively affect the value of other potential coalitions.  
In this article, we focus primarily on CFGs.

For both types of game, coalition formation can be divided into a three-stage process \citep{sandholm99}:
\begin{enumerate}
\item {\bf Coalition Structure Generation:}
This involves calculating the value of each of the possible coalitions, such that for a given set of agents $Ag = \{1,2,\ldots,n\}$, the number of possible coalitions in a CFG is $2^n - 1$, corresponding to the number of subsets of $n$ 
(not counting the empty set $\varnothing$).
For example, given three agents $Ag = \{1,2,3\}$, a total of $2^3-1=7$ possible coalitions exist: 
$\{1,2,3\}$, $\{1,2\}$, $\{1,3\}$, $\{2, 3\}$, $\{1\}$, $\{2\}$, $\{3\}$.  
The partitioning of all of the agents into coalitions results in the following 5 possible \emph{coalition structures}:
$\{\{1,2,3\}\}$, $\{\{1\},\{2,3\}\}$, $\{\{2\},\{1,3\}\}$, $\{\{3\},\{1,2\}\}$,
$\{\{1\},\{2\},\{3\}\}$, typically 
represented using a \emph{Coalition Structure Graph}.

\item {\bf Solving the Optimisation Problem:}
This stage addresses the problem of how the agents within a coalition collaborate to achieve their shared or joint goal.  The objective here is for each agent to work in such a way as to maximise the joint utility obtainable from achieving their goal.
This could, for example, be monetary, resulting in a joint payoff whose value corresponds to the characteristic function.

\item {\bf Payoff Distribution:}
This determines how the overall payoff is divided between the participating agents (if the given value of each coalition has the capability to be transferred between the agents).
This division should take into consideration a number of criteria, including \emph{fairness} (i.e.\ ensuring that the payoff received by each agent is commensurate with their contribution to the task) and \emph{maintaining stability} (this is pertinent when agents are selfish, and thus
may join only those coalitions that can maximise their payoff).
\end{enumerate}

Furthermore, coalition structures should preferably be either:
(i)  those that maximise the social welfare of the agents in the coalition, and thus are optimal; or 
(ii) those that form a stable coalition structure (where no agent would be incentivised to defect to another coalition in order to acquire a better payoff).
Sometimes, a coalition structure may be found that satisfies both (i) and (ii), yet this is not always the case \citep{BranLar,Travis}. 

A number of algorithms and approaches have been proposed for the generation of coalition structures (i.e.\ the first stage of  coalition formation). 
These approaches typically focus on efficiently addressing the exponentially complex problem of finding an optimal (or near-optimal) coalition structure that maximises social welfare. 
Typically, these approaches can be divided into: 
(i) \emph{exact algorithms} \citep{Rahwan009}, which guarantee the identification of an optimal solution (given sufficient time), typically using techniques such as dynamic programming \citep{rahwan08,Rahwan009} or branch-and-bound and anytime search with provable bounds \citep{sandholm99,dang04}; and 
(ii) \emph{non-exact algorithms}, which sacrifice optimality guarantees in order to rapidly produce high-quality solutions, often using heuristic or stochastic methods \citep{shehory98}. 
Although most work in this area has focused on CFGs, there have also been studies on \emph{Partition Function Games}, which explicitly consider externalities between coalitions \citep{RahwanMWJ12}.

This article does not focus on the general problem of \emph{coalition structure generation} \citep{RahwanAIJ2015}, but rather on the specific challenge of distributing coalition value calculations \emph{fairly} across the agents whilst maintaining \emph{self-interest}, but without the need for inter-agent coordination or duplicate allocation of coalition value calculations.
Early approaches assumed that each of the agents would calculate the  value of all coalitions that they were members of \citep{Shehory99,Blankenburg2005}, resulting in a  significant overlap of coalitions in the coalition value calculation sets and thus unnecessary computation costs for the multi-agent system (generating a total of $n2^{(n-1)}$ calculations for $n$ agents).


\citet{shehory95,shehory96,shehory98} investigated a method for distributing coalition value calculations (referred to as SK in this article), where agents contacted each other and negotiated over which coalition values to calculate.
As noted by \citet{Rahwan2007}, their approach had a number of weaknesses:
(i) \emph{high communication overhead} - the number of messages exchanged between agents was large (in some cases exponentially
large) when negotiating;
(ii) \emph{redundancy} - although the method  guaranteed that every coalition value calculation would be allocated, no guarantees were provided that they were allocated \emph{once and only once}, thus there was the risk of redundant calculations; and
(iii) \emph{fairness} - there were no guarantees that the agents' coalition value calculation sets  were approximately equal.
Furthermore, the SK algorithm had a memory requirement that grows exponentially with the number of agents.


A simple method for creating disjoint coalition value calculation sets was later investigated by \citet{Vinyals2012} (referred to here as VBFR), by exploiting an ordering amongst agents to partition the set of feasible coalitions into \emph{leading sets}.  
Specifically, the VBFR algorithm enumerates feasible coalitions in a graph, where a coalition is feasible if and only if it forms a connected subgraph.
This naturally enables the distribution of coalition value calculations across agents.
The algorithm casts the problem of generating all possible coalitions on a graph as a problem of enumerating all possible subgraphs.

VBFR uses an ordering amongst the agents (induced by a pseudotree) to partition the set of feasible coalitions into disjoint ``leading sets'' $M_i$, where each coalition in $M_i$ has agent $i$ as the smallest-index member (i.e.\ $\forall C \in M_i,\; \min(C) = i$).
These sets form a partition of the coalition space and can be used to distribute coalition value calculations across agents, with agent $i$ responsible for evaluating the coalitions in $M_i$. 
The resulting distribution of coalition value calculations is, however, highly unbalanced, as agents with low value IDs have disproportionately more calculations than those with higher IDs.
For example, for a fully connected graph, agent $1$ would always be assigned $2^{n-1}$ coalitions whereas agent $n$ is only assigned the coalition $\{n\}$.
This is illustrated in Table~\ref{tab:example2}, which lists the coalition value calculation sets for agents $2,3$ and $4$ (for $n=6$ agents generating coalitions of size $s=3$); note that agent $2$ has an allocation of $6$ different coalitions, whereas agent $4$ has a single allocation.

The approach guarantees that all coalitions distributed to agent $i$ include itself, thereby satisfying the \emph{self-interest} property. 
This property is particularly relevant in the smart grid domain for which the approach was originally developed, where agents represent \emph{``virtual electricity consumers''} that evaluate the benefits of forming coalitions for collective purchasing.
Thus, if an agent $i$ were assigned a coalition $C$ with $i \notin C$ (as in \citet{Rahwan2005,Rahwan2007}), a self-interested agent would have no incentive to evaluate $C$ accurately, since it would not benefit from that coalition obtaining a discounted electricity price.

\begin{table}[t]
\centering
\caption{A representative sample of coalition value
calculation allocation $CV^{s}_x$ for agents $x \in \{2,3,4\}$, where $n=6$
and the size of the coalitions $s=3$ for the three methods: N-DCA,
DCVC and VBFR.  Note that the allocation is broadly balanced for 
N-DCA and DCVC, but not for VBFR (agent $2$ has an allocation
of 6 coalitions, yet agent $4$ has a single coalition allocated).
Furthermore, for N-DCA and VBFR, each agent is a member of 
the coalitions that are allocated to it; this is
not the case for DCVC.  \label{tab:example2}}

\setlength{\tabcolsep}{10pt}
\renewcommand{\arraystretch}{1.2}

\begin{tabular}{cccccccc}
\toprule

Method & Allocation & \multicolumn{6}{c}{Coalitions} \\

\cmidrule(lr){3-8}

\multirow{3}{*}{N-DCA}
 & $CV^{3}_{2}$ & 2,3,4 & 2,3,5 & 2,3,6 & & & \\
 & $CV^{3}_{3}$ & 3,4,5 & 3,4,6 & 3,4,1 & & & \\
 & $CV^{3}_{4}$ & 4,5,6 & 4,5,1 & 4,5,2 & 4,6,2 & & \\

\midrule

\multirow{3}{*}{DCVC}
 & $CV^{3}_{2}$ & 3,4,5 & 2,5,6 & 2,4,6 & & & \\
 & $CV^{3}_{3}$ & 2,4,5 & 2,3,6 & 2,3,5 & & & \\
 & $CV^{3}_{4}$ & 2,3,4 & 1,5,6 & 1,4,6 & 1,2,4 & & \\

\midrule

\multirow{3}{*}{VBFR}
 & $CV^{3}_{2}$ & 2,5,6 & 2,4,6 & 2,4,5 & 2,3,6 & 2,3,5 & 2,3,4 \\
 & $CV^{3}_{3}$ & 3,5,6 & 3,4,6 & 3,4,5 & & & \\
 & $CV^{3}_{4}$ & 4,5,6 & & & & & \\

\bottomrule
\end{tabular}

\end{table}


The \emph{Distributed Coalition Value Calculation} 
(DCVC) family of algorithms \citep{Rahwan2005,Rahwan2007,Michalak2010} addresses many of the concerns identified in earlier work, by grouping coalitions into lists, and then using a decentralised method to divide the lists into \emph{shares}, one for each agent, resulting in a much fairer allocation of coalition value calculation sets (referred to as \emph{shares} in the DCVC literature) whilst eliminating redundancy.
Despite being endogenous, the allocation remains fully decentralised and can be performed independently by each agent without the need for communication.

The DCVC algorithm represents all feasible coalitions in structured
lists $L_s$, where $s \in \{1,\ldots,n\}$, and where each $L_s$ contains coalitions of size
$s$ ordered in reverse-lexicographical order, such that the first coalition
in the list $L_s$ is $\{n-s+1,\ldots,n\}$ and the last coalition in $L_s$ is $\{1,\ldots,s\}$. Thus, agents are aware of how
$L_s$ is ordered, even though they may not maintain the full list themselves.
Agent $i$ is assigned a contiguous block (or share) of 
$k=\lfloor  \frac{\mid L_s \mid}{n} \rfloor$ coalitions in $L_s$, corresponding to the indices $((i-1)k+1)$ to $ik$.
For those lists whose size is not
an integer multiple of the number of agents, a shared counter $\alpha$ is used to assign the
remaining coalitions (beyond $n \times \lfloor  \frac{\mid L_s\mid }{n} \rfloor$) to agents in a round-robin manner. The value of
$\alpha$ is incremented for each additional 
coalition, and reset to $1$ if it exceeds $n$. 
This counter is maintained consistently across all lists so that the maximum
difference between the number of coalitions calculated by the  
agents is at most one. 
The assignment proceeds sequentially from list $L_1$ onwards, ensuring that all agents consistently determine the  
agent responsible for each additional coalition. 
An example of the coalition value calculation shares where $n=6$
using the DCVC algorithm (for agents $2,3$ and $4$) is given in
Table~\ref{tab:example2}.  Although the DCVC algorithm has a similar
division of coalitions between agents as N-DCA, agents
will sometimes calculate the values for coalitions where they are not 
a member (for example, agent $3$ calculates coalition $\{2,4,5\}$).

There have been a few variants of DCVC in the literature.
Although the basic DCVC algorithm \citep{Rahwan2005} distributes
the coalitions evenly across all agents, the allocation
is not fair with respect to the individual operations (of comparisons
and additions) needed to generate all coalitions in each agent's share, as
the number of operations needed to find the next coalition in each agent's
share can fluctuate.  An extension to the DCVC algorithm was proposed 
\citep{Rahwan2007} to minimise (but not totally eradicate) this issue.
Crucially, none of these variants address the self-interest property, as all of the DCVC variants allocate coalitions to agents based on their \emph{positional share} of the reverse-lexicographic list~$L_s$, and thus an agent may be assigned coalitions of which it is not a member.
This is a structural consequence of the list-based allocation mechanism, rather than the specific share-sizing strategy, and thus applies equally to the basic algorithm, the modified variant, and the version used within D-IP~\citep{Michalak2010}, described below.

Subsequent work by \citet{Michalak2010} proposed a decentralised variant of the algorithm (the \emph{Distributed IP} (D-IP) algorithm) for efficiently finding the optimal coalition structure.
D-IP combined a modified version of DCVC with the IP (Integer Partition--based) search algorithm \citep{Rahwan009} to produce one of the first fully decentralised algorithms for optimal coalition structure generation.  
It consisted of three distributed stages; during the first, each agent uses a variant of DCVC to compute the characteristic function values for a subset of coalitions assigned to it.
During this process, the agents also determine the per-size aggregate statistics (i.e.\ the maximum and average coalition values $Max_s$ and $Avg_s$ for each list $L_s$) and evaluate those coalition structures that can be assessed from locally available values alone, including the grand coalition, the all-singletons partition, and structures consisting of two complementary coalitions.
During the second stage, agents exchange these statistics, enabling each agent to independently compute upper and lower bounds on every subspace, using these statistics together with the integer-partition representation of the search space (where each subspace corresponds to an integer partition of $n$), and to prune unpromising subspaces.
Novel filter rules are also applied to discard individual coalitions that provably cannot appear in an optimal structure, thereby substantially reducing the number of coalition values that need to be exchanged.  
In the final stage, the remaining subspaces are searched in a distributed manner; agents divide the coalitions in each list amongst themselves and apply a branch-and-bound technique to find the optimal coalition structure, exchanging updates as improved solutions are found.

As D-IP relies on DCVC for its value-calculation stage, it inherits certain limitations of that approach, notably that agents may be assigned coalitions of which they are not members, and, as noted earlier by \citet{Rahwan2007}, the index-to-coalition mapping based on Pascal matrices may suffer from integer overflow for standard 64-bit integers when $n \ge 68$.
Crucially, D-IP goes beyond the problem of distributing coalition value calculations, to address the full coalition structure generation problem, i.e.\, identifying the optimal \emph{partition} of agents, as opposed to the specific problem of distributing the value \emph{calculations} themselves fairly whilst ensuring self-interest.
Interestingly, these two concerns are complementary, and thus the coalition values produced by a distributed value-calculation algorithm such as N-DCA could, in principle, serve as input to a distributed structure-search algorithm such as D-IP (thereby replacing the DCVC role). We return to this possibility in Section~\ref{sec:conclusions}.

A precursor of the mathematical approach presented in Section~\ref{sec:ippe} appeared in early work by \citet{RileyAAMAS14,RileyAAAI15}. The notion of \emph{Increment Arrays} was first presented as part of the \emph{Self-Interested Coalition Value Calculation} (SICVC) approach \citep{RileyAAMAS14},
and subsequently evolved as part of the work on the \emph{Distributed Coalition Generation (DCG)} approach \citep{RileyAAAI15}. In DCG, \emph{coalition value calculation shares} (CVs) are generated from which each agent could then identify and generate a subset of the coalitions to investigate.
A theoretical evaluation of the approach proved that all the coalitions for a community of $n$ agents would appear within the CVs.
However, the published algorithm that implemented the approach was lengthy and complex, due to the need to identify canonical sequences for different equivalence classes, and thus eliminate repeated sequences.
Furthermore, no empirical analysis was presented that compared the efficiency of the DCG approach with respect to other state-of-the-art approaches. 
The \emph{Necklace-based Distributed Coalition Algorithm} (N-DCA)\footnote{A preliminary version of the N-DCA algorithm appeared in \cite{PayneIEEEWIC24}.} presented in Section~\ref{sec:algorithm} is a successor to DCG, replacing the explicit combinatorial construction with a formulation grounded in two-colour necklaces.

The final variant of DCVC in the literature to date is the D-SlyCE algorithm \citep{Voice12}. D-SlyCE is an algorithm that distributes coalition value calculations when the agents are represented as nodes on a graph. Edges between nodes (agents) represent synergistic links that allow the linked agents to form a coalition. When the graph is fully connected (i.e.\ when the graph describes a characteristic function game), then D-SlyCE mimics the operations of DCVC. 

 
The algorithms discussed above focus on distributing coalition \emph{value calculations} across the agents.
A complementary body of work has investigated how, given those values, the
space of possible coalition structures can be efficiently searched to find an optimal partition.
\citet{RahwanAIJ2015} present a comprehensive survey, cataloguing both exact and approximate approaches.
Amongst the exact algorithms, the \emph{IP} algorithm \citep{Rahwan009} established the integer partition-based search strategy, which was later refined into the \emph{Improved Dynamic Programming} (IDP) approach \citep{rahwan08}.
\citet{Michalak2016} later combined these two strategies into ODP-IP, a hybrid algorithm that was at that time the fastest known exact solver for the complete set partitioning problem.
More recently, \citet{Changder2020} introduced the \emph{Overlapping, Dividing the subspace, and Subspace Shrinking} (ODSS) algorithm, which partitions the search space into two disjoint sets of subspaces and introduces a novel subspace-shrinking technique, yielding substantial speed-ups over ODP-IP on several benchmark value distributions, although not uniformly across all distributions.
These \emph{centralised} algorithms assume that a single entity has access to all $2^n - 1$ coalition values, which contrasts with the \emph{decentralised} approaches discussed earlier, that address the preceding step of \emph{computing} those values in a distributed setting, and that should ideally produce per-agent allocations that are fair, non-redundant, and self-interested. The coalition values generated by such algorithms could subsequently serve as input to any of the above structure-search algorithms, whether centralised or, as in D-IP, distributed.


Several other recent lines of work address coalition structure generation under additional constraints or using alternative computational paradigms.
For example, in \emph{Graph-Constrained Coalition Formation} (GCCF) by \citet{Bistaffa2017}, feasible coalitions
are restricted to connected subgraphs of an interaction network.
The authors propose CFSS, an anytime branch-and-bound algorithm (and its parallel variant) and demonstrate good performance in application domains including ride sharing and collective energy purchasing.
However, the approach itself remains centralised, and does not address the distribution of coalition-value computations across agents.
Furthermore, in the context of unrestricted characteristic-function games, \citet{Taguelmimt2024SMART,Taguelmimt2025SALDAE} have recently proposed two algorithms, \emph{SMART} and \emph{SALDAE}. 
\emph{SMART} \citep{Taguelmimt2024SMART} is a centralised exact algorithm built on three complementary techniques:
(i) complementarity-based dynamic programming (CDP);
(ii) gradual search over solution subspaces (GRAD);
and (iii) distributed integer-partition-based search (DIPS).
For each problem size, an offline preprocessing step selects the pair of complementary coalition-size sets that minimises the dynamic-programming search space.
In contrast, \emph{SALDAE} \citep{Taguelmimt2025SALDAE} is a centralised anytime algorithm inspired by multiagent path finding that searches a graph of coalition structures using multiple search agents and conflict-resolution mechanisms.
It is designed to scale to large instances with hundreds or thousands of agents, returning progressively better solutions as computation time permits.
All three algorithms are examples of approaches that address coalition \emph{structure} generation
(i.e.\ the second stage~of the framework proposed by \citet{sandholm99}),
and thus complement rather than replace distributed coalition
\emph{value} calculation algorithms, such as DCVC, etc.

None of the above approaches address the problem that specifically motivates the present work, namely the fair, decentralised, and communication-free distribution of coalition value calculations across agents in unrestricted characteristic function games, with formal guarantees that each agent only evaluates coalitions of which it is a member.
N-DCA is, to the best of our knowledge, the only algorithm that provably satisfies all five of the desirable properties listed in Section~\ref{sec:intro}.


\section{Preliminaries} \label{sec:preliminaries}

\subsection{Characteristic Function Games}
\label{sec:backgoundCFG}

The original notion of coalitions emerged from the
seminal work of \citet{vonNeumann1944}, where they
constructed the theory of \emph{$n$-person cooperative games}
in characteristic function form, where each coalition has an
associated real numeric \emph{utility} value that it
can achieve by completing some task or achieving some goal.
A \emph{Characteristic Function Game (CFG)} is denoted
${\mathcal G} = \langle Ag, \nu \rangle$, where $Ag$ is the set
of $n$ agents in the game and $\nu$ is the characteristic 
function that maps every potential coalition $C \subseteq Ag$
to a real numeric value (i.e.\ $\nu : 2^n \rightarrow \mathbb{R}$).
By default an empty coalition receives no utility payoff (i.e.\
$v(\varnothing) = 0$).

In many CFGs, it is often 
assumed that each agent shares the characteristic function
(and thus has perfect information to use the characteristic function 
to get an accurate utility value for each coalition), and that the value of
each coalition does \emph{not} change over time.  Furthermore, each coalition's
characteristic function value is independent of the \emph{externalities} of the
coalition formation (i.e.\ the inter-coalitional dependencies \citep{RahwanMWJ12}).
The \emph{Grand Coalition} $C$ for a characteristic function game
${\mathcal G} = \langle Ag, \nu \rangle$ is a coalition of
all of the agents within the game ($C = Ag$).

In characteristic function games that are not superadditive or subadditive,
the value of each coalition needs to be calculated so that
an optimal
coalition structure can be found.
Calculating the value of every coalition has high computation costs
due to the exponential number of possible coalitions, which is
equal to the powerset of $Ag$ without the emptyset (i.e.\ 
$P(Ag)~\backslash~\{\varnothing\}$).  However, it is possible within a multi-agent system
to spread the computational cost of calculating these values across all of
the agents via a \emph{coalition value calculation set (CVCS)}.  We say that 
a \emph{coalition value calculation allocation (CV)} of an agent $x \in Ag$ is denoted
$CV_x = \{C^1, \ldots, C^j\}$, consisting of a set of $j$ coalitions
that agent $x$ has been assigned to calculate the value of. 
If we specifically want to refer to the coalitions in $CV_x$
that are of size $s$, then the notation $CV_x^s$ is used. 
When all coalitions have been distributed into at least one coalition value 
calculation allocation, we call this combination of all allocations
the \emph{coalition value calculation set} (denoted $CVCS$), which is defined as:
\[
CVCS = \bigcup_{x = 1}^n CV_x = P(Ag)~\backslash~\{\varnothing\}
\]
A coalition value calculation set (\emph{CVCS}) is equivalent to the powerset 
of all agents minus the empty set.  A special type of coalition 
value calculation set is the \emph{minimal coalition value calculation 
set}, denoted $CVCS^*$, where $CVCS^* = \{CV_1, \ldots, CV_n\}$ and none of the elements in 
$CV_x~\cap~CV_y = \varnothing \text{~for any~} x,y \in Ag \text{~where~} x \neq y$.

\subsection{Combinatorial Necklaces}
\label{sec:backgoundNecklaces}

Combinatorial necklaces (or \emph{necklace permutations}) have been used in a variety of different modelling domains, with some of their earliest uses appearing within discussions on cyclic hashing schemes \citep{10.2307/2318994}.
More recently, necklaces have been used for generating kernels to support star tracking \citep{ARNAS2017291}, by mobile ad-hoc networks to facilitate mutual discovery of wireless transceivers \citep{8292573}, and by data compression techniques such as \emph{Compression by Substring Enumeration} \citep{6033796}.
They also appear within music theory to model scales within a chromatic universe \citep{Rappaport2007,7839578} and rhythm \citep{TOUSSAINT20102}.
The necklace problem is a combinatorial problem that determines how many different arrangements of a string of beads exist within a circular loop of $n$ beads, given $k$ colours (see Figure~\ref{fig:necklaceExample}).
Two necklaces are identical (i.e.\ they belong to the same equivalence class) if and only if one can be transformed into another through the cyclic rotation of the beads (Figure~\ref{fig:necklaceExample} left).
Thus, we can say that a \emph{necklace of $n$ beads in $k$ colours} is an equivalence class of $k$-ary $n$-tuples (i.e.\ strings of \emph{beads}) under rotation \citep{RUSKEY1992414}.  
An advantage of utilising combinatorial necklaces is the fact that there are many well understood formulae that can be used to interrogate additional information about the necklace itself; for example, the total number of necklaces of length $n$ given $k$ coloured beads is defined by:

\begin{equation}\label{eqn:totalnecklaces}
N_k(n) = \frac{1}{n}\sum_{d \mid n} \phi(d)k^{\frac{n}{d}}
\end{equation}

\noindent
i.e.\ the sum over all of the divisors $d$ of $n$, where $\phi(d)$ is  Euler's totient function of $d$, an arithmetic function that counts the number of positive integers up to a given integer $d$ that are also co-prime with $d$ (i.e.\ the number of integers $x$ in the range $1 \le x \le d$ for which the greatest common divisor $gcd(d,x)$ is equal to 1).
Within this study, we only consider two-colour necklaces (i.e.\ $k=2$).
Equation \eqref{eqn:fixeddensitynecklaces} determines the number of fixed density 2-colour necklaces of density $s$ \citep{Sawada:1999:EAG:314500.314910}.
\begin{equation}\label{eqn:fixeddensitynecklaces}
    N_2(n,s) = \frac{1}{n}\sum_{d \mid \gcd(n,s)} \phi(d)*\binom{\frac{n}{d}}{\frac{s}{d}}
\end{equation}

\noindent
Two-colour necklaces ($k=2$) are used here as we only consider two different states for each position within the necklace: relating to whether an agent is in a coalition, or it is not in a coalition.
This is achieved through the construction of \emph{increment arrays} (Section~\ref{sec:ippe}), that are used to generate the coalition allocations for each agent. The bijection between 2-colour necklaces and increment arrays (IAs) is given in Section~\ref{sec:bijection}, and the total number of IAs across all coalition sizes for a given number $n$ of agents is given by $N_2(n)$, whereas $N_2(n,s)$ determines the number of IAs generated for a given coalition size $s$ (Equation \eqref{eqn:fixeddensitynecklaces}).

\section{Distributed Coalition Value Allocation} \label{sec:ippe}

This section develops the mathematical framework underlying the distributed, non-redundant, communication-free allocation of coalition value calculations.
The central concept developed here is the notion of an \emph{increment array} (IA), a compact representation of a coalition that captures the spacing between its members.
We show that IAs partition naturally into equivalence classes under circular shifts, and that each equivalence class generates a distinct, non-overlapping set of coalitions (Theorem~\ref{thm:main2}).
For classes whose IA contains a repeating sub-sequence, multiple
agents can generate the same coalition. Thus, a \emph{designation scheme} is described that assigns responsibility to a subset of agents, thereby ensuring that every coalition of every size is computed exactly once.
A rotated variant of this scheme balances the workload so that the total number of coalition value calculations assigned to any two agents differs by at most one (Theorem~\ref{thm:aggregate-balance}).
Throughout, every agent can determine its own assignments independently, using only its identifier and the shared parameter~$n$, thus eliminating the need for communication between agents, and ensuring that each agent is self-interested; i.e.\ it is only assigned coalitions in which it is a member.
The algorithmic realisation of this framework is presented in Section~\ref{sec:algorithm}; Figure~\ref{fig:pipeline} provides an overview of the complete pipeline.

\begin{figure}[t]
\centering
    \centering
    \includegraphics[width=\linewidth,clip,trim=10 440 140 380]{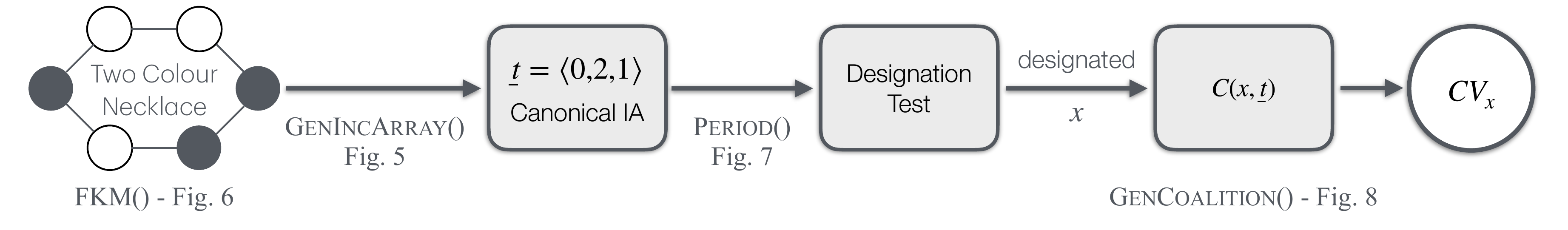}

\caption{Overview of the N-DCA pipeline executed independently by each agent~$x$.  
Two-colour necklaces of length~$n$ are generated by the FKM algorithm (Section~\ref{sec:bijection}); each necklace is converted to its canonical increment array~$\underline{t}$ via \textsc{GenIncArray}; the period of~$\underline{t}$ determines whether agent~$x$ is designated to evaluate it (Section~\ref{sec:designation}); and designated IAs are converted to coalitions via \textsc{GenCoalition}.  The union of all generated coalitions forms agent~$x$'s allocation~$CV_x$.  The entire process requires only the agent's identifier~$x$ and the shared parameter~$n$.}
\label{fig:pipeline}
\end{figure}

\subsection{Preliminaries and Integer Increments}

In the context of a characteristic function game, we assume that there are $n \in \mathbb{N}$ agents (typically $n > 2$), whose identifiers are represented by contiguous natural numbers; i.e.\ $Ag =  \{1, 2, \ldots, n\}$.\footnote{For convenience, this set of agents may also be referred to as $\set{x_1,x_2,\ldots,x_n}$.} 
A coalition $C \subseteq Ag$ of size $s = \mid C \mid$ is therefore represented as an ordered sequence of agents (modulo $n$) with respect to their identifiers (IDs), where $s \in \mathbb{N}$, such that an agent can only appear once in any coalition within the characteristic function game model.
Throughout, all agent identifiers are taken modulo $n$ in the range $\{1,2,\ldots,n\}$; when arithmetic on an agent ID yields $0$, it is replaced by $n$ (since $0\equiv n \pmod{n}$).
We also use $x$ (or $y$) to denote a specific agent identifier, whereas $i$ and $j$ are used as bound variables in quantified statements.

\begin{definition}[Offset increment]
\label{def:offset-increment}
Let $C = \{x_1, x_2, \ldots, x_s\} \subseteq \{1, 2, \ldots, n\}$ be a
coalition of size $s$, with members ordered cyclically (modulo $n$).
The \emph{integer increment} between two consecutive members $x_i$ and
$x_{i+1}$ is the circular distance $(x_{i+1} - x_i) \bmod n$.  This
decomposes as the sum of a \emph{baseline increment} of~$1$ (representing
adjacency) and an \emph{offset increment} $t_i \ge 0$ (representing the
number of agents omitted between $x_i$ and $x_{i+1}$).  That is, the
integer increment between $x_i$ and $x_{i+1}$ equals $t_i + 1$.
\end{definition}
Thus, if a single agent is omitted from the coalition (e.g.\ $x_4$ is omitted from a set of six agents), the integer increment between agents $x_3, x_5$ will be 2 (i.e.\ the baseline and a single offset increment), whereas the integer increment for the other five agents is simply 1. 
As the total integer increment for a coalition of size $s$ from $n$ agents is $n$, it trivially follows that:


\begin{lemma}\label{lemma:offset-increment}
Let $C\subseteq\set{1,2,\ldots,n}$ be a coalition with $s={\mid}C{\mid}$ agents, such that
$C=\set{x_1,x_2,\ldots,x_s}$, and 
$\langle t_0, t_1\ldots, t_{(s-1)}\rangle$
is the corresponding sequence of offset increments. The total
offset increment for $C$ is given as:

\[
	\sum_{i=0}^{s-1} t_i = n - s
\]
\end{lemma}
\begin{proof}
As a pair of contiguous agents will have an offset increment of 0, the Grand Coalition (of $n$ agents) will also have a total offset increment of 0.  
Removing one agent merges two adjacent integer increments (each of value 1) into a single increment of value 2, thereby increasing the total offset by 1. Removing $n-s$ agents therefore increases the total offset from 0 to $n - s$.
\end{proof}

To illustrate this, consider the coalition $\set{1, 3, 6}$ ($s=3$) generated from a set of six agents (i.e.\ $n=6$).  
The integer increment between the first and second agent ($1$ and $3$ respectively) consists of a single baseline increment, and an offset increment,\footnote{Note that offset increments are indexed from $t_0$ (referring to positions within the IA tuple), whereas agent identifiers are indexed from $1$. Thus $t_0$ denotes the offset increment between the first and second members of the coalition, $t_1$ between the second and third, and so on.} such that $t_0 = 1$.  
Likewise, the integer increment between $3$ and $6$ consists of a single baseline increment, and an offset increment of 2; i.e.\ $t_1 = 2$.
The integer increment between $6$ and $1$ is only 1 (i.e.\ the baseline increment) as by increasing the value $6$ by the baseline
increment, the value $7$ is obtained.
Given that $1$ is congruent to $7$ modulo $6$ (i.e.\ $1 \equiv 7 \bmod(n)$), this yields an offset increment $t_2 = 0$, and a total offset increment of $\sum_{i=0}^{s-1} t_i = 3$.
This is illustrated in Figure~\ref{fig:increment} where the agents in the coalition appear within a cycle (to represent the modulo assumption).  

\begin{figure}[t]
    \centering
    \includegraphics[width=0.5\linewidth]{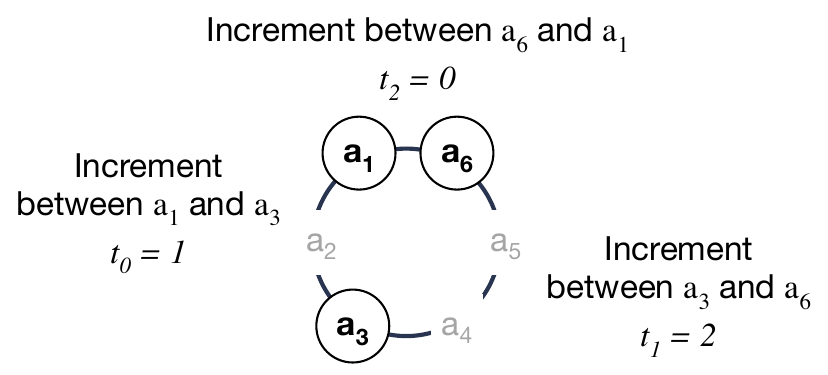}
    \caption{The increments between three agents in the coalition $\set{1, 3, 6}$ ($s=3$ from $n=6$ agents), arranged as cycle.  Note that the missing agents are shown in a lighter font. Each increment represents these missing agents.}
	\label{fig:increment}
\end{figure}

It is possible to characterise a coalition by a sequence of offset increments, such that a coalition $C \subset Ag$ (where $s < n$) will have a total of $n-s$ offset increments.
Thus, multiple coalitions can be characterised by using the same sequence of offset increments if the starting agent of the coalition is changed.
For example, when $n=6$ and $s=3$, coalitions $\{1,2,5\}$ and $\{2,3,6\}$ have the same increment between first and second agent (i.e.\ the baseline increment of 1 and therefore an offset increment of $t_0 = 0$), the same increment between the second and third agent (i.e.\ the baseline increment plus an offset increment of $t_1 = 2$), and the same increment between the third and first agent (i.e.\ the baseline increment plus an offset increment of $t_2 = 1$).
The total number of unique coalitions of size $s$ from $n$ agents can be calculated using binomial coefficients, i.e.\ $\binom{n}{s} = \frac{n!}{s!(n-s)!}$ and share the same \emph{total offset increments} (Lemma~\ref{lemma:offset-increment}); however the distribution of \emph{total offset increments} will vary, depending on the coalition. This is illustrated in Table~\ref{tab:example1}, which lists all of the coalitions possible for $n=6$.
We therefore characterise the various distributions of offset increments by defining different \emph{increment arrays} (Definition~\ref{def:increment-array}).

An \emph{integer partition} of $k$ is a collection of natural numbers (i.e.\ summands, or \emph{parts}) that add up to exactly $k$.
It is therefore possible to generate a subset of integer partitions for $n-s$ when $s < n$, to determine the different offset increments between different coalition members.
The full set of partitions is denoted ${\mathcal I}(n-s)$.
For example, if $n=6$ and $s = 3$, then this results in a set of 3 partitions: ${\mathcal I}(3) = \{(3),(2,1),(1,1,1)\}$,  with $1$, $2$, and $3$ \emph{parts} respectively, whereas if $s=2$, this results in a set of 5 partitions: ${\mathcal I}(4) = \{(4),(3,1),(2,2),(2,1,1),(1,1,1,1)\}$.

\subsection{Increment Arrays}
\label{sec:increment-arrays}

An \emph{increment array} (IA) denoted $\underline{t} = \langle t_0,t_1,\ldots,t_{s-1} \rangle$ is a tuple that is constructed from an integer partition $\lambda \in {\mathcal I}(n-s)$, where the number of parts of $\lambda$ is less than or equal to $s$ (i.e.\  ${\mid}\lambda{\mid}\leq s$). If the number of parts of $\lambda$ is less than $s$, then $\lambda$ can be used to construct the IA with $s - {\mid}\lambda{\mid}$ additional elements that have the value $0$ in the tuple. 
Since the elements of a partition may be arranged in different
orders within the tuple, a single partition can give rise to
multiple distinct IAs.  However, many of these are related by
circular shifts: for example, $\langle 0,1,2 \rangle$ and
$\langle 2,0,1 \rangle$ are circular shifts of each other and
generate the same set of coalitions (across all starting agents).
We formalise this in Section~\ref{sec:equivalence} by defining an
\emph{equivalence class} $\approx$ under circular shifts; only one
canonical representative per class is needed to generate a minimal
coalition value calculation set ($CVCS^*$), i.e.\ the union of all
coalition value calculation allocations with no duplicate coalitions.
 
\begin{definition}[Increment array]
\label{def:increment-array}
An \emph{increment array} (IA) of size $s$ for $n$ agents is a tuple
$\underline{t} = \langle t_0, t_1, \ldots, t_{s-1} \rangle$ of
non-negative integers satisfying
$\sum_{i=0}^{s-1} t_i = n - s$,
constructed from an integer partition
$\lambda \in \mathcal{I}(n - s)$ with $|\lambda| \le s$ parts, padded with
$s - |\lambda|$ zeros if necessary.  Each $t_i$ is the offset increment
(Definition~\ref{def:offset-increment}) at position $i$ within the
coalition.
\end{definition}

Again, 
consider the case where $n=6$, $s = 3$, and the partition $\lambda=(2,1)$ is used to construct an IA.  As the number of parts in $\lambda$ is less than $s$, the partition will be augmented by an additional element whose value is $0$ (which does not affect the sum of the partition). The selection of this partition would result in 6 tuples, of which two can be selected as representative IAs, whereas the others are simply permutations due to circular shifts; i.e.\ 
$\langle 0,1,2 \rangle$, has the two circular shifts
$\langle 2,0,1 \rangle$ and
$\langle 1,2,0 \rangle$; yet 
$\langle 1,0,2 \rangle$, has the two circular shifts
$\langle 2,1,0 \rangle$ and
$\langle 0,2,1 \rangle$.
The resulting IA $\langle 0,1,2 \rangle$ represents offset increments for the coalition $\set{5,6,2}$ (amongst others), whereas 
the IA 
$\langle 0,2,1 \rangle$ represents the offset increments for the coalition $\set{5,6,3}$.  In contrast, the IA $\langle 5 \rangle$ represents singleton coalitions for all 6 agents, as the offset increment and baseline increment $t_0 + 1 = n$. As the addition of the increment $n$ to any agent $x \in C$ (modulo $n$) results in $x$, any IA of size 1 is equivalent to the identity function.

\subsection{Generating Coalition Value Calculation Allocations}
\label{sec:coalition-value-calculation-allocations}

An IA characterises a coalition, starting with an initial agent identifier, as each element of the IA defines the increments in the value of the identifier for each of the agents listed within the coalition. We can define the \emph{cumulative integer increment} $\varphi_i$ for $1\le i\le s$ as the cumulative sum of the baseline and offset increments for the agents preceding $x_i$ in the coalition $C=\set{x_1,x_2,\ldots,x_s}$:
\begin{equation}\label{eqn:offset}
\varphi_i~=~\left\{
{
\begin{array}{lcl}
0&\mbox{ if }&i=1\\
\sum_{k=0}^{i-2}~(t_k + 1)&\mbox{ if }&2\leq i\leq s
\end{array}
}
\right.
\end{equation}
Note that $\varphi_{s+1}=n$, as this is equivalent to the sum of all offset increments (which, by Lemma~\ref{lemma:offset-increment}, is $n-s$) and $s$ baseline increments.

\begin{definition}[Coalition generation function]
\label{def:coalition-generation}
Given an IA $\underline{t} = \langle t_0, t_1, \ldots, t_{s-1} \rangle$
and a starting agent $x \in \{1, 2, \ldots, n\}$, the \emph{coalition
generated by $\underline{t}$ from $x$} is:
\begin{equation}\label{eqn:valuex}
    C(x, \underline{t}) \;=\; \{x\} \;\cup\;
    \bigcup_{i=2}^{s} \bigl\{\, (x + \varphi_i) \bmod n \,\bigr\}
\end{equation}
where the cumulative integer increment $\varphi_i$ is defined as in
Equation~(2).  By construction, the starting agent $x$ is always a
member of $C(x, \underline{t})$.
\end{definition}

Each IA $\underline{t}$ represents the necessary offset increments from one agent ID of the coalition array to the next (modulo $n$), and thus can be used by an agent $x$ to \emph{generate a coalition} $C$ that includes itself.
This is based on the assumption that the first element of $C$ always refers to this agent (i.e.\ the agent $x$ is \emph{guaranteed} to be a member of the coalition $C$), and the remaining agents are determined using the function $C(x,\underline{t})$ as defined above in Equation \eqref{eqn:valuex}.

The following lemma states that for any coalition $C$, there is always some IA $\underline{t}$ and an agent $x\in C$ that can generate it:

\begin{lemma}\label{lemma:ns-sequence}
Let $C\subseteq\set{1,2,\ldots,n}$ with ${\mid}C{\mid}=s$. 
There is an increment array (IA) $\underline{t}$ and $x\in C$ such that $C=C(x,\underline{t})$.

\end{lemma}
\begin{proof}
Let $C = \{c_1, c_2, \ldots, c_s\} \subseteq \{1, 2, \ldots, n\}$ with $c_1 < c_2 < \cdots < c_s$. Set $x = c_1$ and define
\begin{align*}
t_k &= c_{k+2} - c_{k+1} - 1 \quad \text{for } 0 \leq k \leq s-2, \\
t_{s-1} &= (n - s) - \sum_{i=0}^{s-2} t_i
\end{align*}
We verify that $\underline{t}$ is a valid IA and that it generates $C$.

\begin{enumerate}[label=(\roman*)]
\item \emph{Non-negativity.}
For $0 \leq k \leq s-2$: since $c_{k+1} < c_{k+2}$ are distinct positive integers, $c_{k+2} - c_{k+1} \geq 1$, so $t_k \geq 0$.

For $t_{s-1}$: summing the first $s-1$ terms gives
\[
\sum_{k=0}^{s-2} t_k
= (c_2 - c_1 - 1) + (c_3 - c_2 - 1) + \cdots + (c_s - c_{s-1} - 1)
= c_s - c_1 - (s-1)
\]
where all intermediate $c_j$ cancel in pairs. Hence
\[
t_{s-1} = (n-s) - (c_s - c_1 - (s-1)) = n - 1 - c_s + c_1 \geq 0
\]
since $c_1 \geq 1$ and $c_s \leq n$.

\item \emph{Correct sum.}
By definition, $t_{s-1} = (n-s) - \sum_{i=0}^{s-2} t_i$, so $\sum_{i=0}^{s-1} t_i = \sum_{i=0}^{s-2} t_i + t_{s-1} = n - s$.

\item \emph{Correct coalition.}
By Equation \eqref{eqn:offset}, for $2 \leq i \leq s$:
\[
\varphi_i
= \sum_{k=0}^{i-2}(t_k + 1)
= \sum_{k=0}^{i-2}(c_{k+2} - c_{k+1})
= c_i - c_1
\]
Hence $(x + \varphi_i) \bmod n = (c_1 + c_i - c_1) \bmod n = c_i$ for $2 \leq i \leq s$, since $1 \leq c_i \leq n$. Together with $x = c_1$, we obtain $C(x, \underline{t}) = C$.

\end{enumerate}
\end{proof}
\noindent
Table~\ref{tab:example1} illustrates the subset of the coalition value calculation sets for each agent, grouped into lists (i.e.\ columns) 
for each of the canonical IAs of size $s$, generated for each of the $n=6$ agents (i.e.\ rows).
The union of the coalition value calculation sets listed in this Table represent the $\binom{n}{s}$ unique subsets of size $s$ in the powerset of $N=\{1,2,3,4,5,6\}$.
Each column represents a set of coalitions for each canonical IA $\underline{t}$.
The rows represent the different coalitions that can be generated using $C(x,\underline{t})$ for each agent $x$ for some $1\leq x\leq n$, where $x$ appears as the first element in the coalition.
These are the coalitions that represent the \emph{allocation} of \emph{coalition value calculation allocations} given to each agent $x$.
Note that every possible coalition is assigned once and only once, and that an integer partition may form more than one IA; for example the two IAs $\langle 0,1,2\rangle$ and $\langle 0,2,1\rangle$ are both formed from the \{2,1\} integer partition.  

\begin{table}[t]
\centering
\caption{Coalition value calculation allocations $CV_x^{s}$ for all agents $x$ where $1 \le x \le n, n=6$, resulting in $2^6 - 1 = 63$ unique coalitions. Periodic sequence designations using the rotated designation scheme (Section~\ref{sec:designation}) are shown in \textbf{bold}, resulting in a balanced allocation between 10 and 11 coalitions per agent (final column).  Note that the allocation of the grand coalition (not shown); i.e.\ $s=6, \underline{t}= \langle0,0,0,0,0,0\rangle$ is designated to agent 3 (i.e.\ $CV^{6}_{3}$) using this scheme, and as such this agent is allocated 11 coalitions.}
\label{tab:example1}

\footnotesize
\setlength{\tabcolsep}{0.7pt}
\renewcommand{\arraystretch}{1.3}

\begin{tabular}{@{}l c ccc cccc ccc c c@{}}
\toprule

{} & \multicolumn{1}{c}{s=1} & \multicolumn{3}{c}{s=2} & \multicolumn{4}{c}{s=3} & \multicolumn{3}{c}{s=4} & \multicolumn{1}{c}{s=5} & \multicolumn{1}{c}{Total}\\

\cmidrule(lr){2-2} \cmidrule(lr){3-5} \cmidrule(lr){6-9} \cmidrule(lr){10-12} \cmidrule(lr){13-13} \cmidrule(lr){14-14}

\multicolumn{1}{c}{$\underline{t}$} & $\langle5\rangle$ & $\langle0,4\rangle$ & $\langle1,3\rangle$ & $\langle2,2\rangle$ & $\langle0,0,3\rangle$ & $\langle0,1,2\rangle$ & $\langle0,2,1\rangle$ & $\langle1,1,1\rangle$ & $\langle0,0,0,2\rangle$ & $\langle0,0,1,1\rangle$ & $\langle0,1,0,1\rangle$ & $\langle0,0,0,0,1\rangle$ & $|CV_x|$\\

\cmidrule(r){1-13} \cmidrule(l){14-14}
\multicolumn{1}{c}{$\varpi(\underline{t})$} & 6 & 6 & 6 & 3 & 6 & 6 & 6 & 2 & 6 & 6 & 3 & 6 & {} \\
\cmidrule(r){1-13} \cmidrule(l){14-14}

$CV^{s}_{1}$ & 1 & 1,2 & 1,3 & \textbf{1,4} & 1,2,3 & 1,2,4 & 1,2,5 & {} & 1,2,3,4 & 1,2,3,5 & \textbf{1,2,4,5} & 1,2,3,4,5 & 11\\

$CV^{s}_{2}$ & 2 & 2,3 & 2,4 & \textbf{2,5} & 2,3,4 & 2,3,5 & 2,3,6 & {} & 2,3,4,5 & 2,3,4,6 & \textbf{2,3,5,6} & 2,3,4,5,6 & 11\\

$CV^{s}_{3}$ & 3 & 3,4 & 3,5 & \textbf{3,6} & 3,4,5 & 3,4,6 & 3,4,1 & {} & 3,4,5,6 & 3,4,5,1 & {} & 3,4,5,6,1 & 11\\

$CV^{s}_{4}$ & 4 & 4,5 & 4,6 & {} & 4,5,6 & 4,5,1 & 4,5,2 & \textbf{4,6,2} & 4,5,6,1 & 4,5,6,2 & {} & 4,5,6,1,2 & 10\\

$CV^{s}_{5}$ & 5 & 5,6 & 5,1 & {} & 5,6,1 & 5,6,2 & 5,6,3 & \textbf{5,1,3} & 5,6,1,2 & 5,6,1,3 & {} & 5,6,1,2,3 & 10 \\

$CV^{s}_{6}$ & 6 & 6,1 & 6,2 & {} & 6,1,2 & 6,1,3 & 6,1,4 & {} & 6,1,2,3 & 6,1,2,4 & \textbf{6,1,3,4} & 6,1,2,3,4 & 10 \\

\bottomrule
\end{tabular}
\end{table}

To illustrate this process, the coalition value calculation set (CV) for agent $x=5$ (assuming $n=6$ and $s=3$) can therefore be determined using the first three IAs:
{
  \medmuskip=1mu
  \thinmuskip=1mu
  \thickmuskip=1mu
  \begin{align*}
  C(5,\langle 0,0,3\rangle) &= \set{5, (5+0+1)\bmod(6), ((5+0+1)+0+1)\bmod(6)} \equiv \set{5,6,1} \\
  C(5,\langle 0,1,2\rangle) &= \set{5, (5+0+1)\bmod(6), ((5+0+1)+1+1)\bmod(6)} \equiv \set{5,6,2} \\
  C(5,\langle 0,2,1\rangle) &= \set{5, (5+0+1)\bmod(6), ((5+0+1)+2+1)\bmod(6)} \equiv \set{5,6,3}
  \end{align*}
}
The existence of repeating sub-sequences within an IA can lead to duplicate coalitions being formed, depending on the initial agent.
For example, in Table~\ref{tab:example1}, the column for $\langle 1,1,1\rangle$ only assigns coalitions to two agents, as the IA $\underline{t}$ contains a repeated sub-sequence $\langle 1 \rangle$ which appears $3$ times.
Therefore, the coalition generated for $C(1,\underline{t})$, $C(3,\underline{t})$, and $C(5,\underline{t})$ all contain the same agents (albeit in a different order):
{
  \medmuskip=1mu
  \thinmuskip=1mu
  \thickmuskip=1mu
  \begin{align*}
  C(1,\langle 1,1,1\rangle) &= \set{1, (1+1+1)\bmod(6), ((1+1+1)+1+1)\bmod(6)} \equiv \set{1,3,5} \\
  C(3,\langle 1,1,1\rangle) &= \set{3, (3+1+1)\bmod(6), ((3+1+1)+1+1)\bmod(6)} \equiv \set{3,5,1} \\
  C(5,\langle 1,1,1\rangle) &= \set{5, (5+1+1)\bmod(6), ((5+1+1)+1+1)\bmod(6)} \equiv \set{5,1,3}
  \end{align*}
}
In the following subsection, we explore why this is the case, and prove that the coalition value calculation sets generated by these four different IAs are equivalent to all of the possible $\binom{n}{s}$ sets of size $s=3$ from $n=6$ agents.

\subsection{Equivalence Classes}
\label{sec:equivalence}

It is possible to generate a number of IAs that are permutations of each other through circular shifts, and as such all belong to the same \emph{equivalence class} $\approx$. Therefore, only a single, canonical representative IA for each equivalence class ($[\underline{t}]_{\approx}$) is required when constructing a coalition value calculation.

\begin{definition}[IA equivalence]
\label{def:ia-equivalence}
Two IAs $\underline{t} = \langle t_0, t_1, \ldots, t_{s-1} \rangle$ and
$\underline{u} = \langle u_0, u_1, \ldots, u_{s-1} \rangle$ of the same
size $s$ are \emph{equivalent}, written $\underline{t} \approx
\underline{u}$, if $\underline{u}$ is a circular shift of
$\underline{t}$; that is, there exists $0 \le k \le s - 1$ such that
\[
    \langle u_0, u_1, \ldots, u_{s-1} \rangle
    \;=\;
    \langle t_k, t_{k+1}, \ldots, t_{s-1}, t_0, t_1, \ldots, t_{k-1}
    \rangle
\]
The relation $\approx$ is an equivalence relation.  The equivalence class
of $\underline{t}$ is denoted $[\underline{t}]_\approx$, and its
\emph{canonical representative} is the lexicographically smallest member
of $[\underline{t}]_\approx$.
\end{definition}

\begin{lemma}\label{lemma:same-sets}
If $\underline{t}\approx\underline{u}$ then
\[
\bigcup_{i=1}^{n}~\set{~C(i,\underline{t})~}~~=~~\bigcup_{i=1}^{n}~\set{~C(i,\underline{u})~}
\]
\end{lemma}
\begin{proof}[Proof sketch (full proof in Appendix~\ref{sec:appendix})]
It suffices to show the result for a single circular shift; the
general case follows by composition.  Without loss of generality,
let $\underline{u}$ be obtained from $\underline{t}$ by moving the
last element to the front:
$\underline{u} = \langle t_{s-1},\, t_0,\, t_1,\, \ldots,\, t_{s-2} \rangle$.
Write $\psi_r$ for the cumulative increments under~$\underline{u}$,
analogous to~$\varphi_r$ under~$\underline{t}$.
One can verify that $\psi_k = \varphi_k + (t_{s-1} - t_{k-2})$
for $2 \le k \le s$, from which it follows that each coalition
$C(i, \underline{t})$ equals $C(j, \underline{u})$ for a specific
starting agent~$j$ determined by~$i$ and~$t_{s-1}$.
Crucially, the mapping $i \mapsto j$ is a bijection on
$\{1,\ldots,n\}$, so the union over all starting agents is identical
for $\underline{t}$ and $\underline{u}$.
\end{proof}

\noindent
Lemma~\ref{lemma:same-sets} shows that increment arrays (IAs) belonging to the same equivalence class ($[\underline{t}]_{\approx}$)
generate exactly the same set of coalitions of size $s$ from $\set{1,\ldots,n}$. 
For example, given 
$\underline{u}' = \langle 1,0,2 \rangle$ and
$\underline{t}' = \langle 0,2,1\rangle$, we 
can say that they both belong to the same
equivalence class, i.e.\ $\underline{t}'\approx\underline{u}'$
as we can generate the coalition
with the values $\{1,4,6\}$ by simply selecting different agents: 
{
  \medmuskip=1mu
  \thinmuskip=1mu
  \thickmuskip=1mu
  \begin{align*}
  C(4,\langle 1,0,2\rangle) &= \set{4, (4+1+1)\bmod(6), ((4+1+1)+0+1)\bmod(6)} \equiv \set{4,6,1} \\
  C(6,\langle 0,2,1\rangle) &= \set{6, (6+0+1)\bmod(6), ((6+0+1)+2+1)\bmod(6)} \equiv \set{6,1,4}
  \end{align*}
}
\noindent
Note that $\underline{u}'$ results from a single circular shift from $\underline{t}'$.

\begin{lemma}\label{lemma:ordering}
Let $C=C(x_i,\underline{t})$ and
$C=\set{x_1,~x_2,\ldots,x_i,\ldots,x_s}$
with $x_i<x_{i+1}$ for all $1\leq i<s$. There is an increment array (IA), $\underline{u}$, for which $\underline{t}\approx\underline{u}$
and $C(x_1,\underline{u})$ generates $C$ in strictly increasing ordering of $x_i$, i.e.\
\[
x_i~~\in~~\left\{{x_1~+~\sum_{k=0}^{i-2}~u_k~+~i-1,~x_1~+~\sum_{k=0}^{i-2}~u_k~+~i-1-n}\right\}~~~\forall~2\leq i\leq s
\]
\end{lemma}
\begin{proof}
Given $\underline{t}$, suppose
$C(x_i,\underline{t})=\set{x_1,x_2,\ldots,x_i,\ldots,x_s}$.
The first observation is that the following terms are strictly increasing:
$x_i~+~\sum_{k=0}^{r-2}~t_k~+~r-1~~=~~x_i~+~\varphi_r$.
It follows that if $x_i\not=x_1$ there must be a \emph{unique} 
index, $p$, for which:
\[
x_i~+~\varphi_r~~\mbox{ is }~~\left\{
{
\begin{array}{lr}
\leq~n&\mbox{ if }r<p\\
>n&\mbox{ if }r\geq p
\end{array}
}
\right.
\]
In consequence, $x_1~=~x_i+\varphi_p-n$, otherwise we cannot have $x_1\in C(x_i,\underline{t})$. More generally, however,
it must hold that:
\[
\begin{array}{ll}
x_{k}=x_i+\varphi_{p+k-1}-n&\mbox{ $\forall$ }1~\leq~k~\leq~s-p+1\\
x_{k}=x_i+\varphi_{p-(s-k)-1}&\mbox{ $\forall$ }s-p+2~\leq k~\leq s
\end{array}
\]
This, however, corresponds to the behaviour of the IA $\underline{u}$, whose definition is:
\[
\underline{u}~~=~~\tuple{t_{p-1},~t_{p},~\ldots,~t_{p+k},~\ldots,~t_{s-1},~t_0,~\ldots,~t_{p-2}}
\]
Clearly $\underline{u}\approx\underline{t}$ and 
$C(x_1,\underline{u})=C(x_i,\underline{t})$ as claimed.
\end{proof}
\noindent
As an easy consequence of Lemma~\ref{lemma:ordering} we obtain:
\begin{lemma}\label{lemma:no-overlap}
Let $\underline{t}$ and $\underline{u}$ be IAs for which $\underline{t}\not\approx\underline{u}$.
In such cases,
\[
\bigcup_{i=1}^{n}~\left\{{C(i,\underline{t})}\right\}~\bigcap~\bigcup_{i=1}^{n}~\left\{{C(i,\underline{u})}\right\}~~=~~\varnothing
\]
\end{lemma}
\begin{proof}
Suppose the contrary and that $C=\set{x_1,\ldots,x_s}$ can be generated by $C(x_i,\underline{t})$ and $C(x_j,\underline{u})$
for choices of $\underline{t}$ and $\underline{u}$ belonging to different equivalence classes of $\approx$.
As a consequence of Lemma~\ref{lemma:ordering} we know that there are IAs, $\underline{t}'$ and
$\underline{u}'$ for which,
$\underline{t}\approx\underline{t}'$, $\underline{u}\approx\underline{u}'$
and $C(x_1,\underline{t}')=C(x_i,\underline{t})=C(x_j,\underline{u})=C(x_1,\underline{u}')$.
Furthermore, $C(x_1,\underline{t}')$ and $C(x_1,\underline{u}')$ produce the elements of $C$
in increasing ordering of $x_i\in C$. This, however, is only possible if:
\[
x_i~~=~~x_1~+~\sum_{k=0}^{i-2}~t'_k~+~i-1~~=~~x_1~+~\sum_{k=0}^{i-2}~u'_k~+~i-1
\]
that is, $t'_i=u'_i$ for each $0\leq i\leq s-1$. This, however, implies that $\underline{t}\approx\underline{u}$,
in contradiction to our starting premise.
\end{proof}

\subsection{Sub-sequences and Periodic IAs}
\label{sec:periodic}

For any given $s$, the canonical representative IAs can be
constructed from integer partitions of $n - s$ (Definition~\ref{def:increment-array}).
When an IA is periodic, i.e.\ it consists of a shorter sub-sequence repeated multiple times, different starting agents can produce the same coalition.
For example, in Table~\ref{tab:example1}, the periodic IAs $\langle 2,2 \rangle$, $\langle 1,1,1 \rangle$, and $\langle 0,1,0,1 \rangle$ each assign identical coalitions to multiple agents.
This subsection characterises exactly when and why such duplication occurs.
 
\begin{definition}[Period, repetition count, and distinct coalitions]
\label{def:period}
Let $\underline{t} = \langle t_0, t_1, \ldots, t_{s-1} \rangle$ be an
IA of size~$s$.
 
\begin{enumerate}
\item The \emph{period} of $\underline{t}$, denoted
    $\pi(\underline{t})$, is the length of the shortest prefix
    $\langle t_0, \ldots, t_{p-1} \rangle$ such that $\underline{t}$
    consists of $s/p$ identical copies of that prefix:
    \[
        \pi(\underline{t}) =
        \min\bigl\{p \in \{1, \ldots, s\} :
        \underline{t} = \underbrace{
            \langle t_0, \ldots, t_{p-1},\;
            t_0, \ldots, t_{p-1},\;
            \ldots \rangle
        }_{s/p \text{ copies}} \,\bigr\}
    \]

\item The \emph{repetition count} is $\mu = s /\pi(\underline{t})$,
    which corresponds to the number of times the minimal sub-sequence is repeated within
    $\underline{t}$. This is because $\underline{t}$ is composed of identical copies of its minimal prefix, and $\mu \in \mathbb{N}^{+}$.

\item The \emph{stride} of $\underline{t}$ is $\varpi(\underline{t}) = n/\mu = n\,\pi(\underline{t})/s$.

\item An increment array (IA) $\underline{t}$ is \emph{aperiodic} if $\pi(\underline{t}) = s$
    (equivalently, $\mu = 1$), and \emph{periodic} if
    $\pi(\underline{t}) < s$ (equivalently, $\mu > 1$).
\end{enumerate}
\end{definition}

The terminology reflects the fact that, as established in Theorem~\ref{thm:main}, $\varpi(\underline{t})$ equals the number of distinct coalitions generated by $\underline{t}$; exactly $\varpi(\underline{t})$ consecutive agents produce distinct coalitions before the pattern repeats.

The existence of repeating sub-sequences within a periodic IA will result in the assignment of duplicate coalitions to different agents.
For example, $\underline{t} = \langle 1,1,1\rangle$ contains a repeated sub-sequence $\langle 1 \rangle$ which appears $\mu=3$ times, where $\pi(\underline{t}) = 1$, and a stride $\varpi=2$. Therefore, it will result in the assignment of two different coalitions, each being assigned to three agents.\footnote{Table~\ref{tab:example1} lists only the unique coalitions rather than illustrating all duplicated coalitions; this reflects the designation strategy discussed in Section~\ref{sec:designation}.} For example, the coalitions generated for $C(1,\underline{t})$, $C(3,\underline{t})$, and $C(5,\underline{t})$ all contain the same agents, albeit in a different order.

Theorem~\ref{thm:main} shows that exactly $\varpi(\underline{t})$ distinct coalitions are generated by $\underline{t}$, and that $C(i,\underline{t})=C(j,\underline{t})$ if and only if $i \equiv j \pmod{\varpi(\underline{t})}$. Thus, each IA can be used $r$ times to generate the same coalition; i.e.\ the coalitions formed for agents $i$ and $j$ are equivalent if and only if there exists some $r \in \mathbb{N}_0$ within the stated range.
\begin{theorem}\label{thm:main}
For any increment array (IA) $\underline{t}$, and for all $1\leq i\leq j\leq n$,
\begin{equation}\label{eqn:thm1}
C(i,\underline{t})=C(j,\underline{t})~~~\Leftrightarrow~~~
\exists~0\leq r\leq \frac{(n-i)s}{n\pi(\underline{t})}~:~
j~=~i~+~r\times\left({\frac{n\pi(\underline{t})}{s}}\right)
\end{equation}
\end{theorem}
\begin{proof}[Proof sketch (full proof in Appendix~\ref{sec:appendix})]

$(\Leftarrow)$\enspace
We show that $j = i + r \cdot (n\pi(\underline{t})/s)$ implies
$C(i,\underline{t}) = C(j,\underline{t})$.  The key observation is that
periodicity of $\underline{t}$ yields a closed-form for the cumulative
integer increments beyond one period:
\[
  \varphi_{\pi(\underline{t})+k}
    \;=\; \frac{n\pi(\underline{t})}{s} \;+\; \varphi_k
    \qquad (1 \le k \le s - \pi(\underline{t}))
\]
Thus, advancing the starting agent by $n\pi(\underline{t})/s$ shifts
every cumulative increment by the same additive constant.
Consequently the (uncorrected) sequence of coalition members generated
from~$j$ is a cyclic rearrangement of that generated from~$i$, so
after reduction modulo~$n$ the two coalitions coincide.  The general
case follows by applying this shift $r$~times.

\medskip\noindent
$(\Rightarrow)$\enspace
Suppose $C(i,\underline{t}) = C(i + \varphi_\rho, \underline{t})$ for
some $2 \le \rho \le s$.  Matching the two sorted coalition
representations element-by-element produces a system of identities on
the entries of~$\underline{t}$ (Equations \eqref{eqn:appx3} and \eqref{eqn:appx4} in the Appendix).
Analysing these identities shows that $\underline{t}$ must repeat under
a cyclic shift of $(\rho-1)$ positions, i.e.\
$t_j = t_{j+(\rho-1) \bmod s}$ for all~$j$.  It follows that
$\pi(\underline{t})$ divides $\gcd(s, \rho-1)$, and hence that
$\underline{t}$ repeats in blocks of size $g = \gcd(s,\rho-1)$.
Substituting back, the starting-agent offset $\varphi_\rho$ evaluates
to $(\rho-1)\,n/s$, which is an integer multiple of
$n\pi(\underline{t})/s$.  This establishes that $j$ has the required
form $i + r \cdot (n\pi(\underline{t})/s)$.  The upper bound on~$r$
follows from $j \le n$.
\end{proof}
\noindent
The coalitions formed for agents $i$ and $j$ are equal if and only if there exists some 
$r$ in the range stated in Equation \eqref{eqn:thm1}.  
This $r$ represents \emph{the number of distinct agents} that produce duplicate coalitions when using the same $\underline{t}$.
More precisely, $j~=~i~+~r\times(n\pi(\underline{t})/s)$ indicates exactly which other starting agents will generate the same coalition as agent $i$.
Thus, in the trivial case, if $r=0$, then $i \equiv j$.
The key quantity driving this is $\mu = s/\pi(\underline{t})$; the number of times the repeating block occurs within $\underline{t}$, such that:
\begin{itemize}
    \item If $\underline{t}$ is \emph{aperiodic} (i.e.\ $\pi(\underline{t})=s$, so that $\mu = 1$), then $n\pi(\underline{t})/s = n$, and the only solution is $r=0$, i.e.\ $j=i$.
    Every agent will produce a unique coalition, and no duplicates occur.
    \item If $\underline{t}$ is \emph{periodic} with $\mu > 1$ repeating subsequences, then $n\pi(\underline{t})/s = n/\mu < n$, and exactly $\mu$  agents will be spaced evenly around the cycle of $n$ agents that all generate the same coalition.  In this case, $r$ can take the values $0,1,\ldots,\mu -1$, giving exactly $\mu$ duplicate-generating agents.
\end{itemize}

Since $\mu$ divides both $s$ (by definition) and $n$ (because the offset increments within each repeated block must sum to a whole number), $\mu$ can be at most $\gcd(n, s)$. Therefore, when $n$ and $s$ share no common factor (i.e.\ $\gcd(n, s) = 1$) then $\mu = 1$, and every IA of that size $s$ is aperiodic.

Continuing from our previous example where the periodic IA $\underline{t} = \langle 1,1,1 \rangle$ and $s=3$ (Table~\ref{tab:example1}),
we have $\mu=3$ repeated sequences with a period $\pi(\underline{t}) = 1$.  If we consider
the case where $i=1$, then the upper bound for $r$ is:
\[
\frac{(n-i)s}{n\pi(\underline{t})} = \frac{(6-1) \times 3}{6\times 1} =  \frac{15}{6}
\]
and therefore $\exists~0\leq r\leq 15/6$; i.e.\ $r\in \set{0,1,2}$. Furthermore, as the stride is $\varpi{\underline{t}}=n\pi(\underline{t})/s=(6 \times 1)/3=2$, then this periodic IA will generate only 2 unique coalitions. 
From this, given that $C(1,\underline{t}) = \set{1,3,5}$, we have:
\begin{align*}
r=0, &~~~j=i+0=1 & C(j,\underline{t}) = C(1,\tuple{1,1,1}) = \set{1,3,5} \\
r=1, &~~~j=i+1\times 2=3 & C(j,\underline{t}) = C(3,\tuple{1,1,1}) = \set{3,5,1} \\
r=2, &~~~j=i+2\times 2=5 & C(j,\underline{t}) = C(5,\tuple{1,1,1}) = \set{5,1,3} 
\end{align*}

Theorem~\ref{thm:main2} shows that IAs of the same equivalence class generate identical sets of coalitions, and that this is the only way that identical coalitions can be formed.

\begin{theorem}\label{thm:main2}
For any increment array (IA) $\underline{t}$ and $\underline{u}$:
\begin{eqnarray}
\bigcup_{i=1}^{n}~\left\{{C(i,\underline{t})}\right\}~\bigcap~\bigcup_{i=1}^{n}~\left\{{C(i,\underline{u})}\right\}~~\not=~~\varnothing
&\mbox{ $\Leftrightarrow$ }&\underline{t}\approx\underline{u} \label{eqn:main2a}\\
\bigcup_{i=1}^{n}~\left\{{C(i,\underline{t})}\right\}~~=~~\bigcup_{i=1}^{n}~\left\{{C(i,\underline{u})}\right\}&
\mbox{ $\Leftrightarrow$ }& \underline{t}\approx\underline{u} \label{eqn:main2b}
\end{eqnarray}
\end{theorem}
\begin{proof}
The fact that IAs belonging to the same 
equivalence class of $\approx$ generate identical sets of coalitions
of size $s$ over all agents (and hence have a non-empty intersection)
follows directly from Lemma~\ref{lemma:same-sets}.
This establishes the $\Leftarrow$ implication for Equation \eqref{eqn:main2b}
and hence \eqref{eqn:main2a}. This is the only
way that two increment arrays can generate a coalition
in common, i.e.\ that should $\underline{t}$ and $\underline{u}$
produce the same subset of $\set{1,2,\ldots,n}$ then
$\underline{t}$ and $\underline{u}$ belong to the same
equivalence class of $\approx$ is immediate from 
Lemma~\ref{lemma:ordering} and Lemma~\ref{lemma:no-overlap}.
We thus have the $\Rightarrow$ implications of Equations \eqref{eqn:main2a}
and \eqref{eqn:main2b} thereby completing the theorem proof.
\end{proof}
\noindent
When a periodic IA $\underline{t}$ is encountered, only one of the $\mu$ 
duplicate-generating agents should actually be assigned that coalition. This is discussed in the next section.

\subsection{The Designation Problem} \label{sec:designation}

Theorem~6 establishes that when a canonical IA $\underline{t}$ is periodic (i.e.\ $\pi(\underline{t}) < s$), exactly $\mu = s / \pi(\underline{t})$ agents generate the same coalition as each other, such that the agents that produce duplicate coalitions are spaced exactly $\varpi(\underline{t}) = n / \mu$ apart; i.e.\ agents $x$ and $y$ generate the same coalition, if and only if $x \equiv y \pmod{\varpi(\underline{t})}$.
This partitions the $n$ agents into $\varpi(\underline{t})$ groups, each of
size $\mu$, called the \emph{residue classes} modulo
$\varpi(\underline{t})$.  

Agents within the same residue class are
redundant (i.e.\ they would all compute the same coalition value, resulting in duplicate coalitions) while agents in
different residue classes produce distinct coalitions.
For example, with $n = 6$ and the periodic IA
$\underline{t} = \langle 1, 1, 1 \rangle$ ($\pi = 1$, $\mu = 3$,
$\varpi = 2$), the two residue classes modulo 2 are
$\{1, 3, 5\}$ and $\{2, 4, 6\}$.  Within each class, all three agents
generate the same coalition (e.g.\ agents 1, 3, and 5 all generate
$\{1, 3, 5\}$), so only one agent per class is needed.

To avoid redundant computation, exactly one agent from each residue class
should be assigned the corresponding coalition.  A set that contains
exactly one representative from every residue class is called a
\emph{transversal}.  In the example above, any pair such as $\{1, 2\}$,
$\{3, 4\}$, or $\{5, 6\}$ would be a valid transversal; i.e.\ each pair picks one agent from each of the two residue classes.  

The \emph{designation problem} is therefore to choose, for each periodic IA, a transversal of $\varpi(\underline{t})$ agents to be responsible for computing the
corresponding coalition values.\footnote{The periodic IA $\underline{t} = \langle 1, 1, 1 \rangle$ in Table~\ref{tab:example1} is designated to the agents in the transversal $\{4,5\}$.}
For aperiodic IAs ($\pi(\underline{t}) = s$, $\mu = 1$), no designation
is needed: every agent produces a unique coalition, and all agents are
implicitly designated. 
We refer to the agents selected for a periodic IA as the \emph{designated}
agents, whereas the resulting assignment of coalitions to agents across
all IAs and all sizes is the \emph{allocation}.
The challenge therefore lies entirely with the periodic IAs.

Let $\mathcal{E}(n,s)$ denote the set of canonical representative IAs
for coalition size $s$ (one per equivalence class under $\approx$).  We
partition $\mathcal{E}(n,s)$ into two disjoint subsets:
\begin{align}
    \mathcal{A}(n,s) &= \{\, \underline{t} \in \mathcal{E}(n,s)
        : \pi(\underline{t}) = s \,\} \label{eq:aperiodic-set} \\
    \mathcal{P}(n,s) &= \{\, \underline{t} \in \mathcal{E}(n,s)
        : \pi(\underline{t}) < s \,\} \label{eq:periodic-set}
\end{align}
so that $\mathcal{E}(n,s) = \mathcal{A}(n,s) \;\dot\cup\;
\mathcal{P}(n,s)$.  The aperiodic IAs $\mathcal{A}(n,s)$ generate $n$
unique coalitions each, whereas each periodic IA $\underline{t} \in
\mathcal{P}(n,s)$ generates only $\varpi(\underline{t}) = n/\mu$
unique coalitions.
The total number of coalitions of size $s$ given $n$ agents can therefore be expressed as:
\begin{equation}
\label{eq:coalition-decomposition}
    \binom{n}{s} \;=\; n \cdot |\mathcal{A}(n,s)|
    \;+\; \sum_{\underline{t}\, \in\, \mathcal{P}(n,s)}
        \varpi(\underline{t})
\end{equation}

\subsubsection{The Lowest-ID Designation Scheme}

The simplest approach is to designate the agents with the lowest
identifiers.  For a periodic IA $\underline{t}$ with
$d = \varpi(\underline{t}) = n / \mu$, this scheme
designates agents $\{1, 2, \ldots, d\}$; i.e.\ a transversal of the $d$ residue
classes modulo $d$, since these $d$ consecutive agents each fall in a
distinct residue class.

Under this scheme, the allocation to agent $x$ for coalition size $s$ is:
\begin{equation}
\label{eq:lowest-id-count}
    |CV_x^s| \;=\; |\mathcal{A}(n,s)| \;+\;
    \bigl|\{\, \underline{t} \in \mathcal{P}(n,s) :
    x \le \varpi(\underline{t}) \,\}\bigr|
\end{equation}
Agent 1 is designated for every periodic IA (since $1 \le d$ for all $d \ge 1$), receiving the maximum allocation of $|\mathcal{E}(n,s)|$ coalitions.  
Agent $n$ is designated only for aperiodic IAs, receiving the minimum of $|\mathcal{A}(n,s)|$ coalitions.

Continuing with our running example, if we have $\underline{t} = \langle 1,1,1 \rangle, s=3, n=6$, then $d=2$ and the designated agents are $\{1,2\}$; i.e.\ the two lowest numbered agents, with agent 1 computing $C(1, \langle 1,1,1 \rangle) = \{1,3,5\}$ and agent 2 computing $C(2, \langle 1,1,1 \rangle) = \{2,4,6\}$.  Thus agents 1 and 2 would receive 4 coalitions each while agents 3--6 would each receive only 3.

More generally, under the lowest-ID scheme, the maximum imbalance at any coalition size~$s$ is $|\mathcal{P}(n,s)|$; agent~1 is designated for every periodic IA, whereas agent~$n$ is designated for none.
As $n$ grows and periodic IAs become more numerous, this imbalance can become substantial.  The rotated designation scheme introduced next reduces the worst-case imbalance to at most 1.

\subsubsection{The Rotated Designation Scheme} \label{sec:rotated-designation}

To achieve a tighter load balance, we introduce a rotated designation scheme that distributes the periodic-IA designations evenly across all agents.
The key idea is to shift the designation window cyclically for each successive periodic IA, so that no single agent is systematically favoured.
Crucially, the scheme requires no central coordinator: each agent can independently determine its own designations using only its identifier $x$ and the shared parameter $n$.

\begin{definition}[Rotated designation]
\label{def:rotated-designation}
For a given coalition size $s$, the periodic canonical IAs are enumerated
in the deterministic order produced by a shared generation algorithm\footnote{The specific algorithm used to generate the IAs in a deterministic ordering is presented in Section~\ref{sec:algorithm}. For the purposes of this section, it suffices that all agents use the same deterministic ordering.} as
$\underline{t}_1, \underline{t}_2, \ldots, \underline{t}_K$
where $K = |\mathcal{P}(n,s)|$.  For each $\underline{t}_j$, let
$d_j = \varpi(\underline{t}_j) = n / \mu(\underline{t}_j)$ denote the
number of distinct coalitions it generates.  The cumulative offset $h_j$ is defined as:
\begin{equation}
\label{eq:cumulative-offset}
    h_j = \sum_{i=1}^{j-1} d_i, \qquad h_1 = 0
\end{equation}
An agent $x$ determines that it is \emph{designated} for $\underline{t}_j$
if and only if
\begin{equation}
\label{eq:designation-test}
    (x - 1 - h_j) \bmod n \;<\; d_j
\end{equation}
Since the enumeration order and the values $d_j$ are determined entirely by $n$ and $s$, every agent evaluates this test identically.
The set of agents satisfying~(\ref{eq:designation-test}) for a given $\underline{t}_j$ is a window of $d_j$ consecutive agents (cyclically in $\{1, \ldots, n\}$) starting at agent $(h_j \bmod n) + 1$; i.e.\ 
\begin{equation}
\label{eq:designation-window}
    D_j \;=\; \bigl\{\, (h_j + k) \bmod n + 1 : 0 \le k \le d_j - 1
    \,\bigr\}
\end{equation}
\end{definition}

Following our example where $n=6, s=3$, assume that there is one periodic
canonical IA at this size: $\underline{t}_1 = \langle 1, 1, 1 \rangle$
with $\pi = 1$, $\mu = 3$, and $d_1 = \varpi(\underline{t}_1) = 2$.
Since $K = 1$, the cumulative offset is $h_1 = 0$, and the designation
window starts at agent $(0 \bmod 6) + 1 = 1$.  Applying the designation
test~(\ref{eq:designation-test}) to each agent:
\[
\begin{array}{c|cccccc}
    x & 1 & 2 & 3 & 4 & 5 & 6 \\
    \hline
    (x {-} 1 {-} 0) \bmod 6 & 0 & 1 & 2 & 3 & 4 & 5 \\
    {<}\; d_1 = 2\text{?}
    & \checkmark & \checkmark & & & &
\end{array}
\]
Thus $D_1 = \{1, 2\}$ is the set of two consecutive agents, one from each of the two
residue classes $\{1, 3, 5\}$ and $\{2, 4, 6\}$.  Agent~1 computes
$C(1, \underline{t}_1) = \{1, 3, 5\}$ and agent~2 computes
$C(2, \underline{t}_1) = \{2, 4, 6\}$, covering all
$\varpi(\underline{t}_1) = 2$ distinct coalitions generated by this IA.
The remaining agents 3--6 are not designated for $\underline{t}_1$ and
need not evaluate it.  When multiple periodic IAs are present ($K > 1$),
the cumulative offset $h$ advances after each one, rotating the window
so that the next IA's designations fall on different agents; this is
illustrated for all four periodic IAs of $n = 6$ in the worked example
of Section~\ref{sec:example-n6}.

The designation window $D_j$ advances by $d_j$ positions after each
periodic IA, causing successive windows to rotate around the cycle of $n$
agents.  Since aperiodic IAs do not modify the offset (all agents are
designated for them by default), $h$ is advanced only by periodic IAs.
 
\begin{lemma}[Transversal validity]
\label{lem:transversal}
Let $d \mid n$ (i.e.\ $d$ is a positive divisor of $n$).
Any set of $d$ agents with consecutive identifiers
(modulo $n$) forms a transversal of the $d$ residue classes of
$\{1, \ldots, n\}$ modulo $d$.
\end{lemma}
 
\begin{proof}
Using $0$-indexed agents $\{0, 1, \ldots, n{-}1\}$ for clarity, the
residue classes modulo $d$ are $R_j = \{j, j{+}d, j{+}2d, \ldots\}$
for $j = 0, \ldots, d{-}1$.  Consider $d$ consecutive agents starting at
position $h$: agent $(h+k) \bmod n$ for $k = 0, \ldots, d{-}1$.  Their
residues modulo $d$ are $(h+k) \bmod d$ for $k = 0, \ldots, d{-}1$,
which is a permutation of $\{0, 1, \ldots, d{-}1\}$.  Hence exactly one
agent from each residue class is included.
\end{proof}
\noindent 
Since $\mu(\underline{t}) \mid n$ (as established by Definition~\ref{def:period}, where
$\varpi(\underline{t}) = n\pi(\underline{t})/s$ must be a positive
integer), we have $d_j = \varpi(\underline{t}_j)$ divides $n$ for every
periodic IA.
Lemma~\ref{lem:transversal} therefore guarantees that each designation
$D_j$ is a valid transversal, ensuring that every unique coalition from
$\underline{t}_j$ is computed by exactly one agent and no coalition is
missed.

\subsubsection{Balancing the Coalition Value Calculation load across agents} \label{sec:loadbalancing}

Using the rotated designation scheme (Definition~\ref{def:rotated-designation}), we can show that the allocation of coalitions, including those from periodic IAs, will differ by no more than 1 across all of the agents. 

\begin{theorem}[Per-size load balance]
\label{thm:persize-balance}
Under the rotated designation scheme, for any fixed coalition size $s$ and
any agent $x \in \{1, \ldots, n\}$:
\begin{equation}
\label{eq:persize-bound}
    \left\lfloor \frac{\binom{n}{s}}{n} \right\rfloor
    \;\le\; |CV_x^s| \;\le\;
    \left\lceil \frac{\binom{n}{s}}{n} \right\rceil
\end{equation}
The maximum difference in allocation between any two agents for a given
coalition size is therefore at most 1.
\end{theorem}
\begin{proof}
Each agent receives exactly $|\mathcal{A}(n,s)|$ coalitions from aperiodic
IAs, since every aperiodic IA generates $n$ distinct coalitions (one per
agent).
For the periodic IAs, define the total number of periodic designations at
size $s$:
\[
    D(n,s) = \sum_{j=1}^{K} d_j
    = \sum_{\underline{t}\, \in\, \mathcal{P}(n,s)} \varpi(\underline{t})
\]
where $K = |\mathcal{P}(n,s)|$ is the number of periodic canonical IAs at size $s$.
Under the rotated scheme, the $D(n,s)$ designation slots are laid out
consecutively (i.e.\ the first $d_1$ slots for $\underline{t}_1$, the next
$d_2$ slots for $\underline{t}_2$, and so on) and mapped to agents
cyclically; the $p$-th slot (for $p = 1, \ldots, D(n,s)$) is assigned to
agent $((h_1 + p - 1) \bmod n) + 1$.  Therefore, the number of slots assigned to any
agent is either $\lfloor D(n,s) / n \rfloor$ or
$\lceil D(n,s) / n \rceil$.
Combining, the total allocation for agent $x$ at size $s$ gives us:
\[
    |CV_x^s| = |\mathcal{A}(n,s)| + \delta_x
\]
where $\delta_x \in \{\lfloor D(n,s)/n \rfloor,\;
\lceil D(n,s)/n \rceil\}$.  From (\ref{eq:coalition-decomposition}),
$\binom{n}{s} = n \cdot |\mathcal{A}(n,s)| + D(n,s)$, and therefore:
\[
    \frac{\binom{n}{s}}{n} = |\mathcal{A}(n,s)| + \frac{D(n,s)}{n}
\]
Since $|\mathcal{A}(n,s)|$ is an integer, the bound (\ref{eq:persize-bound}) follows.
\end{proof}

\noindent
Note that when $\gcd(n,s) = 1$, there are no periodic IAs at size~$s$ (as established in Section~\ref{sec:periodic}); thus
$\mathcal{P}(n,s) = \varnothing$ and every agent receives exactly $|\mathcal{A}(n,s)| = |\mathcal{E}(n,s)|$ coalitions.  
In this case the allocation is \emph{perfectly} balanced, with zero imbalance.  
The bound of Theorem~\ref{thm:persize-balance} is therefore tight only when periodic IAs exist.

Although Theorem~\ref{thm:persize-balance} guarantees a per-size imbalance of at most 1, the aggregate imbalance across all coalition sizes can be as large as the number of sizes containing periodic IAs if the rotated designation scheme is applied independently for each coalition size; i.e.\ if the cumulative offset $h$ is reset to 0 for each new size $s$.
This occurs because the same agents may receive the ``extra'' periodic coalition at every size when the designation window always starts from the same position.

This can be avoided by a simple modification: rather than resetting $h$ at each new coalition size, the offset is \emph{carried across all sizes}; i.e.\ 
after processing all periodic IAs at size $s$,
the cumulative offset becomes the initial offset for size $s + 1$.  
We refer to this as the \emph{global offset} variant.

Formally, let the periodic canonical IAs across \emph{all} sizes
$s = 1, \ldots, n$ be indexed in the order they are encountered as
$\underline{t}^{(1)}, \underline{t}^{(2)}, \ldots, \underline{t}^{(M)}$,
where $M = \sum_{s=1}^{n} |\mathcal{P}(n,s)|$ is the total number of
periodic canonical IAs.  Each $\underline{t}^{(j)}$ generates
$d^{(j)} = \varpi(\underline{t}^{(j)}) = n / \mu(\underline{t}^{(j)})$
distinct coalitions.  The global cumulative offset
is:
\begin{equation}
\label{eq:global-offset}
    H_j = \sum_{i=1}^{j-1} d^{(i)}, \qquad H_1 = 0
\end{equation}
and the designation test (Equation \eqref{eq:designation-test}) is applied with
$H_j$ in place of $h_j$.

\begin{theorem}[Aggregate load balance]
\label{thm:aggregate-balance}
Under the rotated designation scheme with global offset, for any agent
$x \in \{1, \ldots, n\}$:
\begin{equation}
\label{eq:aggregate-bound}
    \left\lfloor \frac{2^n - 1}{n} \right\rfloor
    \;\le\; |CV_x| \;\le\;
    \left\lceil \frac{2^n - 1}{n} \right\rceil
\end{equation}
The maximum difference in the total number of coalition value calculations
between any two agents, across all coalition sizes, is at most 1.
\end{theorem}

\begin{proof}
Every agent receives exactly $A = \sum_{s=1}^{n} |\mathcal{A}(n,s)|$
coalitions from aperiodic IAs (this quantity is the same for all agents).
The total number of periodic designation slots across all sizes is
$D_{\mathrm{global}} = \sum_{j=1}^{M} d^{(j)}$, where $M$ is the total number of periodic canonical IAs across all sizes, and $d^{(j)}$
is the number of distinct coalitions generated by the $j$-th periodic IA.

With the global offset, the designation windows for all $M$ periodic IAs
are laid out as a single consecutive sequence of $D_{\mathrm{global}}$
slots, mapped cyclically to the $n$ agents.  By the same counting argument
as in the proof of Theorem~\ref{thm:persize-balance}, each agent receives
either $\lfloor D_{\mathrm{global}} / n \rfloor$ or
$\lceil D_{\mathrm{global}} / n \rceil$ periodic coalitions.

Since $2^n - 1 = n \cdot A + D_{\mathrm{global}}$, we have:
\[
    |CV_x| = A + \delta_x,
    \qquad
    \frac{2^n - 1}{n} = A + \frac{D_{\mathrm{global}}}{n}
\]
where $\delta_x \in \{\lfloor D_{\mathrm{global}} / n \rfloor,\;
\lceil D_{\mathrm{global}} / n \rceil\}$ is the number of periodic
designations assigned to agent $x$.
The bound (\ref{eq:aggregate-bound}) follows as before.
\end{proof}
\noindent
Interestingly, the per-size bound of Theorem~\ref{thm:persize-balance} continues to hold
under the global offset.  The starting offset at each size may differ from
zero, but this merely shifts the designation windows without affecting the
counting argument: the number of periodic designations received by any
agent within a single size $s$ still lies between
$\lfloor D(n,s)/n \rfloor$ and $\lceil D(n,s)/n \rceil$.

\subsubsection{Full worked example for \emph{n}~=~6}
\label{sec:example-n6}

The rotated designation scheme with global offset is illustrated for the case where $n=6$.  This complements Table~\ref{tab:example1}, which lists the coalitions assigned to each agent for both periodic and aperiodic IAs for each size $s$. As the total number of coalitions is $2^6 - 1 = 63$, each agent should receive either $\lfloor 63/6 \rfloor = 10$ or $\lceil 63/6 \rceil = 11$ coalitions.

Recall that each agent will be allocated a coalition for each aperiodic canonical IA.
Therefore, we focus on the designation of the subset of coalitions resulting from periodic canonical IAs across all coalition sizes $s$ in increasing order.
These IAs are illustrated in Table~\ref{tab:designation-n6} in the order they are encountered, together with the repetition count $\mu$, the number of distinct coalitions generated $d = \varpi(\underline{t})$, the global designation offset $h$ at the point of encounter, and the resulting designation window.

\begin{table}[t]
\setlength{\tabcolsep}{12pt}
\centering
\caption{Periodic canonical IAs for $n = 6$ across all coalition sizes,
with the global offset $h$ and designated agents under the rotated scheme.}
\label{tab:designation-n6}
\begin{tabular}{@{}cclccc@{}}
\toprule
Size $s$ & IA $\underline{t}$ & $\mu$ & $d$ & $h$ & Designated agents \\
\midrule
2 & $\langle 2, 2 \rangle$ & 2 & 3 & 0
  & $\{1, 2, 3\}$ \\
3 & $\langle 1, 1, 1 \rangle$ & 3 & 2 & 3
  & $\{4, 5\}$ \\
4 & $\langle 0, 1, 0, 1 \rangle$ & 2 & 3 & 5
  & $\{6, 1, 2\}$ \\
6 & $\langle 0, 0, 0, 0, 0, 0 \rangle$ & 6 & 1 & 8
  & $\{3\}$ \\
\bottomrule
\end{tabular}
\end{table}

As the coalition sizes $s = 1$ and $s = 5$ contain no periodic IAs (since $\gcd(n, s) = 1$ in both cases), the offset passes through them
unchanged.  The global offset $h$ advances through the following values: $0 \to 3 \to 5 \to 8 \to 9$ across the four periodic IAs, causing each designation window to fall on a different set of agents. Table~\ref{tab:allocation-n6} illustrates the number of resulting allocations to each agent for differing coalition sizes, together with the total number of coalitions (as shown in Table~\ref{tab:example1}), and the theoretical bound.  Note that the designations for periodic IAs appear in Table~\ref{tab:example1} in bold (with the exception of the grand coalition; i.e.\ when $s=6$, which is not shown, but in this example is designated to agent 3).

\begin{table}[th]
\centering
\setlength{\tabcolsep}{9pt}
\caption{Per-agent allocation for $n = 6$ under the rotated designation scheme with global offset for each coalition size $s$.}
\label{tab:allocation-n6}
\begin{tabular}{@{}cccccccccc@{}}
\toprule
& \multicolumn{6}{c}{Number of Coalitions at size $s$} & \\
\cmidrule(lr){2-7}
Agent $x$ & 1 & 2 & 3 & 4 & 5 & 6 & Total $|CV_x|$
& Bound \\
\midrule
1 & 1 & 3 & 3 & 3 & 1 & 0 & 11
& $\lceil 63/6 \rceil = 11$ \;\checkmark \\
2 & 1 & 3 & 3 & 3 & 1 & 0 & 11
& $11$ \;\checkmark \\
3 & 1 & 3 & 3 & 2 & 1 & 1 & 11
& $11$ \;\checkmark \\
4 & 1 & 2 & 4 & 2 & 1 & 0 & 10
& $\lfloor 63/6 \rfloor = 10$ \;\checkmark \\
5 & 1 & 2 & 4 & 2 & 1 & 0 & 10
& $10$ \;\checkmark \\
6 & 1 & 2 & 3 & 3 & 1 & 0 & 10
& $10$ \;\checkmark \\
\midrule
\multicolumn{7}{r}{Total:} & \multicolumn{2}{l}{$3 \times 11 + 3 \times 10
= 63 = 2^6 - 1$ \;\checkmark} \\
\bottomrule
\end{tabular}
\end{table}

The aggregate imbalance is $11 - 10 = 1$, and the per-size imbalance is at most 1 for every coalition size (e.g.\ at $s = 2$: agents receive either 2 or 3; at $s = 3$: either 3 or 4).  Note that the per-size allocation varies across agents (e.g.\ agent 3 receives 3 coalitions for $s = 3$ but only 1 for $s = 6$, while agent 4 receives 4 at $s = 3$ and none for $s = 6$), but the totals are balanced to within 1.
For comparison, with the lowest-ID designation the totals for increasing sizes $s$ would be 13, 12, 11, 9, 9, 9; i.e.\ an imbalance of 4.  The rotated scheme with global offset eliminates this bias entirely.

\subsection{Discussion}
\label{sec:IA-discussion}

The use of increment arrays (IAs) address each of the desirable properties outlined in Section~\ref{sec:intro} for coalition value calculation allocation methods, which are discussed below:

\begin{enumerate}
\item \textbf{Eliminating the need for communication between agents:}
    As noted in Definition~\ref{def:rotated-designation}, the designation test~(\ref{eq:designation-test}) depends only on the agent's own identifier $x$ and the shared parameter $n$.
    The cumulative offset $h$ and the strides $d_j$ are functions of the canonical IAs, which every agent generates in the same deterministic order.
    Consequently, no inter-agent communication, shared memory, or external coordination is required, as each agent independently and identically determines its own designations across all coalition sizes.

\item \textbf{Ensuring that the allocation of coalitions is equitable:}
    Under the rotated designation scheme (Definition~\ref{def:rotated-designation}), each agent is designated for all aperiodic IAs and for an approximately equal share of the periodic IAs for each coalition size (Theorem~\ref{thm:persize-balance}).
    When the global offset variant is used, the aggregate allocation across all coalition sizes is also balanced: the total number of coalition value calculations assigned to any agent lies between $\lfloor (2^n - 1)/n \rfloor$ and $\lceil (2^n - 1)/n \rceil$ (Theorem~\ref{thm:aggregate-balance}).
    Therefore, the maximum difference between any two agents is at most 1.
    
\item \textbf{Every possible coalition value is calculated once and only
    once, thus eliminating redundancy:}
    Each coalition of size~$s$ is generated by exactly one equivalence class of IA (Theorem~\ref{thm:main2}).
    For aperiodic IAs ($\pi(\underline{t}) = s$), every starting agent produces a distinct coalition, and thus all $n$ coalitions in that class are covered without duplication.
    For periodic IAs ($\pi(\underline{t}) < s$), as the $\mu$ agents within each residue class will generate the same coalition; the designation scheme assigns exactly one agent per residue class (i.e.\ a transversal of size $\varpi(\underline{t})$) so that each distinct coalition is computed only once.
    Taken together, every coalition $C \subseteq \{1,\ldots,n\}$ with $|C| = s$ is computed by exactly one designated agent (Lemma~\ref{lemma:ns-sequence} guarantees that a generating IA exists
    for every such coalition).

\item \textbf{Maintaining a balanced load across agents:}
    Beyond receiving approximately equal \emph{numbers} of coalitions
    (property~(ii)), it is also desirable that agents perform
    approximately equal amounts of \emph{computation}.
    Every agent generates coalitions from the same set of canonical
    representative IAs (one per equivalence class under $\approx$; see
    Section~\ref{sec:equivalence}).
    For each IA $\underline{t}$ of size~$s$, the coalition generation
    function $C(x, \underline{t})$ (Definition~\ref{def:coalition-generation})
    computes $s - 1$ cumulative additions (one per element of the IA),
    regardless of the starting agent~$x$.
    Consequently, every agent that is designated for the same set of IAs
    performs \emph{exactly} the same number of arithmetic operations to
    generate its coalitions; there is no dependence on agent identity or
    on the lexicographic position of the coalitions within a list.
    This should be contrasted with the cost of \emph{evaluating} the
    characteristic function $\nu(C)$ for each assigned coalition, which
    is application-dependent and may vary with coalition size.
    The per-size balance of Theorem~\ref{thm:persize-balance} ensures
    that, within any given size~$s$, allocations differ by at most one
    coalition, so that if the cost of evaluating $\nu(C)$ depends only
    on $|C|$, the total evaluation cost across agents is also balanced.
    
\item \textbf{An agent is \emph{self-interested}; i.e.\ it is only assigned
coalitions in which it is a member:}
    The coalition generation function $C(x, \underline{t})$ guarantees
    this by construction; the starting agent~$x$ is always included in
    the coalition (Equation~\eqref{eqn:valuex}).
    Furthermore, since the designation schemes only select \emph{which}
    agents are designated, and a designated agent~$x$ always computes
    $C(x, \underline{t})$, the self-interest property is preserved
    regardless of the designation scheme used.

\end{enumerate}


\section{Necklace-based Distributed Coalition Algorithm} \label{sec:algorithm}

The preceding section established that a minimal coalition value
calculation set $CVCS^*$ can be constructed from canonical
representative IAs, one per equivalence class under
$\approx$, together with a rotated designation scheme that
balances the workload across agents.  Two questions remain: how to
generate the canonical IAs efficiently, and how to combine
generation with the designation scheme into a complete,
communication-free algorithm.  To address these, we
first show that canonical representative IAs correspond exactly
to two-colour combinatorial necklaces
(Section~\ref{sec:bijection}), so that necklace generation
algorithms can be used to enumerate them.  We also describe the
FKM algorithm for generating necklaces in constant amortised time, before presenting the full
\emph{Necklace-based Distributed Coalition Algorithm (N-DCA)}
that each agent executes independently to
determine its own coalition value calculation allocation
(Section~\ref{sec:ndcva_algorithm}).  We conclude the discussion with a
complexity analysis in Section~\ref{sec:complexity}.

\subsection{The Necklace--Increment Array Correspondence}
\label{sec:bijection}

An IA for a coalition of size $s$ given $n$ agents can be
generated by considering the different combinations of a
two-colour necklace of length $n$ with $s$ white beads.  We
assume that black beads are represented by the value $1$,
indicating that an agent has been omitted from a coalition,
whereas white beads are represented by the value $0$,
representing agents that are included in a coalition
(Figure~\ref{fig:necklace}).

\begin{figure}[h]
    \centering
    \includegraphics[width=\linewidth]{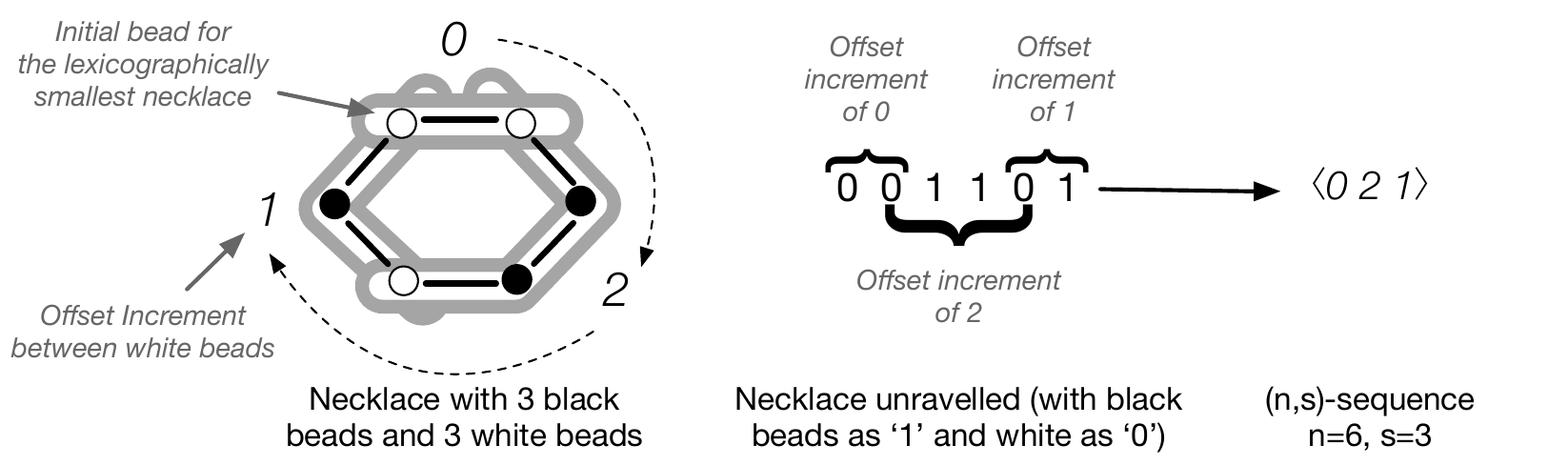}
    \caption{Necklace with $n=6$ beads, of which $s=3$ represent agents in a coalition.}
	\label{fig:necklace}
\end{figure}

\begin{figure}[b]
\scriptsize
\begin{algorithmic}[1]
\Procedure{GenIncArray}{$n$, $a$} \Comment{Generate Increment Array $\underline{t}$ from necklace $a$}
\State \textbf{precondition:} $a[1] = 0$ \Comment{guaranteed by FKM (Figure~\ref{algorithm:fkm})}
 
\State $j \gets 0$ \Comment{counts contiguous `1' elements}
\State $s \gets 0$ \Comment{number of white beads encountered}
\State $\underline{t} \gets \langle\rangle$ \Comment{initially empty}
 
\For{$i=1$ to $n$}
    \If{$a[i] = 1$}
        \State{$j \gets j+1$} \Comment{increment black-bead count}
    \Else
        \If{$s > 0$} \Comment{skip the first white bead}
            \State append $j$ to $\underline{t}$
        \EndIf
        \State $j \gets 0$ \Comment{reset count}
        \State $s \gets s+1$
    \EndIf
\EndFor
 
\If{$s = 0$}
    \State \Return $\underline{t}$ \Comment{all-black necklace: empty IA}
\EndIf
 
\State append $j$ to $\underline{t}$ \Comment{wrap-around: trailing blacks after last white}
 
\State \Return $\underline{t}$
\EndProcedure
\end{algorithmic}
\caption{Algorithm to generate an Increment Array $\underline{t}$
of size $s$ from a necklace $a$ of $n$ beads.
The precondition $a[1]=0$ (line~2) holds for any
lexicographically least necklace containing at least one bead of
each colour, and is guaranteed by the FKM algorithm
(Figure~\ref{algorithm:fkm}), which produces canonical
representatives.  The scan skips the first white bead encountered
(line~10, where $s = 0$) so that the initial zero-count is
not recorded, and instead appends the trailing black-bead count as
the cyclic wrap-around element (line~16), yielding exactly $s$
elements satisfying $\sum_{i=0}^{s-1} t_i = n - s$  (Lemma~\ref{lemma:offset-increment}).}
\label{algorithm:genIncArr}
\end{figure}

The elements within an IA represent the offset increments
between each white bead in the necklace, which is equivalent
to the number of sequential black beads between each white
bead.  Consider the example in Figure~\ref{fig:necklace}, where
the lexicographically least representative necklace is
represented by the sequence $001101$.  Note that the first
white bead is immediately followed by another white bead,
with no black beads separating them.  Therefore the offset
increment between these two beads is $0$.  In contrast, there
are two black beads between the second and third white bead,
resulting in an offset increment of $2$.  For the last offset
increment, we consider the number of contiguous black beads
between the last and first white bead; in this case there is
one final black bead resulting in an offset increment of $1$.

\begin{figure}[t]
\scriptsize
\begin{algorithmic}[1]
\Procedure{FKM}{$k,n$} \Comment{Generate necklaces of length $n$ with $k$ colours}

\For{$i\gets 0, n$}
    \State $a[i] \gets 0$
\EndFor
\State $i \gets n$
\State \textsc{ProcessNecklace}($n$, $a$) \Comment{Output the all-zeros necklace $0^n$}
\Repeat
	\State $a[i] \gets a[i]+1$
	\For{$j\gets 1, n-i$}
		\State $a[j+i] \gets a[j]$
	\EndFor
	\If{$(n \bmod{i} = 0)$}
        \State \textsc{ProcessNecklace}($n$, $a$) \Comment{Output necklace $a$}
    \EndIf
    \State $i \gets n$
    \While{$(a[i] = (k-1))$}
        \State $i \gets i-1$
    \EndWhile
\Until{$i = 0$}
\EndProcedure

\end{algorithmic}
\caption{FKM Algorithm, based on
\citet{FREDRICKSEN1978207,FREDRICKSEN1986181} and discussed in
\citet{RUSKEY1992414}.  In this implementation, a necklace of
length $n$ is stored in an array of length $n+1$, where
$a[0]=0$.  The all-zeros necklace is output explicitly
(line~6) before the main loop, which generates all remaining
necklaces beginning with $0^{n-1}1$.  In the context of
N-DCA, \textsc{ProcessNecklace} corresponds to executing
\textsc{GenIncArray} followed by the designation test (see Figure~\ref{algorithm:ndcva}).}
\label{algorithm:fkm}
\end{figure}

This mechanism for generating the unique run-length encoding
of a necklace $a$ as its corresponding IA $\underline{t}$ is
given by \textsc{GenIncArray} in Figure~\ref{algorithm:genIncArr}.
The algorithm scans the necklace left to right, maintaining
a running count $j$ of contiguous black beads (lines~7--8).
Each white bead ($a[i] = 0$) after the first triggers the
storage of the current count as the next IA element
(lines~10--11); the first white bead is skipped because the
precondition $a[1]=0$ (line~2) ensures the count before it is
always zero (such that the guard at line~10 fails when $s = 0$).
After the main loop, the trailing count $j$ is appended as the
final element (line~16), capturing the cyclic wrap-around from the
last white bead back to the first.  The result is an IA of
exactly $s$ elements.
Table~\ref{table:fkmOutput} lists each of the possible
two-colour necklaces containing $n=6$ beads and the
corresponding IA for coalitions of size $s$, which also
corresponds to the number of white beads in the necklace.
Note that the necklace $000000$ corresponds to the grand
coalition (i.e.\ $s=6$), whereas the necklace $111111$
corresponds to an empty coalition (i.e.\ with no agents,
such that $s=0$).

As for generating the necklaces themselves, several algorithms have been proposed that, amongst others, calculate all possible necklaces
\citep{CATTELL2000267,FREDRICKSEN1978207,FREDRICKSEN1986181,RUSKEY1992414,8005509}, those of fixed content
\citep{Sawada:2003:FAG:2780704.2781151} or those of fixed density
\citep{Sawada:1999:EAG:314500.314910}.  The \emph{FKM} algorithm,
originally proposed by Fredricksen and Kessler \cite{FREDRICKSEN1986181}, 
is a simple and efficient algorithm which
generates the lexicographically smallest string for each
necklace of length $n$ given $k$ coloured beads, beginning with the string $0^n$ and ending with
$(k-1)^n$.  It has been proven to generate all
necklaces, for a given $k$ and $n$, in \emph{constant
amortised time} or CAT, i.e.\ where the total time is
$O(N_k(n))$, based on Equation \eqref{eqn:totalnecklaces}.  A version of the \emph{FKM} algorithm  \citep{FREDRICKSEN1986181,FREDRICKSEN1978207} is illustrated in
Figure~\ref{algorithm:fkm}. Note that the algorithm is presented here for the general case of $k$ colours; it is used by the procedure {\sc N-DCA} (Figure~\ref{algorithm:ndcva}) to only generate necklaces of two-colours (i.e.\ $k=2$), with $n$ corresponding to the total number of agents.

The FKM increment-and-propagate loop (lines~7--14)
generates necklaces beginning with $0^{n-1}1$; the all-zeros
necklace $0^n$ must therefore be output separately before
the loop is entered (line~6).  This necklace corresponds to
the grand coalition ($s = n$), which is processed through
the same designation machinery as all other necklaces.

\begin{table}[t]
\centering
\setlength{\tabcolsep}{15pt}
\caption{The lexicographical least representative necklaces
and their corresponding Increment Arrays (IAs) for $n=6$,
where $s$ denotes the number of agents in a coalition,
or the number of white beads in the necklace.  The order in
which the necklaces are generated using the FKM algorithm
(based on \citet{FREDRICKSEN1978207,FREDRICKSEN1986181} and
discussed in \citet{RUSKEY1992414}) is also given.}
\label{table:fkmOutput}
\begin{tabular}{r@{\;\;}l@{\;\;}c@{\;\;}l@{\qquad\qquad}r@{\;\;}l@{\;\;}c@{\;\;}l}
\toprule
Order & Necklace & $s$ & IA &
Order & Necklace & $s$ & IA \\
\midrule
1:  & \texttt{000000} & 6 & $\langle 0, 0, 0, 0, 0, 0 \rangle$ &  8: & \texttt{001101} & 3 & $\langle 0, 2, 1 \rangle$ \\
2:  & \texttt{000001} & 5 & $\langle 0, 0, 0, 0, 1 \rangle$   &  9: & \texttt{001111} & 2 & $\langle 0, 4 \rangle$ \\
3:  & \texttt{000011} & 4 & $\langle 0, 0, 0, 2 \rangle$      & 10: & \texttt{010101} & 3 & $\langle 1, 1, 1 \rangle$ \\
4:  & \texttt{000101} & 4 & $\langle 0, 0, 1, 1 \rangle$      & 11: & \texttt{010111} & 2 & $\langle 1, 3 \rangle$ \\
5:  & \texttt{000111} & 3 & $\langle 0, 0, 3 \rangle$         & 12: & \texttt{011011} & 2 & $\langle 2, 2 \rangle$ \\
6:  & \texttt{001001} & 4 & $\langle 0, 1, 0, 1 \rangle$      & 13: & \texttt{011111} & 1 & $\langle 5 \rangle$ \\
7:  & \texttt{001011} & 3 & $\langle 0, 1, 2 \rangle$         & 14: & \texttt{111111} & 0 & $\langle \rangle$\\
\bottomrule
\end{tabular}
\end{table}

The following proposition establishes that this encoding is a bijection between
necklaces and IA equivalence classes.

\begin{proposition}[Necklace--IA bijection]
\label{prop:bijection}
Let $\mathcal{N}(n,s)$ denote the set of two-colour necklaces
of length~$n$ with exactly $s$ beads of colour~$0$, and let
$\mathcal{E}(n,s)$ denote the set of canonical representative
IAs of size~$s$ for $n$ agents (i.e.\ one per equivalence
class under $\approx$).  The run-length encoding implemented
by \textsc{GenIncArray} defines a bijection
$\mathcal{N}(n,s) \to \mathcal{E}(n,s)$.
\end{proposition}

\begin{proof}
Consider a two-colour necklace of length $n$ with $s$ beads
of colour $0$ (white) and $n - s$ beads of colour $1$ (black),
where $1 \leq s \leq n$.%
\footnote{The boundary cases $s = 0$ (all black, empty
coalition) and $s = n$ (all white, grand coalition) each
consist of a single necklace and a single IA, and the
correspondence for these cases is trivial.}
Its lexicographically smallest representative is a binary
string $a = a_1 a_2 \cdots a_n$.  Since $0 < 1$, any
rotation starting with a $1$ is lexicographically larger
than any rotation starting with a $0$; hence $a_1 = 0$
for the canonical representative of any necklace containing
at least one bead of each colour.

\smallskip\noindent\textbf{Encoding.}\quad
We define the run-length encoding $\Phi$ as follows.
Because $a_1 = 0$, the string has the form
\begin{equation}\label{eq:string-structure}
    a \;=\; \underbrace{0\,\underbrace{1 \cdots 1}_{t_0}}
    \;\underbrace{0\,\underbrace{1 \cdots 1}_{t_1}}
    \;\cdots\;
    \underbrace{0\,\underbrace{1 \cdots 1}_{t_{s-1}}}
\end{equation}
where $t_i \geq 0$ is the number of consecutive $1$s
following the $(i{+}1)$-th~$0$.  The encoding produces
$\Phi(a) = \underline{t} = \langle t_0, t_1, \ldots,
t_{s-1} \rangle$, which is exactly the output of
\textsc{GenIncArray}$(n, a)$.

We establish four properties:
\begin{enumerate}[label=(\roman*)]
\item \emph{$\underline{t}$ is a valid IA.}
Each $t_i \geq 0$ counts black beads between successive
white beads, and $\sum_{i=0}^{s-1} t_i = n - s$ since the
$n - s$ black beads are partitioned exhaustively among the
$s$ gaps.  Hence $\underline{t}$ satisfies
Definition~\ref{def:increment-array}.

\item \emph{The encoding is invertible.}
Given any valid IA
$\underline{t} = \langle t_0, \ldots, t_{s-1} \rangle$ of
size $s$ with $\sum t_i = n - s$, the string
$a = 0\,1^{t_0}\,0\,1^{t_1} \cdots 0\,1^{t_{s-1}}$ is a
binary string of length $s + \sum t_i = n$ with exactly $s$
zeros.  By inspection, $\Phi(a) = \underline{t}$.
Furthermore, the structure~\eqref{eq:string-structure}
shows that $a$ is the \emph{unique} string starting with $0$
that maps to $\underline{t}$ under $\Phi$: the positions of
the $0$s are completely determined by the gap counts.
Hence $\Phi$ is a bijection between binary strings of the
form~\eqref{eq:string-structure} and valid IA tuples of
size~$s$ summing to $n - s$.

\item \emph{Rotation invariance.}
A cyclic rotation of a binary string by one position shifts
the starting point of the circular scan, cyclically permuting
which $0$ is encountered first.  This applies a circular
shift to the gap counts, producing an IA related by
$\approx$ (Definition~\ref{def:ia-equivalence}).
Conversely, any circular shift of an IA corresponds to
starting the run-length encoding from a different $0$, which
is a rotation of the string.  Hence necklace equivalence
under rotation coincides exactly with IA equivalence
under $\approx$.

\item \emph{Lexicographic order is preserved.}
Let $\underline{t}$ and $\underline{u}$ be two distinct IAs
of the same size $s$, and let $a$ and $b$ be their
corresponding strings under~\eqref{eq:string-structure}.
Let $j$ be the first index at which $\underline{t}$ and
$\underline{u}$ differ, i.e.\ $t_i = u_i$ for $i < j$ and
$t_j \neq u_j$.  Without loss of generality, suppose
$t_j < u_j$.  The strings $a$ and $b$ agree on the first
$j + 1 + \sum_{i=0}^{j-1} t_i + t_j$ positions (through
the $t_j$~ones following the $(j{+}1)$-th zero in $a$).
At the next position, $a$ has a $0$ (the start of the next
block) while $b$ has a $1$ (since $u_j > t_j$ means its
run of $1$s has not yet ended).  Hence
$a <_{\text{lex}} b$.  Since the argument is symmetric,
$\underline{t} <_{\text{lex}} \underline{u}$ if and only if
$a <_{\text{lex}} b$: the encoding $\Phi$ is
order-preserving.

It follows that the lexicographically smallest rotation of
a necklace string (the canonical necklace representative)
maps under $\Phi$ to the lexicographically smallest circular
shift of the IA (the canonical IA representative).
\end{enumerate}

\noindent
By (i) and (ii), $\Phi$ maps each canonical necklace
representative to a valid IA, and every valid IA arises
from exactly one such string.  By~(iii), necklace equivalence
classes under rotation correspond to IA equivalence classes
under~$\approx$.  By~(iv), canonical representatives map to
canonical representatives.  Hence $\Phi$ restricts to a
bijection $\mathcal{N}(n,s) \to \mathcal{E}(n,s)$.
\end{proof}

\subsection{The N-DCA Algorithm}
\label{sec:ndcva_algorithm}

We can now present the complete \emph{Necklace-based Distributed
Coalition Value Calculation Allocation} (N-DCA) algorithm
that each agent executes independently to determine its own
coalition value calculation allocation.  The algorithm
combines necklace generation (Section~\ref{sec:bijection})
with the rotated designation scheme
(Section~\ref{sec:rotated-designation}) to produce
the set of coalitions $CV_x$ that agent $x$ is responsible
for evaluating.

The \textsc{N-DCA} algorithm (Figure~\ref{algorithm:ndcva}) requires three other supporting procedures in addition to \textsc{FKM} (Figure~\ref{algorithm:fkm}).  The first, \textsc{GenIncArray} (Figure~\ref{algorithm:genIncArr}), converts a necklace to its
canonical IA.  The second, \textsc{Period} (Figure~\ref{algorithm:period}) determines the period of an IA (Definition~\ref{def:period}); and the third, \textsc{GenCoalition} (Figure~\ref{algorithm:coalgen}) generates a coalition from a starting agent and an IA (Definition~\ref{def:coalition-generation}).

\subsubsection{Period Detection}
\label{sec:period_detection}

The designation scheme requires knowing whether each
canonical IA is aperiodic ($\pi(\underline{t}) = s$) or
periodic ($\pi(\underline{t}) < s$), and in the periodic
case, the value of the stride
$\varpi(\underline{t}) = n\pi(\underline{t})/s$.  The period
can be determined by checking whether $\underline{t}$
consists of $s/p$ identical blocks of length $p$, testing
divisors $p$ of $s$ in increasing order and returning the
first match.
 
\begin{figure}[h]
\scriptsize
\begin{algorithmic}[1]
\Procedure{Period}{$\underline{t}$, $s$}
    \Comment{Compute period $\pi(\underline{t})$}
\For{each divisor $p$ of $s$ in increasing order} \Comment{try candidate period}
    \State $\mathit{periodic} \gets \mathbf{true}$ \Comment{assume periodic until mismatch}
    \For{$k \gets p$ to $s-1$} \Comment{compare with block counterpart}
        \If{$t[k] \neq t[k \bmod p]$}
            \State $\mathit{periodic} \gets \mathbf{false}$ \Comment{mismatch: $p$ not the period}
            \State \textbf{break}
        \EndIf
    \EndFor
    \If{$\mathit{periodic}$} \Comment{all elements matched}
        \State \Return $p$
    \EndIf
\EndFor
\EndProcedure
\end{algorithmic}
\caption{Algorithm to compute the period
$\pi(\underline{t})$ of an IA $\underline{t}$ of size $s$ (Definition~\ref{def:period}).
The procedure tests divisors of $s$ in increasing order
(line~2), checking whether $\underline{t}$ consists of
$s/p$ identical copies of its first $p$ elements
(lines~4--7), and returns the first $p$ for which this
holds (line~10).  For an aperiodic IA, no proper divisor
matches and the procedure returns $s$.}
\label{algorithm:period}
\end{figure}
 
The procedure \textsc{Period} is given in Figure~\ref{algorithm:period}.
For each candidate period~$p$ (line~2), the inner loop
(lines~4--7) checks whether every element beyond the first
block matches its counterpart at offset $k \bmod p$ (line~5).
A single mismatch (line~6) breaks immediately, and if all
elements agree, the period $p$ is returned (line~10).
In the worst case (when $\underline{t}$ is aperiodic), all
divisors of $s$ are tested, but the total number of element
comparisons is $O(s)$ as early mismatches will cause the inner
loop to terminate quickly.  For periodic IAs, the procedure
terminates as soon as the shortest repeating block is found.

\subsubsection{Coalition Generation}
\label{sec:coalition_gen}

The procedure for generating the coalition, \textsc{GenCoalition}, appears in Figure~\ref{algorithm:coalgen}.
Given a starting agent $x$ and an IA $\underline{t} = \langle t_0, t_1, \ldots, t_{s-1} \rangle$,
the coalition $C(x, \underline{t})$ is constructed by computing the cumulative integer increments $\varphi_i$ (Equation~\eqref{eqn:offset}) on line 5 and applying them modulo $n$ on line 6 (Definition~\ref{def:coalition-generation}). The  agent $x$ is included as the first element of the coalition (line 2)

\begin{figure}[h]
\scriptsize
\begin{algorithmic}[1]
\Procedure{GenCoalition}{$x$, $\underline{t}$, $s$, $n$}
    \Comment{Generate $C(x, \underline{t})$}
\State $C[1] \gets x$
\State $\varphi \gets 0$
\For{$i \gets 2$ to $s$}
    \State $\varphi \gets \varphi + t[i-2] + 1$
    \State $C[i] \gets ((x - 1 + \varphi) \bmod n) + 1$
        \Comment{Map to $\{1,\ldots,n\}$}
\EndFor
\State \Return $C$
\EndProcedure
\end{algorithmic}
\caption{Algorithm to generate the coalition
$C(x, \underline{t})$ of size $s$ for agent $x$ given an IA
$\underline{t}$, with $n$ agents.  The coalition members are
determined using the cumulative integer increments
(Equation~\eqref{eqn:offset}, Definition~\ref{def:coalition-generation}).}
\label{algorithm:coalgen}
\end{figure}

\subsubsection{The Complete Algorithm}
\label{sec:complete_algorithm}

\begin{figure}[t]
\scriptsize
\begin{algorithmic}[1]
\Procedure{N-DCA}{$x$, $n$}
    \Comment{Coalition allocation for agent $x$}
\State $h[1 \ldots n] \gets 0$
    \Comment{Per-size cumulative offsets}
\For{each necklace $a$ generated by
    \textsc{FKM}($2$, $n$)}
    \State $\underline{t} \gets
        \textsc{GenIncArray}(n, a)$
    \State $s \gets |\underline{t}|$
        \Comment{Coalition size = number of white beads}
    \If{$s = 0$}
        \State \textbf{continue}
        \Comment{Skip empty coalition}
    \EndIf
    \State $\pi \gets \textsc{Period}(\underline{t}, s)$
    \If{$\pi = s$}
        \Comment{Aperiodic: agent $x$ is designated}
        \State \textsc{GenCoalition}($x$,
            $\underline{t}$, $s$, $n$)
    \Else
        \Comment{Periodic: apply designation test}
        \State $d \gets n \cdot \pi \;/\; s$
            \Comment{Stride $\varpi(\underline{t})$}
        \If{$(x - 1 - h[s]) \bmod n < d$}
            \State \textsc{GenCoalition}($x$,
                $\underline{t}$, $s$, $n$)
        \EndIf
        \State $h[s] \gets h[s] + d$
            \Comment{Advance offset for size $s$}
    \EndIf
\EndFor
\EndProcedure
\end{algorithmic}
\caption{The N-DCA algorithm for agent $x$ with $n$ agents.
Canonical IAs are generated via two-colour necklaces produced
by \textsc{FKM} (Figure~\ref{algorithm:fkm}).  Aperiodic IAs
are evaluated by every agent; periodic IAs are evaluated only
by designated agents, as determined by the rotated designation
scheme (Definition~\ref{def:rotated-designation})  A per-size cumulative
offset $h[s]$ is maintained to ensure that the per-size load
balance of Theorem~\ref{thm:persize-balance} is achieved.}
\label{algorithm:ndcva}
\end{figure}

The full N-DCA algorithm is presented in
Figure~\ref{algorithm:ndcva}.  Each agent $x \in \{1,\ldots,n\}$
executes this procedure independently, using only its own
identifier $x$ and the shared parameter $n$, thereby requiring no
communication or shared state between agents.

The \textsc{FKM} procedure generates all two-colour necklaces of length $n$ in
lexicographic order (line 3).
As FKM interleaves necklaces of different densities (e.g.\ for $n = 6$, the sequence
includes necklaces with $s = 6, 5, 4, 4, 3, 4, 3, 3, 2,
\ldots$ as shown in Table~\ref{table:fkmOutput}), the
algorithm maintains a separate cumulative offset $h[s]$ for
each coalition size $s$ (line 2).  
The body of the main loop (lines 4-15) realises the \textsc{ProcessNecklace} callback of the FKM algorithm (Figure~\ref{algorithm:fkm}): each necklace is converted to its canonical IA via \textsc{GenIncArray} (line 4), and then based on its period (line 8), it is subjected to the designation test, which determines whether agent $x$ should evaluate it (lines 9-15).  

For aperiodic IAs ($\pi = s$), every starting agent produces
a distinct coalition (Section~\ref{sec:periodic}), and
every agent is implicitly designated and generates its
coalition directly (line 10).  For periodic IAs ($\pi < s$),
$\mu = s/\pi$ agents would generate the same coalition, and therefore
only a transversal of $d = \varpi(\underline{t}) = n/\mu$
agents should evaluate it.  The designation test
(Equation~\ref{eq:designation-test}) determines whether agent $x$
falls within the current designation window for size $s$; if
so, agent $x$ generates its coalition (lines 13-14).  The offset $h[s]$ is
then advanced by $d$, rotating the window for the next
periodic IA of the same size (line 15).

As all agents execute the same \textsc{FKM} procedure with
the same parameter $n$, they encounter the necklaces in the
same deterministic order and compute the same sequence of
offsets.  The per-size offset array ensures that the designation windows
for each coalition size rotate independently, preserving the
per-size load balance of Theorem~\ref{thm:persize-balance}: for any
fixed $s$, the number of coalitions assigned to any two
agents differs by at most one.

The aggregate load balance of Theorem~\ref{thm:aggregate-balance}
can also be achieved by replacing the per-size offset array
$h[s]$ with a single global offset $H$ that is advanced
after every periodic IA, regardless of size.  This trades the
per-size guarantee for the stronger aggregate guarantee
that the \emph{total} number of coalition value calculations
assigned to any two agents, across all coalition sizes,
differs by at most one.  Both variants have identical
computational cost; the choice depends on whether per-size
or aggregate balance is preferred for the application at
hand.

\subsection{Complexity Analysis}
\label{sec:complexity}

The per-agent computational cost of N-DCA is determined by
the cost of generating necklaces and processing each one.  We
consider each component in turn.

The \textsc{FKM} procedure generates all $N_2(n)$ two-colour
necklaces of length $n$ in constant amortised time (CAT),
i.e.\ $O(1)$ per necklace and $O(N_2(n))$ in total
\citep{RUSKEY1992414}.  For each necklace, three operations
are performed: \textsc{GenIncArray} converts the necklace
string to an IA in $O(n)$ time (a single pass over the $n$
beads); \textsc{Period} determines the period of the IA in
$O(s)$ time in the worst case; and the designation test
requires $O(1)$ time (a single modular comparison).  For
designated agents, \textsc{GenCoalition} generates the
coalition in $O(s)$ time.  Since $s \leq n$, the per-necklace
cost is dominated by the $O(n)$ cost of \textsc{GenIncArray}.

The total per-agent time complexity is therefore $O(n \cdot N_2(n))$.  
Using the formula for the total number of two-colour necklaces (Equation~\eqref{eqn:totalnecklaces}), $N_2(n) \sim 2^n/n$ asymptotically. Since each agent processes all $N_2(n)$ necklaces and performs $O(n)$ work per necklace (scanning the bead sequence, detecting the period, and generating the coalition), the per-agent cost is $O(2^n/n) \cdot O(n) = O(2^n)$.
This is optimal in the following sense: the total number of coalitions across all sizes is $2^n - 1$, and each agent must evaluate $\Theta(2^n / n)$ of them (by the load-balance guarantee of Theorem~\ref{thm:aggregate-balance}).
Since each coalition has
$O(n)$ members, the cost of merely \emph{writing down} the
assigned coalitions is already $\Omega(2^n)$.  The overhead
introduced by necklace generation and the designation scheme
is therefore asymptotically negligible.

The space requirement is $O(n)$ per agent: $O(n)$ for the
necklace string and IA, and $O(n)$ for the per-size offset
array $h[s]$.  The \textsc{FKM} procedure operates in-place
on the necklace array and requires no additional storage
beyond the current necklace.

For comparison, the DCG algorithm in the original study
generates canonical IAs by enumerating integer partitions of
$n - s$ and filtering out permutations that belong to the
same equivalence class \citep{RileyAAAI15}.  This filtering step requires
additional bookkeeping and results in wasted computation on
non-canonical permutations.  N-DCA eliminates this overhead
entirely: the FKM algorithm produces only canonical
necklaces, and the bijection of
Proposition~\ref{prop:bijection} guarantees that each
corresponds to exactly one equivalence class.  An earlier empirical
comparison by \citet{PayneIEEEWIC24} demonstrated that the resulting performance is competitive with both DCG and the
DCVC algorithm of \citet{Rahwan2007}, but only considered the range $2\le n \le 17$.
In the following section (Section~\ref{sec:evaluation}) we provide a more extensive empirical comparison of N-DCA against DCVC across a number of additional dimensions, in the range $2\le n \le 25$.

This section has established that the generation of canonical
representative IAs can be reduced to the well-studied problem
of enumerating two-colour combinatorial necklaces
(Proposition~\ref{prop:bijection}), and that the resulting
N-DCA algorithm (Figure~\ref{algorithm:ndcva}) allows each
agent to independently determine its own coalition value
calculation allocation using only its identifier $x$ and the
shared parameter $n$.  The algorithm is simple to implement,
requiring only the FKM necklace generator and three short
supporting procedures, and its per-agent cost of $O(2^n)$ is
asymptotically optimal. 


\section{Empirical Evaluation} \label{sec:evaluation}

We validate and evaluate N-DCA through two sets of evaluations; in Section~\ref{sec:DCVCevaluation} we provide a quantitative and qualitative comparison with a comparable approach (DCVC) across several dimensions.  We then focus in Section~\ref{sec:balance_evaluation} specifically on an analysis of the allocation imbalance of N-DCA that occurs through different variants of the rotated designation scheme, thus providing empirical evidence to support the worst-case bounds defined by 
Theorems~\ref{thm:persize-balance} and~\ref{thm:aggregate-balance}.

\subsection{Comparative Evaluation with DCVC}\label{sec:DCVCevaluation}

This section presents an empirical evaluation of the N-DCA algorithm.
We compare N-DCA with the DCVC
algorithm by \citet{Rahwan2007} across four dimensions: (i) overall
coalition-generation time (Section~\ref{sec:eval-runtime});
(ii) per-agent generation time (Section~\ref{sec:eval-single});
(iii) working-memory requirements (Section~\ref{sec:eval-memory}); and (iv) the
internal time profile of N-DCA's constituent operations
(Section~\ref{sec:eval-breakdown}).  These results are then situated in a
practical context by analysing how generation overhead compares with
the cost of evaluating the characteristic function
(Section~\ref{sec:eval-amortised}), and the
scalability characteristics of both algorithms are discussed
(Section~\ref{sec:eval-scalability}).  Finally, the load-balance
properties of the two offset variants are validated empirically
(Section~\ref{sec:balance_evaluation}).

\subsubsection{Experimental Methodology}
\label{sec:eval-setup}

Implementations of both the N-DCA algorithm (Section~\ref{sec:algorithm})
and the DCVC algorithm were written using the language C (thus eliminating any delays that could occur during any garbage collection processes).  The DCVC implementation follows the
basic allocation scheme of~\citet[Section~3.1]{Rahwan2007}: for each
coalition size~$s$, the $\binom{n}{s}$ coalitions are listed in
reverse-lexicographic order; each agent receives
$\lfloor \binom{n}{s}/n \rfloor$ consecutive coalitions by walking
the predecessor function from a starting coalition determined by the
Pascal-matrix index mapping; any remainder is distributed via a
rotating pointer~$\alpha$.

Both implementations were verified for correctness by
exhaustively checking (for $2 \leq n \leq 20$) that:
(i) all $2^n - 1$ non-empty coalitions are generated exactly once across
$n$~agents (completeness and non-redundancy); and 
(ii) the global
set of coalitions produced by both algorithms is identical
(i.e.\ the union of each algorithm's per-agent outputs covers exactly
the same $2^n - 1$ coalitions).
The N-DCA implementation was additionally verified to
satisfy the self-interest property: every coalition assigned to
agent~$x$ contains~$x$ as a member.  DCVC does not guarantee this
property, as its allocation is purely mechanical; dividing the
reverse-lexicographic list evenly among agents without regard to
coalition membership.  

All timing measurements use \texttt{clock\_gettime(CLOCK\_MONOTONIC)}
(nanosecond resolution).
Each experiment begins with one untimed warm-up execution to bring both
code and data into cache.
To prevent the compiler from eliminating coalition-generation code as
dead, a running XOR checksum is accumulated over all generated
coalition members and verified after each run.
To mitigate against the effect of other processes delaying the processing time, each configuration was executed $R = 1000$ times for $n \leq 20$ and
$R = 100$ times for $n > 20$; the reported results reflect the mean~$\bar{x}$
and sample standard deviation~$\sigma$ of the per-run wall-clock times.
In cases where time ratios (N-DCA~/~DCVC) are reported, we also state the
95\% confidence interval for the mean execution time, computed as
$\bar{x} \pm t_{0.025, R-1} \cdot \sigma / \sqrt{R}$, where
$t_{0.025,R-1}$ is the critical value of the Student $t$-distribution.
For $R = 1000$ this yields $t \approx 1.962$; for $R = 100$,
$t \approx 1.984$.  In all cases, the 95\% confidence intervals for
the two algorithms are non-overlapping for $n \geq 7$, confirming
that the observed differences are statistically significant.

All experiments were conducted on a 2024 MacBook Pro (M4 Max) with 36GB memory. 
It should be noted that the aim here is not to ascertain specific running times on a specific platform, but rather to compare the computational profiles of the two algorithms under identical conditions.
All non-empty coalition sizes were generated for each
set of agents, including the grand coalition.

\subsubsection{Experiment~1: Total Execution Time}
\label{sec:eval-runtime}

\begin{table}[t]
\centering
\caption{Experiment~1: mean total execution time (ms) for all
$n$~agents to generate all $2^n - 1$ coalitions, with standard
deviation and N-DCA\,/\,DCVC time ratio.  Results are averaged over
$R$~runs.}
\label{table:runtime-new}
\setlength{\tabcolsep}{8pt}
\begin{tabular}{r r r c r@{${}\pm{}$}l r@{${}\pm{}$}l c}
\toprule
& Coalitions & \multicolumn{1}{c}{IAs} & Samples & \multicolumn{4}{c}{Execution Time} & Ratio \\
$n$ & \multicolumn{1}{c}{$2^n-1$} & \multicolumn{1}{c}{$N_2(n)+1$} & $R$ &
\multicolumn{2}{c}{N-DCA (ms)} &
\multicolumn{2}{c}{DCVC (ms)} &
$\frac{\text{N-DCA}}{\text{DCVC}}$ \\
\midrule
 5 &          31 &             7 & 1000 & 0.001 & 0.000 & 0.000 & 0.000 &  2.2 \\
 8 &         255 &            35 & 1000 & 0.005 & 0.001 & 0.002 & 0.000 &  2.6 \\
10 &       1\,023 &          107 & 1000 & 0.021 & 0.002 & 0.008 & 0.001 &  2.7 \\
12 &       4\,095 &          351 & 1000 & 0.082 & 0.006 & 0.028 & 0.002 &  2.9 \\
14 &      16\,383 &       1\,181 & 1000 & 0.355 & 0.021 & 0.116 & 0.007 &  3.1 \\
15 &      32\,767 &       2\,191 & 1000 & 0.760 & 0.043 & 0.234 & 0.011 &  3.3 \\
17 &     131\,071 &       7\,711 & 1000 & 4.554 & 0.163 & 0.972 & 0.031 &  4.7 \\
20 &   1\,048\,575 &     52\,487 & 1000 & 47.46 & 3.07 & 8.29 & 0.17 &  5.7 \\
22 &   4\,194\,303 &    190\,745 &  100 & 200.7 & 1.3 & 34.7 & 0.5 &  5.8 \\
25 &  33\,554\,431 & 1\,342\,183 &  100 & 1701.6 & 18.4 & 294.6 & 3.3 &  5.8 \\
\bottomrule
\end{tabular}
\end{table}

The complexity analysis of Section~\ref{sec:complexity} establishes
that N-DCA has a per-agent time complexity of
$\mathcal{O}(2^n)$, yielding an aggregate complexity of
$\mathcal{O}(n \cdot 2^n)$ when all agents are run sequentially.
However, the asymptotic analysis does not reveal the constant
factors involved, nor does it confirm that the predicted scaling
behaviour is reached for pragmatic values of $n$.  Therefore the aim of this
experiment is to empirically validate the theoretical claims of Section~\ref{sec:complexity} in that the algorithm exhibits the expected exponential growth, and to quantify the
constant-factor difference between N-DCA and DCVC that is not apparent given the $\mathcal{O}$-notation.
We achieve this by measuring the total wall-clock time for \emph{all} $n$~agents to
generate all $2^n - 1$ coalitions, which corresponds to a centralised
scenario in which an orchestrator runs each agent's generation
procedure sequentially, and represents the total computational work
required by each approach.  Even in a distributed deployment, where
agents run in parallel, the total work determines the aggregate
resource consumption and is therefore a meaningful basis for
comparison.
 
Table~\ref{table:runtime-new} presents the number of (non-empty) coalitions, number of IAs, samples generated, the execution times ($\bar{x}$ and $\sigma$) for each algorithm, and their ratio given a selection of different sizes of agents in the range $2 \leq n \leq 25$.
Both algorithms exhibit execution times that
grow linearly with the number of coalitions (which is itself
exponential in~$n$), consistent with the $\mathcal{O}(n \cdot 2^n)$
aggregate bounds.  However, the constant factors differ; DCVC is
faster than N-DCA at all tested values of~$n$, with the ratio
increasing from approximately $2{\times}$ at small~$n$ to
approximately $5.8{\times}$ at $n = 25$.  This widening ratio is
consistent with preliminary results reported by~\citet{PayneIEEEWIC24} for
$n \leq 17$, and the extended range suggests that the ratio
stabilises beyond $n \approx 20$.  The source of this
constant-factor difference is analysed in detail in
Section~\ref{sec:eval-single}, where a per-coalition operation-count
analysis attributes it to the $\mathcal{O}(n)$ cost of
\textsc{GenIncArray}'s necklace scan.
The number of IAs reported in column 3 differs by $1$ to the number of two-colour necklaces generated for each $n$.\footnote{The sequence of values for two-colour necklaces also appears in The Online Encyclopedia of Integer Sequences (OEIS) \cite[A052823]{oeis}, which can be found at \url{https://oeis.org/A052823}.} These counts include the grand coalition (represented by monochrome necklaces of only white beads, which explains the difference of 1), but omit \emph{empty} IAs, where $s=0$ (i.e.\ monochrome necklaces of black beads).

\begin{figure}[t]
    \centering
    \includegraphics[width=0.8\linewidth]{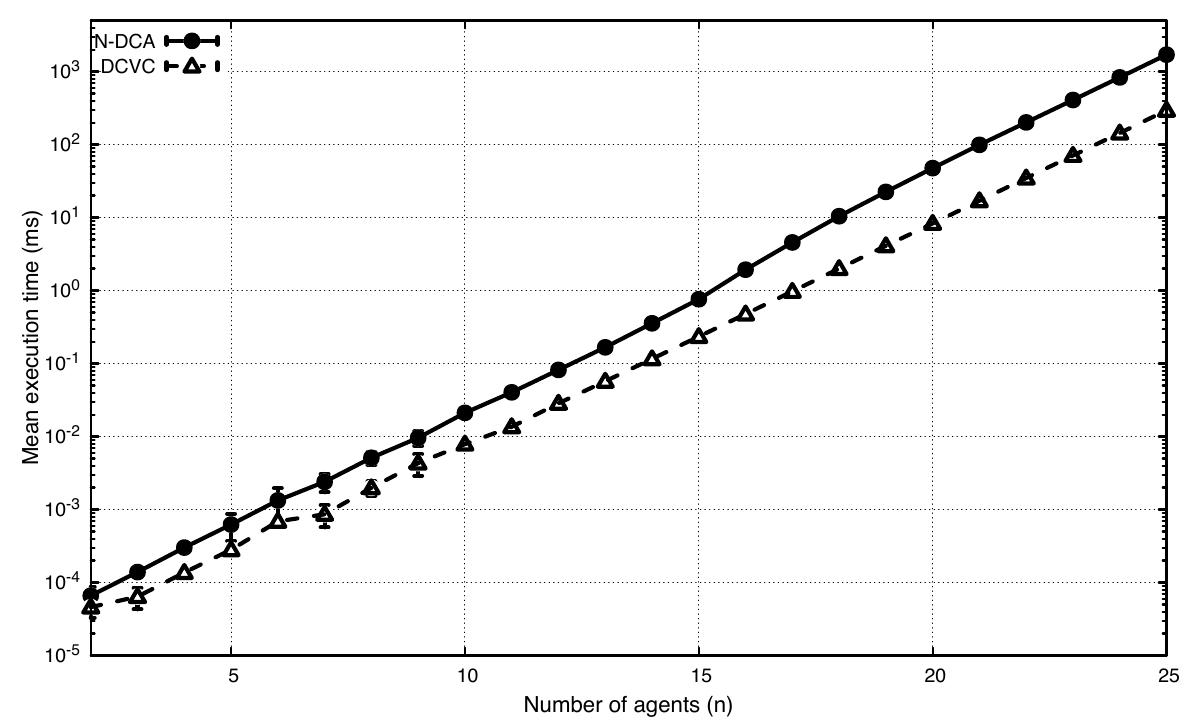}
    \caption{Experiment~1: the mean total execution time for generating all coalitions by N-DCA and DCVC for each of the 
    $n$~agents to generate all $2^n-1$ coalitions. Note that the y-axis is logarithmic, and error bars show $\pm 1$ SD over $R$ runs.}
    \label{fig:runtime-total}
\end{figure}

The mean execution times have been plotted against the different numbers of agents (together with error bars showing $\pm 1$ SD over the $R$ runs) in Figure~\ref{fig:runtime-total} (note that the execution times are plotted on a logarithmic scale).
Both curves are approximately linear on this scale, confirming that the growth rate is dominated by the number of coalitions.  
These results are also consistent with the description and performance of DCVC \citep{Rahwan2007} and the fact that the complexity of the FKM algorithm was shown to have a bound of $\mathcal{O}(N(k,n))$ on the total time required to generate all necklaces.
The vertical separation between the curves corresponds to the constant-factor difference, which reflects the ratio as given in Table~\ref{table:runtime-new}.

\subsubsection{Experiment~2: Single-Agent Execution Time}
\label{sec:eval-single}
 
In a distributed deployment, each agent generates its own coalition
share independently.  The relevant metric is therefore the per-agent
execution time.  In this experiment, we measure the time for agent~$1$
to generate its share of the coalitions.
For the purposes of this experiment, the choice of agent~$1$ is arbitrary: for N-DCA, every agent performs
the same FKM enumeration and differs only in the designation-test
outcomes, which are negligible in cost.  This contrasts with DCVC, where every agent
performs the same number of predecessor steps, differing only in
starting position.
Note that for N-DCA, although each agent iterates through all
necklaces via FKM, it only evaluates (i.e.\ calls
\textsc{GenCoalition} for) its designated share.
 
\begin{table}[t]
\centering
\caption{Experiment~2: mean execution time (ms) for agent~$1$ to
generate its share of coalitions, with standard deviation and
N-DCA\,/\,DCVC time ratio.  The allocation columns report the
exact number of coalitions assigned to agent~$1$ under each
algorithm's designation scheme.}
\label{table:runtime-single}
\setlength{\tabcolsep}{8pt}
\begin{tabular}{r r r r r r@{${}\pm{}$}l r@{${}\pm{}$}l r}
\toprule
$n$ & \multicolumn{1}{c}{N-DCA} & \multicolumn{1}{c}{DCVC} &
\multicolumn{1}{c}{$\Delta$} & $R$ &
\multicolumn{2}{c}{N-DCA (ms)} &
\multicolumn{2}{c}{DCVC (ms)} &
Ratio \\
& \multicolumn{1}{c}{alloc.} & \multicolumn{1}{c}{alloc.} & & & & & & & 
$\frac{\text{N-DCA}}{\text{DCVC}}$ \\
\midrule
 5 &           7 &           7 &  0 & 1000 & 0.000 & 0.000 & 0.000 & 0.000 &  --- \\
 8 &          34 &          32 &  2 & 1000 & 0.001 & 0.000 & 0.000 & 0.000 &  2.5 \\
10 &         105 &         103 &  2 & 1000 & 0.002 & 0.000 & 0.001 & 0.000 &  2.8 \\
12 &         346 &         342 &  4 & 1000 & 0.007 & 0.000 & 0.002 & 0.000 &  2.7 \\
14 &       1\,175 &       1\,171 &  4 & 1000 & 0.025 & 0.001 & 0.009 & 0.001 &  2.9 \\
15 &       2\,189 &       2\,185 &  4 & 1000 & 0.050 & 0.004 & 0.016 & 0.001 &  3.1 \\
17 &       7\,711 &       7\,711 &  0 & 1000 & 0.272 & 0.023 & 0.060 & 0.005 &  4.6 \\
20 &      52\,435 &      52\,429 &  6 & 1000 & 2.323 & 0.053 & 0.435 & 0.012 &  5.3 \\
22 &     190\,655 &     190\,651 &  4 &  100 & 9.131 & 0.119 & 1.651 & 0.051 &  5.5 \\
25 &   1\,342\,181 &   1\,342\,178 &  3 &  100 & 67.48 & 0.57 & 12.20 & 0.28 &  5.5 \\
\bottomrule
\end{tabular}
\end{table}
 
In Table~\ref{table:runtime-single}, the
allocation columns show the exact number of coalitions assigned to
agent~$1$ under each algorithm.  The two counts differ slightly
(column~$\Delta$), reflecting the different distribution mechanisms:
N-DCA's rotated designation scheme
(Definition~\ref{def:rotated-designation}) assigns coalitions based
on necklace periodicity and per-size offsets, whereas DCVC divides
the reverse-lexicographic coalition list evenly and distributes the
remainder via a rotating pointer~$\alpha$.  The differences are
small, at most $\kappa(n)$ coalitions, where $\kappa(n)$ denotes
the number of coalition sizes for which $n \nmid \binom{n}{s}$, and
have no measurable effect on the timing comparison, since even at
$n = 25$ the allocation difference of~$3$ represents less than
$0.001\%$ of the total.
Although both algorithms generate exactly $2^n - 1$ coalitions in total
across all~$n$ agents, agent~$1$ receives a slightly higher allocation under
N-DCA than under DCVC, which is then offset by other agents
receiving a slightly lower allocation.  For example, at $n = 20$, agent~$1$
is allocated $52{,}435$ coalitions under N-DCA versus $52{,}429$
under DCVC ($\Delta = +6$), while agent~$20$ receives $52{,}425$
versus $52{,}428$ ($\Delta = -3$).  The aggregate spread across
agents is wider for N-DCA (up to $\kappa(n)$, compared with at
most~$1$ for DCVC), consistent with the per-size versus aggregate
balance guarantees of Theorems~\ref{thm:persize-balance}
and~\ref{thm:aggregate-balance} respectively. The relationship
between $\kappa(n)$ and the observed allocation spread is examined
in detail in Section~\ref{sec:balance_evaluation}.

The single-agent ratios closely track those of Experiment~1,
confirming that the performance difference arises from per-coalition
overhead rather than from any inter-agent scheduling effect.
Figure~\ref{fig:ratio} plots the time ratio (N-DCA\,/\,DCVC) as a
function of~$n$ for both experiments.  The ratio increases steadily
from approximately $2{\times}$ at $n = 5$ to approximately
$5.5$--$5.8{\times}$ at $n = 25$, stabilising beyond $n \approx 20$.
 
The growing ratio is explained by the per-coalition operation counts
of the two algorithms.  For each necklace, N-DCA performs a full
$\mathcal{O}(n)$ scan in \textsc{GenIncArray} (scanning all $n$~beads) followed
by period detection (testing divisors of the coalition size~$s$).
At $n = 20$, an operation-count analysis shows that agent~1 performs
approximately $1{,}049{,}760$ \textsc{GenIncArray} scans and $298{,}797$ period
comparisons to generate $52{,}435$ coalitions---roughly
$25.7$~operations per coalition.  DCVC, by contrast, generates each
successive coalition via a single predecessor step that scans backwards
through the coalition array until it finds an incrementable position.
At $n = 20$, this averages approximately $2.6$~scan steps per
predecessor call, giving roughly $2.6$~operations per coalition.
The resulting operation-count ratio of approximately $9.8{\times}$
at $n = 20$ is consistent with the observed wall-clock ratio of
$5.3$--$5.7{\times}$ (the gap between the two ratios reflects
differences in the cost per operation: DCVC's predecessor involves
comparison and branching, while \textsc{GenIncArray}'s inner loop is a
simpler scan).
As~$n$ grows, \textsc{GenIncArray}'s $\mathcal{O}(n)$ scan cost per necklace
increases the per-coalition overhead for N-DCA, while DCVC's
per-predecessor scan cost remains approximately constant, causing
the ratio to widen.
 
\begin{figure}[t]
    \centering
    \includegraphics[width=0.8\linewidth]{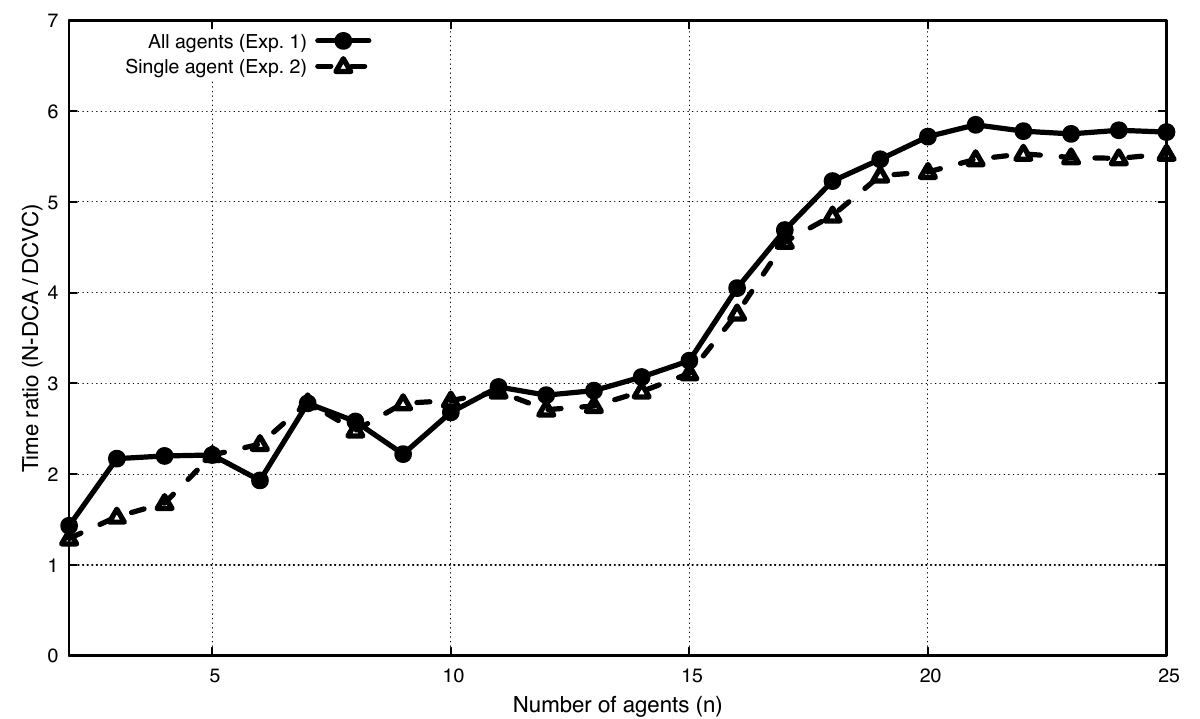}
    \caption{Time ratio (N-DCA\,/\,DCVC) for Experiments~1 and~2.
    Values above~1 indicate DCVC is faster.
    The ratio stabilises near $5.5$--$5.8{\times}$ for $n \geq 20$.}
    \label{fig:ratio}
\end{figure}

\subsubsection{Experiment~3: Working Memory}
\label{sec:eval-memory}

The timing experiments of Sections~\ref{sec:eval-runtime}
and~\ref{sec:eval-single} quantify the computational cost of
each algorithm, but do not address their memory requirements.
In resource-constrained settings, such as embedded multi-agent
systems or agents running on low-power devices, the per-agent
working-memory footprint may be as important as execution speed.
Both algorithms are known to be $\mathcal{O}(n)$ in space, but
the constant factors and the structural characteristics of that
usage (e.g.\ iterative versus recursive) can differ
significantly.  This experiment compares the per-agent
working-memory footprints of N-DCA and DCVC analytically.

Both algorithms use only stack-allocated arrays with $\mathcal{O}(n)$ space, and neither require (or within the implementation, perform) heap allocation.
Therefore, to compare their memory footprints, we analytically count the working variables necessary to support each algorithm's per-agent execution.
N-DCA maintains three arrays: (i) the FKM necklace array~$a[n{+}1]$; (ii) the increment array~$t[n]$; and (iii) the per-size offset array~$h[n{+}1]$, together with a small number of scalar variables, giving a total of approximately $3n + 12$ integers.
DCVC maintains two coalition arrays, $M[n]$ and $M_{\text{prev}}[n]$, together with scalar variables for the index arithmetic.
In addition, the index-to-coalition mapping algorithm (Figure~2 in~\cite{Rahwan2007}), which recovers a coalition from its position in the ordered list~$L_s$ by iteratively decomposing the index through a precomputed Pascal matrix (see Equation~(3) and Figure~3 in~\cite{Rahwan2007}), uses a call stack of up to $n$~frames, each containing local integers and a long-integer accumulator.

\begin{table}[t]
\centering
\caption{Experiment~3: analytical per-agent working memory (bytes)
for N-DCA and DCVC.  Both are $\mathcal{O}(n)$, whereas N-DCA uses
approximately one third of DCVC's footprint.}
\label{table:memory}
\setlength{\tabcolsep}{10pt}
\begin{tabular}{r c c c}
\toprule
$n$ & N-DCA (bytes) & DCVC (bytes) & DCVC\,/\,N-DCA \\
\midrule
 5 &  108 &  264 & 2.4 \\
10 &  168 &  464 & 2.8 \\
15 &  228 &  664 & 2.9 \\
20 &  288 &  864 & 3.0 \\
25 &  348 & 1064 & 3.1 \\
\bottomrule
\end{tabular}
\end{table}

Table~\ref{table:memory} reports the analytical byte counts.
N-DCA uses approximately one third of the working memory of DCVC
across the range tested.  The ratio approaches~$3$ as $n$ grows
because DCVC's recursive call stack dominates at higher values of $n$.  In absolute terms,
both footprints are negligible (under 1.1\,KB even at $n = 25$), and thus
the practical significance lies not in the byte counts themselves but
in the structural difference: N-DCA is purely iterative and requires
no recursion, which may be advantageous in resource-constrained
settings such as embedded multi-agent systems.  Empirical measurements
using \texttt{getrusage} confirmed that neither algorithm performs
hidden heap allocation, and the process resident set size remained
constant before and after execution at $n = 25$.

\subsubsection{Experiment~4: Component Time Profile}
\label{sec:eval-breakdown}

To understand \emph{where} the time is spent within each algorithm,
we decompose the per-agent execution time into its constituent
operations.  Direct per-operation timing (e.g.\ calling
\texttt{clock\_gettime} around each operation) introduces
substantial measurement overhead when operations are fine-grained; for example 
at $n = 25$, N-DCA processes over 1.3~million necklaces, and the
overhead of 5~timing calls per necklace could inflate the time measurements.

We therefore use a \emph{subtraction method}: for each
component depth level $\ell \in \{0, 1, 2, 3, 4\}$, we run a modified version
of N-DCA that executes only the first $\ell$~stages of the pipeline,
and time each run with a single pair of
\texttt{clock\_gettime} calls:
\begin{description}[labelindent=2em,topsep=2pt]
    \item [$\ell = 0$:] \textsc{FKM} necklace stepping only
    (enumerate necklaces without processing);
    \item [$\ell = 1$:] \textsc{FKM} + \textsc{GenIncArray};
    \item [$\ell = 2$:] \textsc{FKM} + \textsc{GenIncArray} + period detection;
    \item [$\ell = 3$:] \textsc{FKM} + \textsc{GenIncArray} + period detection + designation test;
    \item [$\ell = 4$:] full N-DCA (all components including \textsc{GenCoalition}).
\end{description}
The time attributable to each component is therefore $T_\ell - T_{\ell-1}$. 
The results are reported in Table~\ref{table:breakdown}.  
For N-DCA, \textsc{GenIncArray} is the dominant component, accounting for approximately
59\% of total time at large~$n$.  Period detection accounts for
approximately 25\%, while the FKM necklace stepping itself (the
combinatorial enumeration engine) is comparatively fast at under
6\% for $n \geq 20$.  The designation test is negligible ($\approx 1\%$),
confirming that the rotated-designation scheme
(Definition~\ref{def:rotated-designation}) adds minimal overhead.
\textsc{GenCoalition} accounts for approximately 9\% of total time.

\begin{table}[t]
\centering
\caption{Experiment~4: N-DCA component time profile for a single
agent (agent~1), measured by the subtraction method.  Percentages
indicate the fraction of total single-agent time.}
\label{table:breakdown}
\setlength{\tabcolsep}{10pt}
\begin{tabular}{r c c c c c c}
\toprule
& \multicolumn{1}{c}{Total} &
\multicolumn{1}{c}{\sc FKM} &
\multicolumn{1}{c}{\sc GenInc-} &
\multicolumn{1}{c}{Period} &
\multicolumn{1}{c}{Designation} &
\multicolumn{1}{c}{\sc GenCoal-} \\
$n$ & \multicolumn{1}{c}{(ms)} &
\multicolumn{1}{c}{step} &
\multicolumn{1}{c}{\sc Array} &
\multicolumn{1}{c}{detection} &
\multicolumn{1}{c}{test} &
\multicolumn{1}{c}{\sc ition} \\
\midrule
10 & 0.002 &
    19.2\% & 51.3\% & 20.0\% & 5.3\% & 4.2\% \\
15 & 0.055 &
    13.1\% & 58.2\% & 19.1\% & 0.5\% & 9.1\% \\
20 & 2.416 &
    6.7\%  & 58.3\% & 22.5\% & 1.1\% & 11.4\% \\
25 & 68.01 &
    5.8\%  & 58.7\% & 25.5\% & 1.2\% & 8.7\%  \\
\bottomrule
\end{tabular}
\end{table}

The reason for this dominance in N-DCA is a structural asymmetry between
\textsc{FKM} and \textsc{GenIncArray}.  The FKM algorithm achieves \emph{constant
amortised time} (CAT), $\mathcal{O}(1)$ per necklace on
average, because most steps modify only a short suffix of the
necklace array.  \textsc{GenIncArray}, however, performs a complete
$\mathcal{O}(n)$ scan of all $n$~beads for every necklace,
regardless of how few beads changed since the previous step.  At
$n = 25$, this amounts to approximately $25 \times 1{,}342{,}184
\approx 33.6 \times 10^6$ bead inspections, compared with the
$\mathcal{O}(N_2(n))$ total work performed by \textsc{FKM} itself.
 
The quantitative impact of this overhead can be estimated from the
breakdown data.  At $n = 25$, \textsc{GenIncArray} accounts for
$39.95\;\text{ms}$ of the $68.01\;\text{ms}$ single-agent total.
Eliminating this cost entirely would reduce N-DCA's time to
$28.05\;\text{ms}$, lowering the N-DCA\,/\,DCVC ratio from
$5.6{\times}$ to $2.3{\times}$.  A more realistic scenario; for example, an
incremental update scheme achieving $\mathcal{O}(1)$ amortised
cost per necklace, with constant factors comparable to FKM's own
stepping cost ($3.96\;\text{ms}$), would thus yield an estimated
single-agent time of approximately $32\;\text{ms}$, corresponding
to a ratio of $2.6{\times}$.  Period detection ($25.5\%$) would
then become the dominant component.
The possibility of such an incremental scheme is discussed further
in Section~\ref{sec:discussion}.

\subsubsection{Generation Overhead in Context}
\label{sec:eval-amortised}

Although the previous experiments focus on the time required to \emph{generate} coalitions, a pragmatic deployment of coalition value calculation allocations would be followed by the \emph{evaluation} of the
characteristic function $\nu(C)$ itself for each generated coalition.  The
cost of this evaluation is application-dependent and typically dominates
overall execution time. Even a modest evaluation cost of
$1\;\mu\text{s}$ per coalition at $n = 25$
(where each agent generates approximately $1.34 \times 10^6$
coalitions) would add approximately $1{,}342\;\text{ms}$ of
evaluation time, compared with generation times of $67.5\;\text{ms}$
(N-DCA) and $12.2\;\text{ms}$ (DCVC).

\begin{figure}[t]
    \centering
    \includegraphics[width=0.8\linewidth]{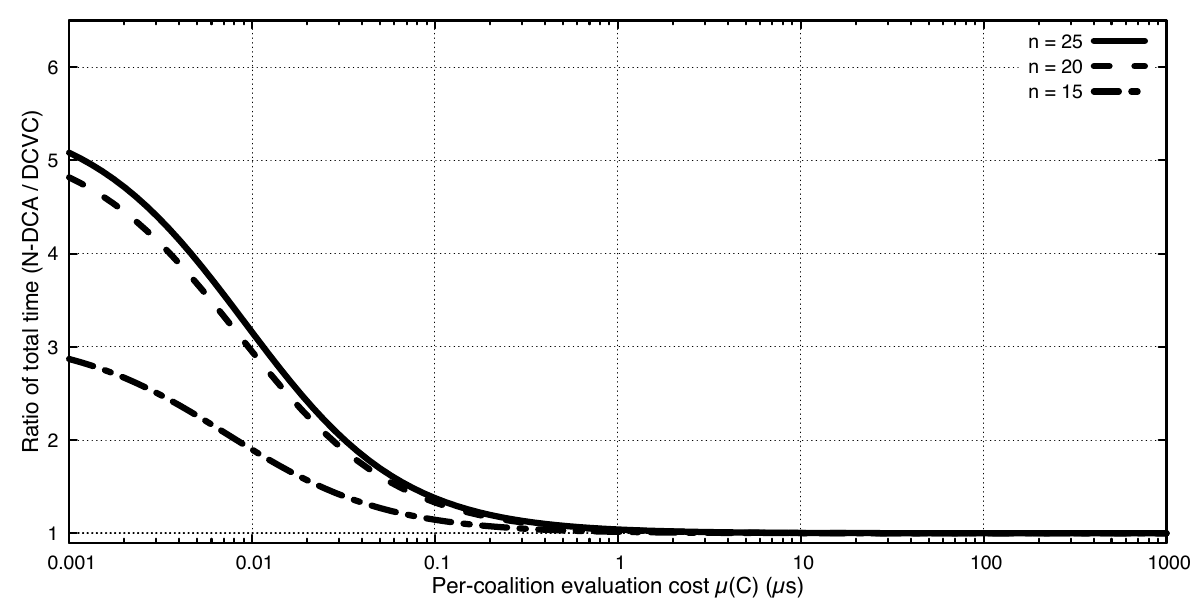}
    \caption{Total time ratio $\eta(c)$ (N-DCA\,/\,DCVC) as a
    function of the per-coalition characteristic-function evaluation
    cost~$c$, for $n \in \{15, 20, 25\}$.  As~$c$ increases beyond
    $\approx 1\;\mu\text{s}$, the ratio converges rapidly toward~$1$.}
    \label{fig:amortised}
\end{figure}

To quantify this effect, we compute the total time ratio
(N-DCA\,/\,DCVC) as a function of the per-coalition evaluation
cost~$c$:
\begin{equation}
    \label{eqn:amortised-ratio}
    \eta(c) \;=\; \frac{T_{\text{gen}}^{\text{N-DCA}} + m \cdot c}
                       {T_{\text{gen}}^{\text{DCVC}} + m \cdot c}
\end{equation}
where $m$ denotes the number of coalitions generated per agent and
$T_{\text{gen}}$ denotes the mean generation time from Experiment~2.
Since both algorithms generate the same coalitions (and hence incur
the same evaluation cost), the evaluation term $m \cdot c$ appears
in both numerator and denominator, causing the ratio to converge
to~$1$ as $c$ increases.

Figure~\ref{fig:amortised} plots $\eta(c)$ for $n \in \{15, 20, 25\}$.
At a per-coalition evaluation cost of $10\;\mu\text{s}$, which is
very conservative for any non-trivial characteristic
function, the time ratio falls below $1.05$ for all values of~$n$
tested, suggesting that the two algorithms have effectively
indistinguishable total execution times.  At
$100\;\mu\text{s}$ per coalition, the ratio is below $1.005$.
This analysis shows that the generation-time difference between
N-DCA and DCVC, while real and measurable in isolation, has
negligible practical impact in typical coalition value calculation
scenarios.  The performance-relevant comparison between the two
algorithms therefore shifts to their structural and algorithmic
properties, where N-DCA offers several distinct advantages:
a per-agent working-memory footprint approximately one third that of
DCVC (Section~\ref{sec:eval-memory}); a purely iterative
implementation requiring no recursion; the self-interest guarantee
that every coalition assigned to agent~$x$ contains~$x$
(a property that DCVC does not provide, as verified empirically for
$n$ up to~20); and the absence of integer-overflow constraints
that limit DCVC to $n < 68$ under standard 64-bit arithmetic
(Section~\ref{sec:eval-scalability}).

\subsubsection{Scalability Considerations}
\label{sec:eval-scalability}

A further structural difference between the two algorithms that is not
captured by the timing experiments concerns their numerical
requirements.

DCVC's allocation procedure relies on computing binomial
coefficients $\binom{n}{s}$ and using them for index arithmetic:
determining each agent's starting index, mapping that index to a
coalition via the Pascal-matrix decomposition, and computing the
remaining number of coalitions.  
All of these operations require exact integer values of $\binom{n}{s}$.
In a standard 64-bit integer representation, the largest representable value is $2^{63}-1$, i.e. \ $\approx 9.22 \times 10^{18}$.
The central binomial coefficient $\binom{n}{\lfloor n/2 \rfloor}$, corresponding to 
the number of coalitions of the most populous size, and hence the largest index that DCVC's Pascal-matrix decomposition must represent (see Figure~2 in~\cite{Rahwan2007}) 
exceeds this threshold at $n = 68$ ($\binom{68}{34} \approx 2.8 \times 10^{19}$), at which point DCVC's index arithmetic overflows and produces incorrect results unless arbitrary-precision integers are used.

In contrast, N-DCA performs no index arithmetic and computes no
binomial coefficients.  The FKM algorithm operates on a binary
array of length~$n$; \textsc{GenIncArray} and period detection operate on
the increment array of length at most~$n$; and \textsc{GenCoalition} uses
only modular arithmetic on agent identifiers.  All quantities remain
bounded by~$n$ throughout, regardless of the number of coalitions.
N-DCA is therefore free of integer-overflow constraints and can, in
principle, be applied at any value of~$n$ for which the exponential
number of coalitions remains computationally tractable.

While values of $n \geq 68$ are unlikely to be encountered in current
practice (as $2^{68} \approx 2.95 \times 10^{20}$ coalitions would
require astronomical computation), the absence of numerical
preconditions simplifies implementation and eliminates an entire
class of potential errors.

\subsection{N-DCA Load Balance Analysis}
\label{sec:balance_evaluation}

The rotated designation scheme (Section~\ref{sec:rotated-designation}) is responsible for distributing the periodic-IA designations evenly across all agents, without the need for any inter-agent communication or central coordinator.
Whilst this ensures that exactly $2^n - 1$ coalitions will be generated across all~$n$ agents, individual agents may receive slightly different allocations, depending on whether the aim is to balance the load for a specific coalition size $s$ (Theorem~\ref{thm:persize-balance}), or if the aggregate load across all sizes should be balanced (Theorem~\ref{thm:aggregate-balance}).
In this subsection, we validate these theoretical guarantees empirically, characterise the conditions under which each bound is tight, and examine the trade-offs between the two variants across a range of values of~$n$.

For each value of $n$ in the range $2 \leq n \leq 25$, the N-DCA algorithm (Algorithm~\ref{algorithm:ndcva}) is executed independently
for every agent $x \in \{1, \ldots, n\}$ under both offset variants:
\begin{itemize}
    \item 
        \emph{Per-size offsets} ($h[s]$): the cumulative offset is maintained
independently for each coalition size~$s$ and reset to~$0$ at the start of each new size (Definition~\ref{def:rotated-designation}).
    \item 
        \emph{Global offset} ($H$): a single cumulative offset is carried across all coalition sizes
(Equation~(\ref{eq:global-offset})).
\end{itemize}
For each variant, the number of coalitions assigned to
each agent is recorded for every coalition size~$s$, yielding the per-agent,
per-size allocation counts $|CV_x^s|$.  From these, we can determine
the per-size imbalance at each size~$s$
($\max_x |CV_x^s| - \min_x |CV_x^s|$),
the maximum per-size imbalance over all sizes,
and the aggregate imbalance
($\max_x |CV_x| - \min_x |CV_x|$,
where $|CV_x| = \sum_{s=1}^{n} |CV_x^s|$).
We also compute $\kappa(n)$ in Equation \eqref{eqn:kappa}, which counts the number of coalition sizes at which the total number of coalitions $\binom{n}{s}$ is not evenly divisible by~$n$ (i.e.\  $n \nmid \binom{n}{s}$).
\begin{equation}
    \label{eqn:kappa}
    \kappa(n) \;=\; \bigl|\bigl\{s \in \{1,\ldots,n\} : n \nmid \tbinom{n}{s}\bigr\}\bigr|
\end{equation}


\begin{table}[t]
\centering
\caption{Maximum per-size imbalance (max ps-imb) and aggregate
imbalance (agg-imb) for the two N-DCA offset variants, for
$2 \leq n \leq 25$.  The per-size offset variant ($h[s]$)
always achieves per-size imbalance~$\leq 1$
(Theorem~\ref{thm:persize-balance}); the global offset variant
($H$) always achieves aggregate imbalance~$\leq 1$
(Theorem~\ref{thm:aggregate-balance}).  $\kappa(n)$ denotes the
number of coalition sizes for which
$n \nmid \binom{n}{s}$.  Prime values of $n$ are
marked with~$\ast$.}
\label{table:imbalance}
\setlength{\tabcolsep}{5pt}
\begin{tabular}{l@{\;\;}c@{\;\;}cc@{\;\;}ccc@{\qquad\qquad}l@{\;\;}c@{\;\;}cc@{\;\;}cc}
\toprule
& &
\multicolumn{2}{c}{$h[s]$} &
\multicolumn{2}{c}{$H$} &
& & &
\multicolumn{2}{c}{$h[s]$} &
\multicolumn{2}{c}{$H$} \\
\cmidrule(lr){3-4} \cmidrule(lr){5-6}
\cmidrule(lr){10-11} \cmidrule(lr){12-13}
$n$ & $\kappa(n)$ &
max & agg &
max & agg &
~ & $n$ & $\kappa(n)$ &
max & agg &
max & agg \\
 & & ps-imb & imb & ps-imb & imb & & & & ps-imb & imb & ps-imb & imb \\

\cmidrule(r){1-6} \cmidrule(){8-13}

 $2^{\ast}$  & 1 & 1 &  1 & 1 & 1 & &
 14           & 8 & 1 &  8 & 1 & 1 \\
 $3^{\ast}$  & 1 & 1 &  1 & 1 & 1 & &
 15           & 7 & 1 &  7 & 1 & 1 \\
 4            & 2 & 1 &  2 & 1 & 1 & &
 16           & 8 & 1 &  8 & 4 & 1 \\
 $5^{\ast}$  & 1 & 1 &  1 & 1 & 1 & &
 $17^{\ast}$ & 1 & 1 &  1 & 1 & 1 \\
 6            & 4 & 1 &  4 & 1 & 1 & &
 18           & 8 & 1 &  8 & 3 & 1 \\
 $7^{\ast}$  & 1 & 1 &  1 & 1 & 1 & &
 $19^{\ast}$ & 1 & 1 &  1 & 1 & 1 \\
 8            & 4 & 1 &  4 & 1 & 1 & &
 20           & 10 & 1 & 10 & 9 & 1 \\
 9            & 3 & 1 &  3 & 1 & 1 & &
 21           & 7 & 1 &  7 & 3 & 1 \\
 10           & 4 & 1 &  4 & 2 & 1 & &
 22           & 8 & 1 &  8 & 8 & 1 \\
 $11^{\ast}$ & 1 & 1 &  1 & 1 & 1 & &
 $23^{\ast}$ & 1 & 1 &  1 & 1 & 1 \\
 12           & 7 & 1 &  7 & 3 & 1 & &
 24           & 14 & 1 & 14 & 12 & 1 \\
 $13^{\ast}$ & 1 & 1 &  1 & 1 & 1 & &
 25           & 5 & 1 &  5 & 1 & 1 \\
\bottomrule
\end{tabular}
\end{table}

Table~\ref{table:imbalance} reports the results for
$2 \leq n \leq 25$. The per-size
offset variant achieves a maximum per-size imbalance of
exactly~$1$ for every value of $n$ tested, confirming
Theorem~\ref{thm:persize-balance}.  The global offset variant
achieves an aggregate imbalance of exactly~$1$ for every value
of $n$, confirming Theorem~\ref{thm:aggregate-balance}.  Both
bounds are therefore tight.

Under per-size offsets, the observed aggregate imbalance matches
$\kappa(n)$ exactly for all values of $n$ tested.  This confirms
that the bound is tight: when each coalition size maintains its
own independent designation window, the remainder coalitions
(i.e.\ the extra coalitions arising when $n \nmid \binom{n}{s}$)
accumulate on the same agent at every such size, and the
aggregate imbalance equals the number of sizes at which a
remainder exists.
A well-known property of binomial coefficients provides a useful
special case; when $n = p$ is prime,
$p \mid \binom{p}{s}$ for all $1 \leq s \leq p-1$, so the only
non-divisible size is $s = n$ (the grand coalition, where
$\binom{n}{n} = 1$).  Hence $\kappa(p) = 1$ for any prime~$p$,
and both guarantees hold simultaneously under either variant.
This is confirmed empirically for all primes in the range
($n = 2, 3, 5, 7, 11, 13, 17, 19, 23$).

The global offset variant introduces per-size imbalance
because the single offset~$H$ is advanced by strides of
varying magnitude (depending on the stride $\varpi(\underline{t})$
of each periodic IA encountered), so the designation window does
not cycle uniformly within any single coalition size.  The
per-size imbalance under the global offset is generally modest
for small~$n$ but can grow for larger composite values; for
example, at $n = 24$, the maximum per-size imbalance reaches~$12$.
Notably, some composite values (e.g.\ $n = 14, 15, 25$) exhibit
a maximum per-size imbalance of only~$1$ under the global offset,
indicating that the interaction between stride sizes and $n$ can
sometimes yield favourable cancellations.

The two variants are therefore complementary: the per-size offset
variant guarantees per-size balance at the cost of aggregate
imbalance up to $\kappa(n)$, while the global offset variant
guarantees aggregate balance at the cost of per-size imbalance.
Neither variant achieves both guarantees simultaneously for all
composite~$n$.  However, for prime values of $n$, both guarantees hold
under either variant ($\kappa(p) = 1$), eliminating the need to
choose.
The choice between the two variants depends on the application
context:
\begin{itemize}
    \item 
The \emph{per-size offset variant} is appropriate when the cost
of evaluating the characteristic function $\nu(C)$ varies
significantly with coalition size, so that the computational cost
of processing a coalition depends primarily on~$|C|$; balancing
the number of evaluations within each size then ensures that no
agent bears a disproportionate share of the most expensive
evaluations.  The aggregate imbalance of $\kappa(n)$ is a
second-order effect: even at $n = 24$, the maximum aggregate
imbalance is~$14$, which is negligible relative to the per-agent
load of approximately $699{,}050$ coalitions.
\item
The \emph{global offset variant} is appropriate when all coalition value
calculations have roughly equal cost (irrespective of coalition
size), or when the primary concern is minimising the maximum total
number of evaluations per agent; in this case,
Theorem~\ref{thm:aggregate-balance} ensures that no agent
evaluates more than $\lceil (2^n - 1)/n \rceil$ coalitions in
total.
\end{itemize}
For prime values of $n$, the distinction is moot: both guarantees hold
simultaneously with either variant, and $\kappa(n) = 1$.  Since
the number of agents in many practical multi-agent systems is a
design parameter, choosing a prime $n$ eliminates the need to
trade off between per-size and aggregate balance.

\section{Discussion}
\label{sec:discussion}

Having established the framework (Section~\ref{sec:ippe}), algorithm (Section~\ref{sec:algorithm}), and empirical results (Section~\ref{sec:evaluation}), we now situate N-DCA within the broader literature with respect to the five properties (Section~\ref{sec:comparison-discussion}), discuss the principal source of its constant-factor overhead (Section~\ref{sec:discuss-incremental}), and outline future directions (Section~\ref{sec:future-directions}).

\subsection{Comparison with Related Approaches}
\label{sec:comparison-discussion}

Given the presentation and evaluation of N-DCA, it is now possible to compare the approach with three distributed coalition value calculation baselines:
the DCVC  family of algorithms \citep{Rahwan2005,Rahwan2007,Michalak2010,Voice12}; 
VBFR \citep{Vinyals2012}; 
and the SK algorithm \citep{shehory95,shehory96,shehory98}, 
based on  the five properties listed in Section~\ref{sec:intro}.

\begin{enumerate}
\item \textbf{Eliminating the need for communication between agents:}
    The allocation of coalition value calculations in N-DCA depends solely on the agent's own identifier $x$, the increment array (IA) $\underline{t}$ (Definition~\ref{def:coalition-generation}) and the shared parameter $n$.
    A similar mechanism exists for the DCVC and VBFR algorithms, whereby the agents have:
        (i) a pre-agreed ordering of the coalitions; and
        (ii) a pre-agreed algorithm to dictate which exact allocation of this ordering each agent should calculate.
    Conditions (i) and (ii) do not hold for the SK algorithm, thus agents need to coordinate their actions through communication to verify which coalitions will be in each agent's allocation.
    
    It should be noted that for any distributed coalition value calculation algorithm, if the agents want to complete the decentralised coalition formation process, communication costs will be incurred later when agents need to communicate their best coalitions or coalition structures found (e.g.\ \citet{sandholm99,Michalak2010}).
    Even when the best coalition structure is known, further communication may be necessary if the agents need to negotiate on their final utility payoff (e.g.\ \citet{WuSIAM1977,Cesco98,Shehory99,Arnold2002,Lehrer2003,Goradia2007}).

\item \textbf{Ensuring that the allocation of coalitions is equitable:}
    Of the four algorithms under consideration, only N-DCA and the DCVC algorithms guarantee that the agents' coalition value calculation allocations (referred to as \emph{shares} by \citet{Rahwan2007}) will be approximately equal (Section~\ref{sec:loadbalancing}).
    VBFR does not generate balanced allocations to the agents (in the case where all the agents are assumed to be in a fully connected graph).
    In this case, for $n$ agents, agent~1 is always assigned $2^{n-1}$ coalitions because there are $2^{n-1}$ coalitions where agent ID~1 is the smallest ID of that coalition.
    Agent~$n$ on the other hand is always assigned only one coalition because there is only one coalition with agent~$n$ as the smallest ID, which is the singleton coalition $\{n\}$.
    The SK algorithm has no guarantees on the maximum difference between any two agents' allocations; the average difference grows exponentially with the number of agents, as detailed in~\citet{Rahwan2007}.

\item \textbf{Every possible coalition value is calculated once and only
    once, thus eliminating redundancy:}
    One and only one canonical representative from each equivalence class is used to generate the coalition value allocation in N-DCA (Theorem~\ref{thm:main2}). 
    Neither the DCVC algorithm nor VBFR generate any redundant coalition value calculation issues, as the agents are aware how the coalitions are ordered in both algorithms, and therefore each agent's allocation does not ``accidentally'' overlap.
    The SK algorithm does result in an exponentially large redundancy because, as stated in \citet{Rahwan2007}, each agent commits to calculating the value of a set of coalitions with limited knowledge of the other agent's commitments.

\item \textbf{Maintaining a balanced load across agents:}
    Beyond receiving approximately equal \emph{numbers} of coalitions (property 2), it is also desirable that agents perform approximately equal amounts of \emph{computation}.
    The DCVC algorithms do not guarantee an equal number of operations (where the operations are comparisons and additions) when generating each agent's coalition value calculation share, even when the shares are equally sized~\citep{Rahwan2007}.
    This is due to the lexicographical ordering of the coalitions that the DCVC algorithm uses.
    For N-DCA, as each IA of the same size requires the exact same number of addition operations to find the corresponding coalition, equal operations will be performed by agents that are allocated an equal number of coalitions of each size.
    This is because, unlike the DCVC algorithm, N-DCA does not rely on lexicographical order to generate the coalitions themselves.
    The DCVC algorithm keeps track of which agent was responsible for previously generated allocations (through an $\alpha$ pointer) to generate approximately equal coalition value sets.
    A similar approach is used by N-DCA's rotated designation scheme (Definition~\ref{def:rotated-designation}) for periodic IAs (Definition~\ref{def:period}), with the per-size balance of Theorem~\ref{thm:persize-balance} and the aggregate balance of Theorem~\ref{thm:aggregate-balance} providing formal guarantees.

\item \textbf{An agent is \emph{self-interested}; i.e.\ it is only assigned
coalitions in which it is a member:}
    N-DCA, like the SK and VBFR algorithms before it, guarantees that every coalition distributed to an agent $x$ includes $x$ as a member (Equation~\eqref{eqn:valuex}). 
    This is not the case in any of the DCVC algorithms, and is a weakness in certain domains.  
    Note that self-interest is independent of load balance: even when allocations are not perfectly equal (e.g.\ under the per-size offset variant, where aggregate imbalance can reach $\kappa(n)$), every agent still computes only coalitions of which it is a member (Section~\ref{sec:balance_evaluation}).
    
    Consider an $e$-commerce environment where agents (representing a single business) form a temporary coalition in order to gain price discounts and economies of scale. As there may be anti-monopoly penalties or other considerations, such as communication and logistic costs, then this problem cannot simply be solved by forming the grand coalition.
    Therefore the agents together would have to calculate the values of the different coalitions to find their most preferable ones.
    If the values involved in the calculations were public knowledge, then this situation could be treated as a characteristic function game, and the coalition value calculation costs could be shared with the other agents in the system. 
    In this situation, it would make no sense for an agent to calculate the value of a coalition $C$ that does not include itself, as the agent will gain no benefit if $C$ were to form.
    A similar argument is made for VBFR for the smart grid domain~\citep{Vinyals2012}.  Although N-DCA guarantees that agents evaluate only coalitions containing themselves, the possibility of strategic misreporting within those coalitions remains; the implications for mechanism design are discussed in Section~\ref{sec:future-directions}.
    
\end{enumerate}

\begin{table}[t]
\centering
\caption{Comparison of the properties of coalition value allocation algorithms.}
\label{tab:comp}

\setlength{\tabcolsep}{9pt}
\renewcommand{\arraystretch}{1.2}

\begin{tabular}{lcccc}
\toprule
Property & SK & DCVC & VBFR & \textbf{N-DCA} \\
\midrule
1. Eliminates Communication    & $\times$ & $\checkmark$ & $\checkmark$ & $\checkmark$ \\
2. Equitable Allocation        & $\times$ & $\checkmark$ & $\times$     & $\checkmark$ \\
3. Eliminates Redundancy       & $\times$ & $\checkmark$ & $\checkmark$ & $\checkmark$ \\
4. Balanced Load               & $\times$ & $\times$     & $\times$     & $\checkmark$ \\
5. Self-Interest                & $\checkmark$ & $\times$ & $\checkmark$ & $\checkmark$ \\
\bottomrule
\end{tabular}

\end{table}

This discussion, summarised in Table~\ref{tab:comp}, supports the claim that N-DCA
satisfies all the above properties, unlike the other algorithms considered.

\subsection{Towards an Incremental Increment-Array Update}
\label{sec:discuss-incremental}

The current implementation of \textsc{GenIncArray}
(Figure~\ref{algorithm:genIncArr}) performs a complete
$\mathcal{O}(n)$ scan of the necklace array each time a new
necklace is produced by FKM.  This stands in contrast to FKM
itself, which generates successive necklaces in constant amortised
time (CAT) by modifying only a suffix of the
array~\citep{RUSKEY1992414}, leaving positions $1, \ldots, i{-}1$
unchanged.  This locality suggests that an \emph{incremental}
update of the increment array may be possible, maintaining the IA
across successive FKM steps and recomputing only the entries
affected by the modified suffix.  The principal complication is
that the IA is a \emph{cyclic} run-length encoding: changes near
the end of the necklace can affect the wrap-around entry even when
the prefix is unchanged.  Nevertheless, since FKM's own suffix
copies amortise to $\mathcal{O}(1)$ per
necklace~\citep{RUSKEY1992414}, there is reason to expect that the
corresponding IA updates would amortise similarly.

If such a scheme could achieve $\mathcal{O}(1)$ amortised cost per
necklace, the practical impact would be substantial: the breakdown
data of Section~\ref{sec:eval-breakdown} show that reducing
\textsc{GenIncArray}'s cost to that of FKM's own stepping time
would lower the N-DCA\,/\,DCVC ratio from approximately
$5.6{\times}$ to approximately $2.6{\times}$.  The present article
deliberately maintains a clean separation between FKM and
\textsc{GenIncArray} for analytical clarity; an incremental scheme
would trade this modularity for performance.  The design,
implementation, and formal analysis of such a scheme is left to
future work.

\subsection{Broader Implications and Future Directions}
\label{sec:future-directions}
  
The empirical evaluation of Section~\ref{sec:evaluation} covers
the range $2 \le n \le 25$, which spans the regime in which
exhaustive coalition value calculation remains computationally
tractable.  For significantly larger populations, the exponential
number of coalitions ($2^n - 1$) makes exhaustive evaluation
infeasible regardless of how the workload is distributed.  Recent
centralised approaches such as SALDAE~\citep{Taguelmimt2025SALDAE}
have demonstrated scalability to thousands of agents by adopting
anytime search strategies that avoid exhaustive enumeration.
N-DCA operates in a complementary regime: its contribution is not
to extend the frontier of tractable~$n$, but rather to ensure
that within the range where exhaustive calculation is viable, the
distribution across agents is fair, communication-free, and
self-interested.  For applications with moderate numbers of
agents, which remain common in domains such as e-commerce,
logistics, and smart-grid energy trading, these guarantees are
directly relevant.
 
The self-interest property has implications beyond the
value-calculation stage itself.  In a fully decentralised
coalition formation process, value calculation is typically
followed by structure search, negotiation, and payoff
division~\citep{sandholm99,RahwanAIJ2015}.  If agents are assigned coalitions in which they are not members (which can occur when using DCVC), they may have an incentive to misreport values strategically, since they bear no cost from a coalition they would never join.  
By ensuring that every agent evaluates only coalitions containing itself, N-DCA reduces the scope for such manipulation.
More fundamentally, it also eliminates the prior problem of rational non-compliance: an agent $x$ assigned $\nu(C)$ for $C \not\ni x$ has no individually rational reason to perform the computation at all, whether or not it intends to misreport; the self-interest property ensures that every assignment is one the agent has a direct incentive to execute correctly.

Nevertheless, self-interested agents may still misreport values
for coalitions of which they \emph{are} members, for example by
understating the value of a rival coalition in order to steer
the structure-search stage toward a more personally favourable
outcome.  Addressing this residual vulnerability would require
wrapping N-DCA within a mechanism-design
framework~\citep{Nisan2007}, where the coalition value
calculation phase is embedded in a mechanism that constitutes a
Bayesian game.  If the resulting mechanism has Bayesian
equilibria in which the optimal strategy for each agent is to
misreport information, then by the \emph{revelation
principle}~\citep{Myerson1979,Shoham2008} there exists a
payoff-equivalent, individually rational mechanism with an
equilibrium in which agents truthfully report their types.  
The design and analysis of such a truthful mechanism for distributed coalition value calculation remains an open problem.
A related open concern is privacy: computing $\nu(C)$ for $C \not\ni x$ may require agent $x$ to access capability or valuation data belonging to members of $C$, representing a structural information-leakage risk that is independent of any intent to deceive, and that privacy-preserving mechanism design~\citep{Nisan2007} should address. 

A natural direction for future work is to investigate whether N-DCA
can serve as a replacement for DCVC within a distributed coalition
structure generation framework such as D-IP \citep{Michalak2010}.
The work on D-IP has demonstrated that DCVC's distributed value-calculation stage can
be combined with a distributed version of the IP search algorithm
\citep{Rahwan009} to solve the full coalition structure generation
problem without a centralised coordinator.  However, as D-IP
inherits DCVC as its value-calculation substrate, it also inherits
DCVC's limitations; agents may be assigned coalitions of which they are
not members, and the approach is constrained by DCVC's 64-bit index
representation.
Replacing DCVC with N-DCA in the first stage of D-IP would bring
several potential advantages.  First, every agent would compute values
only for coalitions in which it is a member, preserving the
self-interest guarantee and reducing the incentive for misreporting in
adversarial settings.  Second, the necklace-based encoding avoids the
large-integer indices that limit DCVC's scalability.  The principal
challenge would lie in adapting D-IP's communication stages to N-DCA's
allocation scheme: in particular, the aggregate statistics ($Max_s$,
$Avg_s$) and filter rules (FR1, FR2) that D-IP computes during and
after the value-calculation stage would need to be derived from
N-DCA's necklace-based coalition shares rather than from DCVC's
reverse-lexicographic lists.  
Exploring this combination (i.e.\ a fully decentralised, self-interested coalition structure generation algorithm) is a promising direction for future research.

 
More broadly, the mathematical framework developed in Section~\ref{sec:ippe} may be applicable beyond coalition value calculation.  
The core construction: partitioning combinatorial objects with cyclic symmetry into equivalence classes, selecting canonical representatives, and distributing them across agents via a rotated designation scheme, is independent of the coalition-formation domain.  Any distributed
computation requiring a fair, communication-free allocation of
work items that exhibit circular structure (for example, cyclic
task schedules, distributed enumeration of chemical
isomers, or the generation of necklace codes in wireless
networks~\citep{8292573}) could potentially exploit the same
approach.  Investigating such generalisations is a natural
direction for future work.
  
Finally, we note several limitations of the present work.
N-DCA assumes a characteristic function game in which coalition
values depend only on membership; games with externalities (where
the value of a coalition depends on the actions of non-members)
are not addressed.  The approach also requires that all agents know
the total number of agents~$n$ in advance, which may not hold in
open or dynamic environments.  Furthermore, as demonstrated in
Section~\ref{sec:evaluation}, N-DCA incurs a constant-factor
overhead relative to DCVC in generation time, attributable
primarily to the per-necklace $\mathcal{O}(n)$ cost of
\textsc{GenIncArray}.  Although this overhead is negligible for
any non-trivial characteristic-function evaluation cost
(Section~\ref{sec:eval-amortised}), it may matter in settings
where the characteristic function is extremely cheap to evaluate
or where generation must complete within a strict real-time
budget; precisely the scenario that motivates the incremental
update scheme discussed in Section~\ref{sec:discuss-incremental}.

\section{Conclusions and Future Work}
\label{sec:conclusions}

This article presents N-DCA, a necklace-based distributed algorithm for allocating coalition value calculations in unrestricted characteristic function games.
The method exploits combinatorial necklaces to generate increment arrays (IAs) for each equivalence class.
These canonical representatives are in turn used to generate coalition value calculation allocations to each agent such that the allocation is \emph{fair}, both in the number of coalitions assigned to each agent and in the computational burden of generating those assignments, and non-redundant.
By exploiting different two-colour necklace generation algorithms (for example those that generate all combinations, or those that generate fixed density necklaces), the computational cost of generating the coalition value allocation can be managed, and for most algorithms, the complexity has been proven to be \emph{constant amortised time} (CAT) with respect to the number of necklaces (and hence number of increment arrays) generated.
Furthermore, N-DCA is suited for scenarios that include self-interested agents, as agents are guaranteed to be a member of all of the coalitions that they are allocated.
We prove that each coalition value is calculated exactly once and that the resulting allocations satisfy formal balance guarantees in both size and computational effort.
Within the setting considered, N-DCA appears to be the only approach with formal guarantees for all five target properties identified in Section~\ref{sec:intro} \emph{simultaneously}: communication-free allocation, equitable distribution, elimination of redundancy, balanced computational load, and self-interest.

Several directions for future work have been identified in the course of this article, including the development of an incremental update scheme for \textsc{GenIncArray} that could reduce N-DCA's constant-factor overhead (Section~\ref{sec:discuss-incremental}), and the integration of N-DCA with distributed coalition structure generation frameworks such as D-IP to bring the self-interest guarantee to the full structure search problem (Section~\ref{sec:future-directions}).

\backmatter
\bmhead{Acknowledgments}
We would like to thank the anonymous reviewers
for their valuable comments and feedback, and acknowledge the assistance of previous collaborators for their support and comments during the earlier phase of the work.

\begin{appendices}

\section{Proofs for the correctness of IAs}
\label{sec:appendix}

In this appendix, we restate the relevant assumptions from Section~\ref{sec:ippe} and provide the full proofs for Lemma~\ref{lemma:same-sets} and Theorem~\ref{thm:main}.
We assume that there are $n$ agents (typically $n > 2$) represented by $Ag = \{1, 2, \ldots, n\}$.
A coalition $C \subseteq Ag$ of size $s = |C|$ represents a sequence of agents, where $n,s \in \mathbb{N}$.
As the agents are represented as natural numbers, we say that two integers $a$ and $b$ are congruent modulo n (i.e.\ $a\equiv b{\pmod {n}}$) where $a,b \in \set{1 \ldots n}$.
We use $\underline{t}$ to denote an arbitrary sequence $\tuple{t_0,t_1,\ldots,t_{s-1}}$ such that $\sum_{i=0}^{s-1}~t_i~=~n-s$ (see also Lemma~\ref{lemma:offset-increment}).
We call such a sequence an \emph{Increment Array} (IA), whose \emph{period}, denoted by $\pi(\underline{t})$ is:
\[
\min_{1\leq p\leq s}~\mbox{ }~
\underline{t}~~=~~\tuple{t_0,t_1,\ldots,t_{p-1},t_0,t_1,\ldots,t_{p-1},\ldots,t_0,t_1,\ldots,t_{p-1}}
\]
where $\underline{t}$ is formed by $\mu$ identical copies of a
subsequence of length $p = \pi(\underline{t})$; i.e.\ the period of $\underline{t}$
is the \emph{length of the smallest subsection} of $\underline{t}$ that is repeated throughout $\underline{t}$.
Given $C\subseteq\set{1,2,\ldots,n}$ and $(x,\underline{t})$ (for some $1\leq x\leq n$) we say that 
$\underline{t}$ \emph{generates} $C$ from $x$ if 
$C=\set{x_1,x_2,\ldots,x_s}$ using (\ref{eqn:valuex}), where $\varphi$ (given
in (\ref{eqn:offset}) and repeated here for clarity) is the \emph{cumulative integer increment}:

\[
\varphi_i~=~\left\{
{
\begin{array}{lcl}
0&\mbox{ if }&i=1\\
\sum_{k=0}^{i-2}~(t_k + 1)&\mbox{ if }&2\leq i\leq s+1
\end{array}
}
\right.
\]

Note that $\varphi_{s+1}=n$, as this is equivalent to the sum of all offset
increments (which, by Lemma~\ref{lemma:offset-increment}, is $n-s$) and $s$ baseline increments.
Furthermore, Lemma~\ref{lemma:ns-sequence} states that for any coalition $C$, there is always some IA
and an agent that can generate it.

\subsection*{Proof of Lemma~\ref{lemma:same-sets}}

Definition~\ref{def:ia-equivalence} states that two IAs $\underline{t}$ and
$\underline{u}$ of the same
size $s$ are \emph{equivalent} $\underline{t} \approx
\underline{u}$, if $\underline{u}$ is a circular shift of
$\underline{t}$.
Lemma~\ref{lemma:same-sets} states that the same coalitions will be generated by any IA belonging to the same equivalence class $[\underline{t}]_\approx$.

\newtheorem*{lem:same-sets}{Lemma~\ref{lemma:same-sets}}
\begin{lem:same-sets}
If $\underline{t}\approx\underline{u}$ then
\[
\bigcup_{i=1}^{n}~\set{~C(i,\underline{t})~}~~=~~\bigcup_{i=1}^{n}~\set{~C(i,\underline{u})~}
\]
\end{lem:same-sets}
\begin{proof}
Without loss of generality we may assume that $\underline{t}\approx\underline{u}$ is witnessed by the choice $k=s-1$, i.e.\
$\tuple{u_0,u_1,\ldots,u_{s-1}}=\tuple{t_{s-1},t_0,t_1,\ldots,t_{s-3},t_{s-2}}$.
Define $\varphi_r$ for $1\leq r\leq s+1$ as before and $\psi_r$ for $1\leq r\leq s+1$ via:
\[
\psi_r~~=~~\left\{
{
\begin{array}{lcl}
0&\mbox{ if }&r=1\\
(t_{s-1}+1)~+~\sum_{k=0}^{r-3}(t_k+1)&\mbox{ if }&2\leq r\leq s+1
\end{array}
}
\right.
\]
Comparing respective terms we see that for all $2\leq k\leq s$, we have:
$\psi_k=\varphi_k+(t_{s-1}-t_{k-2})$.
We claim that this leads to the following:
\[
C(i,\underline{t})~~=~~\left\{
{
\begin{array}{lcl}
C(n-t_{s-1}+i-1,\underline{u})&\mbox{ if }&1\leq i\leq t_{s-1}+1\\
C(i-t_{s-1}-1,\underline{u})&\mbox{ if }&t_{s-1}+2\leq i\leq n
\end{array}
}
\right.
\]
To see this, consider the case when $1\leq i\leq t_{s-1}+1$. We have
$C(i,\underline{t})=\bigcup_{k=1}^{s}\set{i+\varphi_k}$ which is claimed to be:
{
  \medmuskip=1mu
  \thinmuskip=1mu
  \thickmuskip=1mu
\begin{eqnarray*}
C(n-t_{s-1}+i-1,\underline{u})&\mbox{ $=$ }&\bigcup_{k=1}^{s}\set{n-t_{s-1}+i-1+\psi_k}\\
&\mbox{ $=$ }&\set{n-t_{s-1}+i-1}\cup\bigcup_{k=2}^{s}\set{n-t_{s-1}+i-1+\varphi_k+t_{s-1}-t_{k-2}}
\end{eqnarray*}
}
Consider the terms $n-t_{s-1}+i-1+\varphi_k+t_{s-1}-t_{k-2}$.
For $2\leq k\leq s$, from the fact that $\varphi_k=\sum_{j=0}^{k-2}t_j+k-1$, these are equal to: 
\[
n-t_{s-1}+i-1+\sum_{j=0}^{k-3}t_j+k-1+t_{s-1}=n+i+\varphi_{k-1}
\]
In total we have, for $1\leq i\leq t_{s-1}+1$:
\begin{eqnarray*}
C(i,\underline{t})&\mbox{ $=$ }&\bigcup_{k=1}^{s}~\set{i+\varphi_k}\\
C(n-t_{s-1}+i-1,\underline{u})&\mbox{ $=$ }&\set{n-t_{s-1}+i-1}~\cup~\bigcup_{k=2}^{s}~\set{n+i+\varphi_{k-1}}
\end{eqnarray*}
Noting that after $n$ correction terms $n+i+\varphi_{k-1}$ become $i+\varphi_{k-1}$ all of which
are elements of $C(i,\underline{t})$, the only terms unaccounted for are
$\set{n-t_{s-1}+i-1}\in C(n-t_{s-1}+i-1,\underline{u})$ and $\set{i+\varphi_s}\in C(i,\underline{t})$. For these, however:
\begin{eqnarray*}
i+\varphi_s&\mbox{ $=$ }&i+\sum_{j=0}^{s-2}t_j~+~s-1\\
&\mbox{ $=$ }&i+(n-s-t_{s-1})+s-1\\
&\mbox{ $=$ }&i+n-t_{s-1}-1
\end{eqnarray*}
When $t_{s-1}+2\leq i\leq n$, it is claimed that
$C(i,\underline{t})=\bigcup_{k=1}^{s}\set{i+\varphi_k}$
corresponds to:
\begin{eqnarray*}
C(i-t_{s-1}-1,\underline{u})&\mbox{ $=$ }&\bigcup_{k=1}^{s}\set{i-t_{s-1}-1+\psi_k}\\
&\mbox{ $=$ }&\set{i-t_{s-1}-1}\cup\bigcup_{k=2}^{s}\set{i-t_{s-1}-1+\varphi_k+t_{s-1}-t_{k-2}}
\end{eqnarray*}
Inspecting the terms for $2\leq k\leq s$, we have
$i-t_{s-1}-1+\varphi_k+t_{s-1}-t_{k-2}$.
These, again, simplify to $i+\varphi_{k-1}$, so that:
\begin{eqnarray*}
C(i,\underline{t})&\mbox{ $=$ }&\bigcup_{k=1}^{s}~\set{i+\varphi_k}\\
C(i-t_{s-1}-1,\underline{u})&\mbox{ $=$ }&\set{i-t_{s-1}-1}~\cup~\bigcup_{k=2}^{s}~\set{i+\varphi_{k-1}}
\end{eqnarray*}
When $1\leq k\leq s-1$, the term $i+\varphi_k$ appears in both $C(i,\underline{t})$ and $C(i-t_{s-1}-1,\underline{u})$,
For the terms $i+\varphi_s\in C(i,\underline{t})$ and $i-t_{s-1}-1\in C(i-t_{s-1}-1,\underline{u})$
we have already seen that $i+\varphi_s=i+n-t_{s-1}-1$ which after $n$ correction is 
$i-t_{s-1}-1$ as required.
This establishes the property claimed: if $\underline{t}$ and 
$\underline{u}$ belong to the same equivalence class of $\approx$ then:
\[
\bigcup_{i=1}^{n}~\set{C(i,\underline{t})}~~=~~\bigcup_{i=1}^{n}~\set{C(i,\underline{u})}
\]
\end{proof}

\subsection*{Proof of Theorem~\ref{thm:main}}

Given an \emph{arbitrary} sequence of non-negative integers, $\underline{b}~=~\tuple{b_0,\ldots,b_{s-1}}$ say, 
the $n$-\emph{correction} of $\underline{b}$ is the sequence, $corr(\underline{b},n)$, obtained by replacing each $b_i>n$ with the
value $b_i-n$. If $\underline{c}=corr(\underline{b},n)$ and $\underline{c}\not=\underline{b}$
then $\underline{b}$ is said to be \emph{uncorrected} for $n$.
Theorem~\ref{thm:main} states that
each IA can be used $r$ times to generate the same
coalition; i.e.\ the coalitions formed for agents $i$ and $j$ are equal if and only if there exists some 
$r$ in the range stated in Equation \eqref{eqn:thm1}.

\newtheorem*{thm:main}{Theorem~\ref{thm:main}}
\begin{thm:main}
For any IA $\underline{t}$, and for all $1\leq i\leq j\leq n$,
\begin{equation*}\tag{\ref{eqn:thm1}}
C(i,\underline{t})=C(j,\underline{t})~~~\Leftrightarrow~~~
\exists~0\leq r\leq \frac{(n-i)s}{n\pi(\underline{t})}~:~
j~=~i~+~r\times\left({\frac{n\pi(\underline{t})}{s}}\right)
\end{equation*}
\end{thm:main}

\begin{proof} ~

\noindent
{\bf ($\Leftarrow$) Direction.} 
We first show that $j=i+(rn\pi(\underline{t})/s)$ implies
that $C(i,\underline{t})$ and $C(j,\underline{t})$ are identical. The case $r=0$ gives $j=i$, which is trivial. 
It therefore suffices to establish the case $r = 1$; the general result then follows by applying this case $r$ times inductively.
Note that no assumptions are made concerning $i$ other than those prescribed by the requirement $1\leq i\leq n-\pi(\underline{t})n/s$.

The (uncorrected with respect to $n$) sequence generated by $C(i,\underline{t})$ 
is $\tuple{i+\varphi_1,i+\varphi_2,\ldots,i+\varphi_s}$,
and that by $C(j,\underline{t})$ (again uncorrected with respect to $n$) is:
\[
\tuple{i+(n\pi(\underline{t})/s)+\varphi_1,~i+(n\pi(\underline{t})/s)+\varphi_2,~\ldots,~i+(n\pi(\underline{t})/s)+\varphi_s}
\]
Consider the term $\varphi_{\pi(\underline{t})+k}$ for $1\leq k\leq s-\pi(\underline{t})$:
\begin{eqnarray*}
\varphi_{\pi(\underline{t})+k}&\mbox{ $=$ }&\sum_{p=0}^{\pi(\underline{t})+k-2}~t_p~+~\pi(\underline{t})+k-1\\
&\mbox{ $=$ }&\sum_{p=0}^{\pi(\underline{t})-1}~t_p~~+~\sum_{q=\pi(\underline{t})}^{\pi(\underline{t})+k-2}~t_q~+~\pi(\underline{t})+k-1\\
&\mbox{ $=$ }&\varphi_{\pi(\underline{t})+1}~+~\sum_{q=0}^{k-2}~t_q~+~k-1\\
&\mbox{ $=$ }&\varphi_{\pi(\underline{t})+1}~+~\varphi_{k}
\end{eqnarray*}
The penultimate identity follows from $t_{\pi(\underline{t})+q}~=~t_q$.
Now consider $\varphi_{\pi(\underline{t})+1}$. We have:
\[
\varphi_{\pi(\underline{t})+1}=\sum_{k=0}^{\pi(\underline{t})-1}~t_k~+~\pi(\underline{t})
\]
and
\begin{eqnarray*}
n-s~=~\sum_{k=0}^{s-1}~t_k&\mbox{ $=$ }&\sum_{m=0}^{s/\pi(\underline{t})-1}~\sum_{k=m\pi(\underline{t})}^{(m+1)\pi(\underline{t})-1}~t_k\\
&\mbox{ $=$ }&\sum_{m=0}^{s/\pi(\underline{t})-1}~\sum_{k=m\pi(\underline{t})}^{(m+1)\pi(\underline{t})-1}~t_{k-m\pi(\underline{t})}\\
&\mbox{ $=$ }&\frac{s}{\pi(\underline{t})}\sum_{p=0}^{\pi(\underline{t})-1}~t_p\\
&\mbox{ $=$ }&\frac{s}{\pi(\underline{t})}\left({\varphi_{\pi(\underline{t})+1}-\pi(\underline{t})}\right)
\end{eqnarray*}
which gives:
\begin{equation}
\varphi_{\pi(\underline{t})+1}~~=~~\frac{n\pi(\underline{t})}{s}
\end{equation}
A direct computation using periodicity then gives, for $1\leq k\leq s-\pi(\underline{t})$:
\begin{equation}
\varphi_{\pi(\underline{t})+k}~~=~~\frac{n\pi(\underline{t})}{s}~+~\varphi_k
\end{equation}
Returning to the (uncorrected) sequences defining $C(i,\underline{t})$ and $C(i+n\pi(\underline{t})/s,\underline{t})$, the former
contributes:
\[
\bigcup_{k=1}^{s}~\set{i~+~\varphi_k}~~=~~\bigcup_{m=0}^{s/\pi(\underline{t})-1}~
\bigcup_{k=1}^{\pi(\underline{t})}~\set{i+m\varphi_{\pi(\underline{t})+1}+\varphi_k}
\]
which is:
\begin{equation}
\bigcup_{m=0}^{s/\pi(\underline{t})-1}~\bigcup_{k=1}^{\pi(\underline{t})}~\left\{{i+\frac{mn\pi(\underline{t})}{s}+\varphi_k}\right\}\label{eqn:appx1}
\end{equation}
And, by a similar analysis of $j=i+(n\pi(\underline{t})/s)$, 
we obtain, prior to correction,
$C(j,\underline{t})$ as:
\begin{equation}
\bigcup_{m=0}^{s/\pi(\underline{t})-1}~\bigcup_{k=1}^{\pi(\underline{t})}~\left\{{i+\frac{(m+1)n\pi(\underline{t})}{s}+\varphi_k}\right\}\label{eqn:appx2}
\end{equation}
Comparing (\ref{eqn:appx1}) and (\ref{eqn:appx2}), the terms with $1\leq m<s/\pi(\underline{t})-1$ clearly occur in both sets. The terms for $m=0$
in (\ref{eqn:appx1}) are:
\[
\bigcup_{k=1}^{\pi(\underline{t})}~\set{i+\varphi_k}
\]
Similarly, the terms corresponding
to $m=s/\pi(\underline{t})-1$ in (\ref{eqn:appx2}), are:
\[
\bigcup_{k=1}^{\pi(\underline{t})}~\left\{{i+\frac{s}{\pi(\underline{t})}\frac{n\pi(\underline{t})}{s}+\varphi_k}\right\}
\]
These terms are prior to $n$ correction, so that:
\[
i+\frac{s}{\pi(\underline{t})}\frac{n\pi(\underline{t})}{s}+\varphi_k\in C(j,\underline{t})~~=~~i~+~n~+~\varphi_k~~=~~i+\varphi_k\in
C(i,\underline{t})
\]
This establishes the first part of the Theorem: 
\[
j~=~i+r\left({\frac{n\pi(\underline{t})}{s}}\right)~~\Rightarrow~~C(i,\underline{t})=C(j,\underline{t})
\]

\noindent
{\bf ($\Rightarrow$) Direction.} 
To complete the proof we show that:
\[
C(i,\underline{t})=C(j,\underline{t})~~~\Rightarrow~~~
\exists~0\leq r\leq \frac{(n-i)s}{n\pi(\underline{t})}~:~
j~=~i~+~r\times\left({\frac{n\pi(\underline{t})}{s}}\right)
\]

\noindent
Suppose $C(i,\underline{t})=C(j,\underline{t})$; i.e.\ we can observe that
if $C'=C(j,\underline{t})\subseteq\set{1,2,\ldots,n}$ of size $s$ is the same set as
$C=C(i,\underline{t})=\set{x_1,\ldots,x_s}$ there must be some
$2\leq \rho\leq s$ such that:
\[
corr(\tuple{i+\varphi_\rho+\varphi_1,i+\varphi_\rho+\varphi_2,\ldots,i+\varphi_\rho+\varphi_p,\ldots,i+\varphi_\rho+\varphi_s},n)
\]
is exactly the set $C(i,\underline{t})$ whose elements are:
\[
corr(\tuple{i+\varphi_1,i+\varphi_2,i+\varphi_3,\ldots,i+\varphi_p,\ldots,i+\varphi_s},n)
\]
From $C(i,\underline{t})=C(j,\underline{t})$ we can deduce:
$i\in C(i,\underline{t})$ and each element contributing to $C(i,\underline{t})$ (prior to $n$-correction) 
has the form $i+\varphi_k$ so
that there is a choice, $\rho$, for which $j\in\set{i+\varphi_\rho,i+\varphi_\rho-n}$.

The terms generated by $(i,\underline{t})$ are:
$\beta(i)=corr(\tuple{i+\varphi_1,\ldots,i+\varphi_p,\ldots,i+\varphi_s},n)$,
whereas the terms generated by $(i+\varphi_\rho,\underline{t})$ are
$\beta(\rho)=corr(\tuple{i+\varphi_\rho+\varphi_1,\ldots,i+\varphi_\rho+\varphi_p,\ldots,i+\varphi_\rho+\varphi_s},n)$.
From the premise $C(i,\underline{t}) = C(i+\varphi_\rho,\underline{t})$, every term in $corr(\beta(\rho),n)$ must correspond to some term in $corr(\beta(i),n)$
such that:
\[
\forall~p~\exists~q~~:~~i+\varphi_p~\in~\set{i+\varphi_\rho+\varphi_q,i+\varphi_\rho+\varphi_q-n}
\]
After some trivial rearrangement this is simply:
\[
\forall~p\exists~q:\varphi_p\in\set{\varphi_\rho+\varphi_q,\varphi_\rho+\varphi_q-n}
\]
For any sequence
$
\tuple{y_1,\ldots,y_s}
$
generated via $(y_1,\underline{t})$, there is at most one index $\lambda$ for which:
$y_i~\leq~n~~\mbox{$\forall~i\leq \lambda$}$, $y_{\lambda+1}~>~n$, and $y_i-n~\leq~n~~\forall~\lambda+2\leq i\leq s$.  Consequently,
the values $\tuple{y_1,\ldots,y_\lambda}$ are strictly increasing, as are the values
$\tuple{y_\lambda-n,y_{\lambda+1}-n,\ldots,y_s-n}$. We can therefore deduce that for a particular
$\varphi_\rho$ witnessing the behaviour of $\underline{t}$ in the analysis above, there is a unique
index $p$ for which:
\[
\varphi_\rho+\varphi_q~~\mbox{ is }~~\left\{
{
\begin{array}{lcl}
\leq n&\mbox{ if }&q\leq p\\
>n&\mbox{ if }&q>p
\end{array}
}
\right.
\]
From $i+\varphi_\rho+\varphi_p=i+n=i+n+\varphi_1$, it follows that:
\[
\begin{array}{lcl}
\varphi_\rho+\varphi_p&\mbox{ $=$ }&n+\varphi_1\\
\varphi_\rho+\varphi_{p+1}&\mbox{ $=$ }&n+\varphi_2\\
&\cdots&\\
\varphi_\rho+\varphi_{p+k}&\mbox{ $=$ }&n+\varphi_{k+1}\\
&\cdots&\\
\varphi_\rho+\varphi_{p+(s-p)}&\mbox{ $=$ }&n+\varphi_{s-p+1}\\
&&\\
\varphi_\rho+\varphi_{p-1}&\mbox{ $=$ }&\varphi_s\\
\varphi_\rho+\varphi_{p-2}&\mbox{ $=$ }&\varphi_{s-1}\\
&\cdots&\\
\varphi_\rho+\varphi_{p-k}&\mbox{ $=$ }&\varphi_{s-k+1}\\
\varphi_\rho+\varphi_{p-(p-2)}&\mbox{ $=$ }&\varphi_{s-p+3}\\
\varphi_\rho+\varphi_{p-(p-1)}&\mbox{ $=$ }&\varphi_{s-p+2}\\
\end{array}
\]
In consequence, $s$ and $\rho$, uniquely determine the value of $p$ as $s-\rho+2$.
Rearranging the expression above, gives
\[
\begin{array}{lcl}
n+\varphi_1&\mbox{ $=$ }&\varphi_\rho+\varphi_{(s-\rho+2)}\\
n+\varphi_2&\mbox{ $=$ }&\varphi_\rho+\varphi_{(s-\rho+2)+1}\\
&\cdots&\\
n+\varphi_{k+1}&\mbox{ $=$ }&\varphi_\rho+\varphi_{(s-\rho+2)+k}\\
&\cdots&\\
n+\varphi_{s-(s-\rho+2)+1}&\mbox{ $=$ }&\varphi_\rho+\varphi_{(s-\rho+2)+(s-{s-\rho+2})}\\
&&\\
\varphi_{s-(s-\rho+2)+2}&\mbox{ $=$ }&\varphi_\rho+\varphi_{(s-\rho+2)-((s-\rho+2)-1)}\\
\varphi_{s-(s-\rho+2)+3}&\mbox{ $=$ }&\varphi_\rho+\varphi_{(s-\rho+2)-((s-\rho+2)-2)}\\
&\cdots&\\
\varphi_{s-k+1}&\mbox{ $=$ }&\varphi_\rho+\varphi_{(s-\rho+2)-k}\\
&\cdots&\\
\varphi_{s-1}&\mbox{ $=$ }&\varphi_\rho+\varphi_{(s-\rho+2)-2}\\
\varphi_s&\mbox{ $=$ }&\varphi_\rho+\varphi_{(s-\rho+2)-1}
\end{array}
\]
Giving the final set of identities that $\underline{t}$ must satisfy
in order for $C(i,\underline{t})=C(i+\varphi_\rho,\underline{t})$ after $n$-correction. Recall that $\varphi_{s+1}=n$:
\[
\varphi_\rho~+~\varphi_q~~=~~\left\{
{
\begin{array}{lcl}
\varphi_{\rho+q-1}&\mbox{ if }&2\leq q\leq s-\rho+1\\
\varphi_{s+1}+\varphi_{\rho+q-1-s}&\mbox{ if }&s-\rho+2\leq q\leq s
\end{array}
}
\right.
\]
where 
\[
\varphi_\rho~=~\left\{
{
\begin{array}{lcl}
0&\mbox{ if }&\rho=1\\
\sum_{k=0}^{\rho-2}~t_k~+~(\rho-1)&\mbox{ if }&2\leq \rho\leq s
\end{array}
}
\right.
\]
We first observe that these systems of identities can be expressed solely in terms of $\tuple{t_0,\ldots,t_{s-1}}$:
the terms on the left-hand side contribute $\rho-1+q-1=\rho-q-2$ in addition to the $t_k$ values. The terms on the
right-hand side, however, also contribute $\rho+q-3$ (via
$s+(\rho+q-2-s)$ when $s-\rho+2\leq q\leq s$) in addition to the $t_k$ values, so that these cancel. 

We have so far shown that
if after $n$-correction, $C(i,\underline{t})=C(i+\varphi_\rho,\underline{t})$ 
for some value $\rho$ with $2\leq \rho\leq s$,  $\underline{t}=\tuple{t_0,\ldots,t_{s-1}}$ is a solution
for the system of identities:
\begin{eqnarray}
\sum_{k=0}^{\rho-2} t_k+\sum_{k=0}^{q-2}t_k=&&\sum_{k=0}^{\rho+q-3}t_k\mbox{$~$ $~$ $~$ $~$ $~$}2\leq q\leq s-\rho+1\label{eqn:appx3}\\
\sum_{k=0}^{\rho-2} t_k+\sum_{k=0}^{q-2}t_k=&&\sum_{k=0}^{s-1}t_k+\sum_{k=0}^{\rho+q-3-s}t_k\mbox{$~$ $~$ $~$}s-\rho+2\leq q\leq s\label{eqn:appx4}
\end{eqnarray}
We now show that any such solution must have $\pi(\underline{t})\leq \rho-1$.
Let $\underline{t}$ be an IA that satisfies
the system of identities given in (\ref{eqn:appx3}) and (\ref{eqn:appx4}) using $\rho$. In order to simplify the notation we use
$\eta(n,s,\rho,q)$ to denote the relevant identity.
We analyse the general behaviours of:
\[
\begin{array}{ll}
\eta(n,s,\rho,\rho-p)&\mbox{ for $1\leq p\leq \rho-2$}\\
\eta(n,s,\rho,\rho+p)&\mbox{ for $0\leq p\leq s-\rho$}
\end{array}
\]
For $\eta(n,s,\rho,\rho-p)$ (\ref{eqn:appx3}) and (\ref{eqn:appx4}) yield the identities (post simplification and rearrangement):
\begin{eqnarray}
\sum_{k=0}^{\rho-p-2}~t_k~=~\sum_{k=\rho-1}^{2\rho-p-3}~t_k&&2\rho-s-1\leq p\leq \rho-2\label{eqn:appx5}\\
\sum_{k=2\rho-p-2-s}^{\rho-p-2}~t_k~=~\sum_{k=\rho-1}^{s-1}~t_k&&1\leq p\leq 2\rho-s-2\label{eqn:appx6}
\end{eqnarray}
For $\eta(n,s,\rho,\rho+p)$ we obtain, in a similar manner:
\begin{eqnarray}
\sum_{k=0}^{\rho-2}~t_k~=~\sum_{k=\rho+p-1}^{2\rho+p-3}~t_k&&0\leq p\leq s-2\rho+1\label{eqn:appx7}\\
\sum_{k=2\rho+p-2-s}^{\rho+p-2}~t_k~=~\sum_{k=\rho-1}^{s-1}~t_k&&s-2\rho+2\leq p\leq s-\rho\label{eqn:appx8}
\end{eqnarray}
Consider (\ref{eqn:appx5}) and (\ref{eqn:appx6}). As $p$ \emph{decreases} from its maximum value (for these cases) of $\rho-2$ to
its minimum of $1$ the system of identities follows the pattern:
{
  \medmuskip=1mu
  \thinmuskip=1mu
  \thickmuskip=1mu
\[
\begin{array}{lclr}
t_0&\mbox{ $=$ }&t_{\rho-1}&\mbox{ $p=\rho-2$, $q=2$}\\
t_0+t_1&\mbox{ $=$ }&t_{\rho-1}+t_\rho&\mbox{ $p=\rho-3$, $q=3$}\\
t_0+t_1+t_2&\mbox{ $=$ }&t_{\rho-1}+t_\rho+t_{\rho+1}&\mbox{ $p=\rho-4$, $q=4$}\\
\cdots\\
t_0+t_1+t_2+\cdots+t_{\rho-k-2}&\mbox{ $=$ }&t_{\rho-1}+t_{\rho}+t_{\rho+1}+\cdots+t_{2\rho-k-3}&\mbox{ $p=\rho-k$, $q=k$}\\
\cdots\\
t_0+t_1+t_2+\cdots+t_{\rho-3}&\mbox{ $=$ }&t_{\rho-1}+t_{\rho}+t_{\rho+1}+\cdots+t_{2\rho-4}&\mbox{ $p=1$, $q=\rho-1$}
\end{array}
\]
}
We then have, via (\ref{eqn:appx7}) and the first case of (\ref{eqn:appx8}) (that is, $p=s-2\rho+2$, $q=s-\rho+2$), a sequence of $s-2\rho+3$ identities:
{
  \medmuskip=1mu
  \thinmuskip=1mu
  \thickmuskip=1mu
\[
\begin{array}{lclr}
t_0+t_1+\cdots+t_{\rho-2}&\mbox{ $=$ }&t_{\rho-1}+t_\rho+t_{\rho+1}+\cdots+t_{2\rho-3}&\mbox{ $q=\rho$}\\
t_0+t_1+\cdots+t_{\rho-2}&\mbox{ $=$ }&t_{\rho}+t_{\rho+1}+t_{\rho+2}+\cdots+t_{2\rho-3}+t_{2\rho-2}&\mbox{ $q=\rho+1$}\\
t_0+t_1+\cdots+t_{\rho-2}&\mbox{ $=$ }&t_{\rho+1}+t_{\rho+2}+\cdots+t_{2\rho-2}+t_{2\rho-1}&\mbox{ $q=\rho+2$}\\
\cdots\\
t_0+t_1+\cdots+t_{\rho-2}&\mbox{ $=$ }&t_{\rho+k-1}+t_{\rho+k}+\cdots+t_{2\rho+k-4}+t_{2\rho+k-3}&\mbox{ $q=\rho+k$}\\
\cdots\\
t_0+t_1+\cdots+t_{\rho-2}&\mbox{ $=$ }&t_{s-\rho}+t_{s-\rho+1}+\cdots+t_{s-3}+t_{s-2}&\mbox{ $q=s-\rho+1$}\\
t_0+t_1+\cdots+t_{\rho-2}&\mbox{ $=$ }&t_{s-\rho+1}+t_{s-\rho+2}+\cdots+t_{s-2}+t_{s-1}&\mbox{ $q=s-\rho+2$}
\end{array}
\]
}
Finally, from the remaining $\rho-2$ cases of (\ref{eqn:appx8}), we obtain the identities:
{
  \medmuskip=1mu
  \thinmuskip=1mu
  \thickmuskip=1mu
\[
\begin{array}{lclr}
t_1+t_2+t_3+\cdots+t_{\rho-2}&\mbox{ $=$ }&t_{s-\rho+2}+t_{s-\rho+3}+\cdots+t_{s-2}+t_{s-1}&\mbox{ $q=s-\rho+3$}\\
t_2+t_3+\cdots+t_{\rho-2}&\mbox{ $=$ }&t_{s-\rho+3}+\cdots+t_{s-2}+t_{s-1}&\mbox{ $q=s-\rho+4$}\\
\cdots\\
t_k+\cdots+t_{\rho-2}&\mbox{ $=$ }&t_{s-\rho+k+1}+\cdots+t_{s-2}+t_{s-1}&\mbox{ $q=s-\rho+k+2$}\\
\cdots\\
t_{\rho-3}+t_{\rho-2}&\mbox{ $=$ }&t_{s-2}+t_{s-1}&\mbox{ $q=s-1$}\\
t_{\rho-2}&\mbox{ $=$ }&t_{s-1}&\mbox{ $q=s$}
\end{array}
\]
}
It is clear, from these patterns, that the following equalities must hold in order for $\underline{t}$ to
lead to $C(i,\underline{t})=C(i+\varphi_\rho,\underline{t})$ after $n$-correction.
From the first set of $\rho-2$ identities together with that for $q=\rho$, we obtain:
\begin{equation}\label{eqn:appx9}
\begin{array}{lcl}
t_0&\mbox{$=$}&t_{\rho-1}\\
t_1&\mbox{$=$}&t_\rho\\
\cdots\\
t_{\rho-k-2}&\mbox{$=$}&t_{2\rho-k-3}\\
\cdots\\
t_{\rho-3}&\mbox{$=$}&t_{2\rho-4}\\
t_{\rho-2}&\mbox{$=$}&t_{2\rho-3}
\end{array}
\end{equation}
For the identities which must hold when $\rho\leq q\leq s-\rho+2$, the left-hand of side of these is always
$\sum_{k=0}^{\rho-2}~t_k$ while the right-hand side of the identity for $q$ has exactly two terms in common
with the right-hand side of $q-1$ and two in common with $q+1$. In consequence we further deduce:
\begin{equation}\label{eqn:appx10}
\begin{array}{lcl}
t_{\rho-1}&\mbox{$=$}&t_{2\rho-2}\\
t_\rho&\mbox{$=$}&t_{2\rho-1}\\
\cdots\\
t_{\rho+k-1}&\mbox{$=$}&t_{2\rho+k-2}\\
\cdots\\
t_{s-\rho}&\mbox{$=$}&t_{s-1}
\end{array}
\end{equation}
Finally, in a similar manner to the identities in the first collection we deduce from the final set:
\begin{equation}\label{eqn:appx11}
\begin{array}{lcl}
t_{s-1}&\mbox{$=$}&t_{\rho-2}\\
t_{s-2}&\mbox{$=$}&t_{\rho-3}\\
\cdots\\
t_{s-\rho+k+1}&\mbox{$=$}&t_k\\
\cdots\\
t_1&\mbox{$=$}&t_{s-\rho+2}\\
t_0&\mbox{$=$}&t_{s-\rho+1}
\end{array}
\end{equation}

Combining (\ref{eqn:appx9}) and (\ref{eqn:appx10}), we have $t_j = t_{j+(\rho-1)}, \forall 0 \le j \le s-\rho$. 
Together with (\ref{eqn:appx11}), which gives $t_{s-\rho+1+k} = t_k, \forall 0 \le k \le \rho-2$, these three families establish that $t_j = t_{j+(\rho-1) \bmod s}, \forall 0\le j\le s-1$; i.e.\ the sequence of values comprising $\underline{t}$ repeats under a shift of its own indices by $(\rho-1)$ positions modulo $s$.  
Since any sequence of length $s$ that is invariant under cyclic shift by $d$ positions has period dividing $\gcd(s,d)$, it follows that $\pi(\underline{t})$ divides $\gcd(s,\rho-1) = g$, and hence $\pi(t) \le g$.

To illustrate this, the effect of combining (\ref{eqn:appx9}) - (\ref{eqn:appx11}) on $\tuple{t_0,\ldots,t_{s-1}}$ can be seen in Table~\ref{table:rev-period}.
For each $j$, 
\begin{eqnarray*}
t_{(\rho-1)-j}&\mbox{ $=$ }&t_{2(\rho-1)-j}\mbox{ $~$ $~$ from (\ref{eqn:appx9})}\\
t_{(\rho-1)+j}&\mbox{ $=$ }&t_{2(\rho-1)+j}\mbox{ $~$ $~$ from (\ref{eqn:appx10})}\\
t_{j}&\mbox{ $=$ }&t_{s-(\rho-1)+j}\mbox{$~$ $~$ from (\ref{eqn:appx11})}
\end{eqnarray*}

Now suppose we write $s$ as $s=g\times \mu$ where $g=gcd(s,\rho-1)$: note that since this allows $g=1$ the value of $g$ is always
well defined for any $s$ and $2\leq \rho\leq s$.
In this case the conditions expressed in (\ref{eqn:appx9})--(\ref{eqn:appx11}) indicate, from $s=g\times \mu$ and $(\rho-1)=g\times v$
that $\underline{t}$ conforms to the behaviour in Table~\ref{table:rev-period}.
That this is the behaviour imposed, follows from the conditions arising via (\ref{eqn:appx11}) to the effect
$t_{j}=t_{s-(\rho-1)+j}$, which is equivalent to the identity, 
$t_j=t_{g\mu-gv+j}=t_{g(\mu-v)+j}$,
indicating that values repeat in blocks of size $g$.
This configuration, however, indicates that 
$\underline{t}$ for which 
$C(i,\underline{t})=C(i+\varphi_\rho,\underline{t})$ after $n$-correction must
have the form:
$\tuple{t_0,t_1,\ldots,t_{g-1},t_0,t_1,\ldots,t_{g-1},\ldots,t_0,t_1,\ldots,t_{g-1},t_0,t_1,\ldots,t_{g-1}}$

Now in order to complete the proof it remains to demonstrate that
$\varphi_\rho=l\times(n\pi(\underline{t})/s)$ for some positive integer $l$.
Since $\pi(\underline{t})\leq g$ and we have already shown that $\underline{t}$ repeats in blocks of length $g$, should $\pi(\underline{t})<g$ then we must have $g$ an exact multiple ($>1$) of $\pi(\underline{t})$.

\begin{table}[tb]
\caption{Values of $t_{i(\rho-1)+j}$ implied by (\ref{eqn:appx9})--(\ref{eqn:appx11}), $s=\mu\times g$, $(\rho-1)=v\times g$}\label{table:rev-period}
\begin{tabular}{|c|c|c|c|c|c|c|c||c||c|c|c|}\hline
         & $0$   & $1$   & $\cdots$ & $j$   & $\cdots$ & $(g-2)$     & $(g-1)$    & $\cdots$ & $(v-1)g$ & $\cdots$ & $vg-1$ \\\hline
$0$      & $t_0$ & $t_1$ & $\cdots$ & $t_j$ & $\cdots$ & $t_{(g-2)}$ & $t_{(g-1)}$& $\cdots$ & $t_0$    & $\cdots$ & $t_{(g-1)}$ \\\hline
$1$      & $t_0$ & $t_1$ & $\cdots$ & $t_j$ & $\cdots$ & $t_{(g-2)}$ & $t_{(g-1)}$& $\cdots$ & $t_0$    & $\cdots$ & $t_{(g-1)}$ \\\hline
$\cdots$ & $t_0$ & $t_1$ & $\cdots$ & $t_j$ & $\cdots$ & $t_{(g-2)}$ & $t_{(g-1)}$& $\cdots$ & $t_0$    & $\cdots$ & $t_{(g-1)}$ \\\hline
$i$      & $t_0$ & $t_1$ & $\cdots$ & $t_j$ & $\cdots$ & $t_{(g-2)}$ & $t_{(g-1)}$& $\cdots$ & $t_0$    & $\cdots$ & $t_{(g-1)}$ \\\hline
$\cdots$ & $t_0$ & $t_1$ & $\cdots$ & $t_j$ & $\cdots$ & $t_{(g-2)}$ & $t_{(g-1)}$& $\cdots$ & $t_0$    & $\cdots$ & $t_{(g-1)}$ \\\hline
$\mu-2$    & $t_0$ & $t_1$ & $\cdots$ & $t_j$ & $\cdots$ & $t_{(g-2)}$ & $t_{(g-1)}$& $\cdots$ & $t_0$    & $\cdots$ & $t_{(g-1)}$ \\\hline
$\mu-1$      & $t_0$ & $t_1$ & $\cdots$ & $t_j$ & $\cdots$ & $t_{(g-2)}$ & $t_{(g-1)}$& $\cdots$ & $t_0$    & $\cdots$ & $t_{(g-1)}$ \\\hline
\end{tabular}
\end{table}

\noindent
We have from our analysis above:
\begin{eqnarray*}
\varphi_\rho&\mbox{ $=$ }&\sum_{k=0}^{\rho-2}~t_k~+~\rho-1\\
&\mbox{ $=$ }&\frac{\rho-1}{g}\sum_{k=0}^{g-1}~t_k~~+~~\rho-1\\
&\mbox{ $=$ }&\frac{\rho-1}{g}\left({\sum_{k=0}^{g-1}~t_k~+~g}\right)\\
&\mbox{ $=$ }&\left({\frac{\rho-1}{g}}\right)\varphi_{g+1}
\end{eqnarray*}
In addition, however:
\[
n-s~=~\sum_{k=0}^{s-1}~t_k~=~\frac{s}{g}\sum_{k=0}^{g-1}~t_k
=\frac{s}{g}(\varphi_{g+1}-g)
\]
So that:
\[
\varphi_{g+1}=\frac{g}{s}(n-s)+g=\frac{gn}{s}
\]
Thereby giving $\varphi_\rho$ as:
\[
\varphi_{\rho}~~=~~\left({\frac{\rho-1}{g}}\right)\frac{gn}{s}~~=~~\frac{(\rho-1)n}{s}~~=~~l\times\left({\frac{n\pi(\underline{t})}{s}}\right)
\]
for some choice of $l$ since $(\rho-1)=v\times g=l\times\pi(\underline{t})$.
Setting the theorem's multiplier $r=l=(\rho-1)/\pi(\underline{t})$ gives the required conclusion, and $j\le n$ immediately yields the upper bound $r \le (n-i)s / (n\pi(\underline{t}))$:
\[
C(i,\underline{t})=C(j,\underline{t})~~\Rightarrow~~j=i+r\times\left({\frac{n\pi(\underline{t})}{s}}\right)
\]
\end{proof}




\end{appendices}

\bibliography{necklace26}

\end{document}